\documentclass[a4paper,11pt]{article}
\usepackage{jheppub}

\usepackage{amsthm}
\newtheorem{theorem}{Theorem}
\newtheorem{lemma}[theorem]{Lemma}
\newtheorem{corollary}[theorem]{Corollary}
\theoremstyle{definition}
\newtheorem{definition}{Definition}[section]

\usepackage{enumerate}
\usepackage[utf8]{inputenc}
\usepackage{amssymb,amsmath,amsbsy,mathtools}
\usepackage{braket}
\usepackage{graphicx}
\usepackage{xcolor}
\usepackage{hyperref}
\usepackage{natbib}
\usepackage{bm}
\usepackage{soul}

\newcommand{\mA}{\mathcal{A}}
\newcommand{\mB}{\mathcal{B}}

\newcommand{\mE}{\mathcal{E}}

\newcommand{\mH}{\mathcal{H}}

\newcommand{\mK}{\mathcal{K}}

\newcommand{\mM}{\mathcal{M}}

\newcommand{\mO}{\mathcal{O}}

\newcommand{\mR}{\mathcal{R}}
\newcommand{\mS}{\mathcal{S}}

\newcommand{\mbC}{\mathbb{C}}

\newcommand{\p}{\partial}
\newcommand{\lb}{\left(}
\newcommand{\ep}{\epsilon}
\newcommand{\rb}{\right)}
\newcommand{\nn}{\nonumber}

\newcommand{\vol}{\text{vol}}

\begin{document}
\makeatletter 
\gdef\@fpheader{}
\makeatother


\title{Local Poincar\'e Algebra from Quantum Chaos}
\author{Shoy Ouseph, Keiichiro Furuya,  Nima Lashkari, Kwing Lam Leung, Mudassir Moosa}
\affiliation{Department of Physics and Astronomy, Purdue University, West Lafayette, IN 47907, USA}
\emailAdd{souseph@purdue.edu}
\emailAdd{kfuruya@purdue.edu}
\emailAdd{nima@purdue.edu}
\emailAdd{leung60@purdue.edu}
\emailAdd{mudassir@purdue.edu}

\abstract{The local two-dimensional Poincar\'e algebra near the horizon of an eternal AdS black hole, or in proximity to any bifurcate Killing horizon, is generated by the Killing flow and outward null translations on the horizon. In holography, this local Poincar\'e algebra is reflected as a pair of unitary flows in the boundary Hilbert space whose generators under modular flow grow and decay exponentially with a maximal Lyapunov exponent. This is a universal feature of many geometric vacua of quantum gravity. To explain this universality, we show that a two-dimensional Poincar\'e algebra emerges in any quantum system that has von Neumann subalgebras associated with half-infinite modular time intervals {\it (modular future and past subalgebras)} in a limit analogous to the near-horizon limit. In ergodic theory, quantum dynamical systems with future or past algebras are called quantum K-systems. The surprising statement is that modular K-systems are always maximally chaotic.

Interacting quantum systems in the thermodynamic limit and large $N$ theories above the Hawking-Page phase transition are examples of physical theories with future/past subalgebras. We prove that the existence of (modular) future/past von Neumann subalgebras also implies a second law of (modular) thermodynamics and the exponential decay of (modular) correlators. We generalize our results from the modular flow to any dynamical flow with a positive generator and interpret the positivity condition as quantum detailed balance.
}

\keywords{Future/Past subalgebras, Emergent Poincar\'e group, Second Law, Quantum Dynamical Systems, Quantum Ergodicity, Quantum K-systems, Quantum Anosov systems}

\maketitle


\section{Introduction}

A theory of quantum gravity in the small $G_N$ limit admits a large set of geometric states that host fluctuating quantum fields. To every smooth geometry, we associate a perturbative semiclassical Hilbert space of states of curved spacetime quantum field theory (QFT). The smoothness of the geometry implies that quantum fields in the vicinity of any point transform under a local Poincar\'e algebra corresponding to the local Minkowski space. The emergence of these local Poincar\'e groups is a curious universal feature of many vacua in any theory of quantum gravity. In this work, we identify an origin for the emergence of this universality based on the ergodic properties of modular flows in observable algebras of quantum gravity.

\begin{figure}[b]
    \centering
    \includegraphics[width=0.8\linewidth]{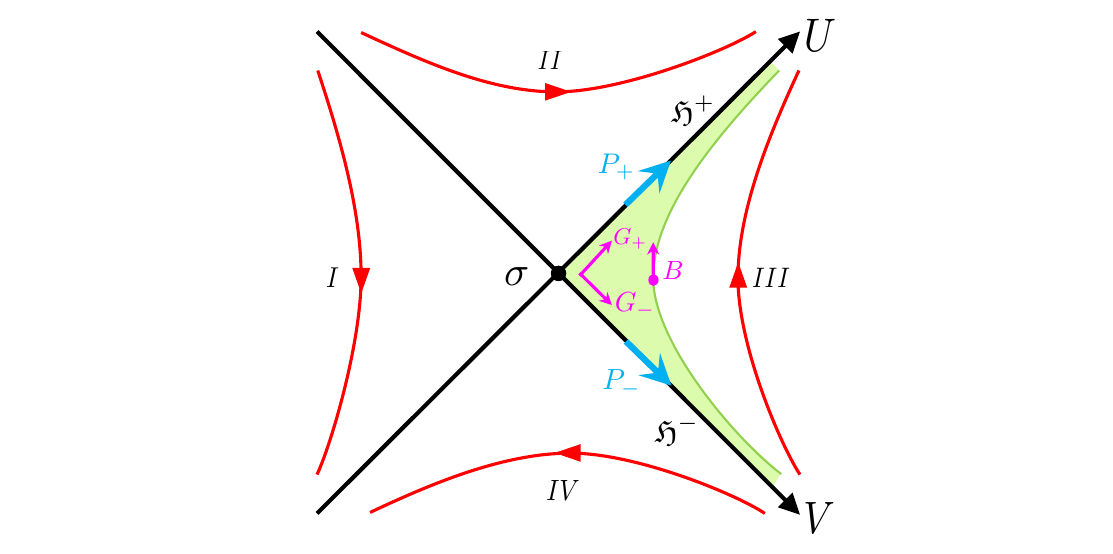}
    \caption{\small{The Killing vector $B$ of a spacetime with a bifurcate Killing horizon splits it into four regions, similar to boosts in Minkowski space. $\sigma$ is the bifurcation surface and the null translations $P_\pm$ on the future and past horizons $\mathfrak{H}^\pm$ are isometries of these surfaces. They can be extended to $G_\pm$ in the near horizon regions.}}
    \label{fig:rindlerboost}
\end{figure}

In a general QFT in curved spacetime, there is no globally time-like Killing vector; therefore, it is not clear how to impose a condition similar to the positivity of the Hamiltonian, and the analytic properties of vacuum correlators in complex coordinates. To overcome this, we focus on QFT in spacetimes with a bifurcate Killing horizon in the Hartle-Hawking state where the KMS condition provides analytic properties of correlators similar to the positivity of the Hamiltonian \cite{kay1991theorems}. Similar to boost in Rindler space, a spacetime with a bifurcate Killing horizon has a global Killing vector $B$ which splits it into four regions; see Figure \ref{fig:rindlerboost}. The future (past) Killing horizon $\mathfrak{H}^+$ ($\mathfrak{H}^-$) is a null surface with enhanced symmetries. The Killing vector $B$ and null translation $P_+$ ($P_-$) generate isometries of these null surfaces, respectively. On these surfaces we have the Lie groups $s_0\to e^{t}s_0+s$ ($s_0\to e^{-t}s_0+s$) and Lie algebraic relations\footnote{For brevity, we will slightly abuse notation and denote the generator of the flow along a vector field $B$ by $B$.} 
\begin{eqnarray}\label{algebrarel1}
    &&\mathfrak{H}^+:\qquad e^{itB} e^{i s P_+}e^{-itB}=e^{i s e^{ t}P_+},\qquad [P_+, B]=i P_+\nn\\
        &&\mathfrak{H}^-:\qquad e^{itB} e^{i s P_-}e^{-itB}=e^{i s e^{- t}P_-},\qquad  [P_-,B]=-i P_-\ .
\end{eqnarray}
For a general space-time with a bifurcate Killing horizon, the translations $P_\pm$ are isometries only when restricted to the future and past Killing horizons. In the vicinity of the Killing horizon $\mathfrak{H}^+$ ($\mathfrak{H}^-$), we extend the null translations $P_+$ ($P_-)$ to the vector fields $G_+$ ($G_-$) in normal Riemann coordinates. From the point of view of a perturbation whose proper distance to the bifurcate surface $\sigma$ is much smaller than the curvature length at the horizon these operators commute, i.e. $[G_+,G_-]\simeq 0$. 
This commutation relation together with (\ref{algebrarel1}) gives a local two-dimensional Poincar\'e algebra. 
As we move further away from the horizon, the commutator $[G_+,G_-]$ is decided by the curvature of the near-horizon geometry \cite{de2020holographic}.\footnote{If in the next order, they continue to commute, the near-horizon geometry has Poincar\'e symmetry which is naturally associated with Rindler space. 
If, instead, they satisfy $[G_+,G_-]=-2iB$ at this order, we have the symmetry group $PSL(2,\mathbb{R})$ which is naturally associated with two-dimensional Anti-de Sitter space as the near-horizon geometry \cite{lin2019symmetries}.}

\begin{figure}[b]
    \centering
    \includegraphics[width=0.85\linewidth]{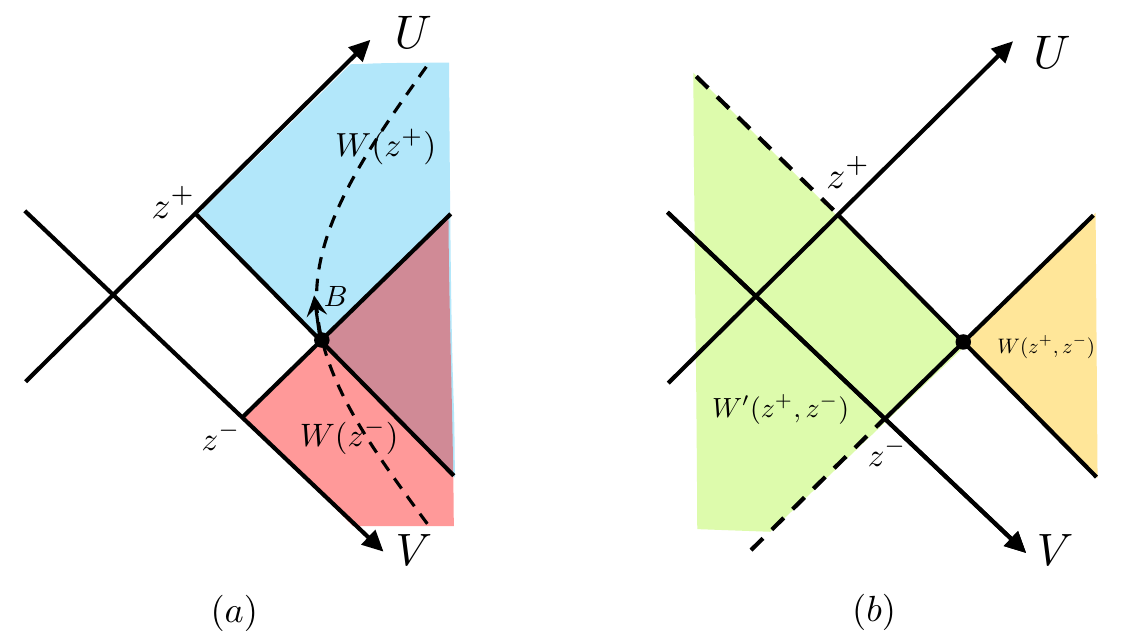}
    \caption{\small{The wedge algebras in a spacetime with bifurcate Killing horizons. (a) $W(z^+)$ is the von Neumann subalgebra of all observables the observer (dashed line) has access to from a particular moment in time (black dot) until eternity forms,  whereas $W(z^-)$ is the von Neumann subalgebra of all observables the observer had access to since past infinity until now. (b)  Given a point $(z^{+},z^{-})$, we define a right wedge $W(z^{+},z^{-}) = (U>z^{+},V>z^{-})$ and a left wedge $W'(z^{+},z^{-}) = (U<z^{+},V<z^{-})$. }}
    \label{fig:wedgealgebras}
\end{figure}

In a theory of quantum gravity, we associate von Neumann algebras with regions of spacetime that are diffeomorphism-invariant, such as the right wedge $W$ (region III in Figure \ref{fig:rindlerboost}) in spacetime with a bifurcate Killing horizon.
The geometric symmetry transformations of the bulk are reflected as unitary flows in the observable algebras of quantum gravity. In particular, in the Hilbert space of quantum gravity states dual to the excitations of the Hartle-Hawking vacuum of a spacetime with a bifurcate Killing horizon, there must be an emergent Poincar\'e symmetry in a regime analogous to the limit that localizes excitations near the bifurcate Killing horizon. The first clue to the origin of these symmetry relations on the boundary comes from the observation that the Killing field $B$ generates the modular flow of the observable algebra of the right wedge: $\Delta_W^{it}=e^{-2\pi i B t}$ (see Theorem \ref{thm:summerscurved}). The modular flow of an observable algebra is an example of a state-preserving quantum dynamical system that we call {\it modular dynamical system}. Consider an observer at a point in the right wedge $W$ who moves along the Killing vector $B$ (see Figure \ref{fig:wedgealgebras}).  They perceive the modular flow of the right wedge $W$ as time evolution. The set of all observables they have access to from a particular moment in time until eternity forms a von Neumann subalgebra of all observables of the right wedge that we call the {\it modular future subalgebra}, e.g. $W(z^+)$ in Figure \ref{fig:wedgealgebras}. Similarly, the set of all observables they had access to from past infinity until now forms the {\it modular past subalgebras}, e.g. $W(z^-)$ in Figure \ref{fig:wedgealgebras}. Since the modular flow of $W$ is the Killing flow of $B$, we find the {\it half-sided modular inclusion relations}
\begin{eqnarray}\label{Halfsidedrel}
&&\forall t>0:\qquad \Delta_W^{\mp it}\mA(W(z^\pm))\Delta_W^{\pm it}\subset \mA(W(z^\pm))\ .
\end{eqnarray}
In this work, we show that, in an arbitrary quantum system with a von Neumann algebra of observables $\mR$ (analog of the algebra of the right wedge $\mA(W)$), the existence of modular future and past subalgebras $\mA^+$ and $\mA^-$ (analogs of $\mA(W(z^+))$ and $\mA(W(z^-))$), respectively, is sufficient for a ``local" Poincar\'e algebra to emerge in a certain scaling regime analogous to the near-horizon limit.
There are two steps to our argument. The first step relies on a key result of quantum ergodic theory (The Half-sided Modular Inclusion Theorem \ref{thm:HSMI}) which says that the modular past and future subalgebra relations $\mA^\pm\subset \mR$:
\begin{eqnarray}
    \forall t>0: \qquad \Delta_\mR^{\mp it}\mA^\pm\Delta_\mR^{\pm it}\subset \mA^\pm
\end{eqnarray}
imply that the operators 
\begin{eqnarray}
    &&\pm G_\pm:=K_\mR-K_{\mA^\pm}\nn\\
    &&K=-\frac{1}{2\pi}\log \Delta
\end{eqnarray}
are positive, generate a symmetry $G_\pm \ket{\Omega}=0$, and satisfy the algebraic relations
\begin{eqnarray}\label{algebrarel}
    \Delta_\mR^{-it} e^{i s G_\pm}\Delta_\mR^{it}=e^{i  e^{\pm 2\pi t}s G_\pm}\ .
\end{eqnarray}
In the example of the quantum gravity algebra dual to the right wedge of the spacetime with bifurcate Killing horizon, the generators $G_\pm$ do not have a nice geometric dual in terms of local spacetime transformation, except for their restrictions to the horizons $\mathfrak{H}^\pm$ which are proportional to null translations $P_\pm$.

In general, there is no reason for $[G_+,G_-]$ to be small unless we introduce a new scaling parameter.
The second part of our argument is a scaling limit that mimics the near-horizon limit by defining the algebras and generators 
\begin{eqnarray}
    &&\mA^\pm( s):=\Delta_\mR^{- is}\mA^\pm \Delta_\mR^{ is}\nn\\
    &&\pm G_\pm(s):=K_\mR-K_{\mA^\pm(s)}\ .
\end{eqnarray}
In Theorem \ref{newthmapprox}, we show that $[G_+(-s),G_-(s)]=O(e^{-4\pi s})$, therefore in the limit of large $s$ we have an emergent Poincar\'e algebra. This limit is the analog of zooming in near the horizon to proper distances much smaller than the local curvature length scale.

As we review in Section \ref{sec:ergodicity}, the relations in (\ref{algebrarel}) constitute a special example of a class of quantum ergodic systems called {\it quantum Anosov systems}. A quantum Anosov system is a quantum dynamical system with unitary ``time-evolution" $e^{i Kt}$ that admits other unitary flows generated by self-adjoint operators $G_\alpha$ such that 
\begin{eqnarray}\label{algebraAnosov}
  e^{i K t} e^{i s G_\alpha}e^{-i Kt}=e^{i s e^{-\lambda_\alpha t}G_\alpha}
\end{eqnarray}
and the constants $\lambda_\alpha$ which are positive (negative) for decaying (growing) modes are often called the {\it Lyapunov} coefficients.
It has been argued when the dynamics is modular flow ($e^{i Kt}=\Delta^{-it}$) there is a universal bound on Lyapunov exponents $|\lambda|\leq 2\pi$ \cite{de2020holographic,faulkner2019modular}. Since the algebraic relations in (\ref{algebrarel}) saturate this bound following we will refer to the algebraic relations in (\ref{algebrarel}) as {\it maximal modular chaos}.
Following \cite{de2020holographic}, we will call $G_\pm$ the {\it modular scrambling modes}.

The AdS/CFT correspondence provides a concrete realization of quantum gravity in asymptotically Anti-de Sitter (AdS) spacetimes. A large $N^2\sim 1/G_N$ theory living on the boundary of AdS is dual to curved spacetime QFT in the bulk. In the strict $N\to \infty$ limit, the boundary is described by Generalized Free Fields (GFF) \cite{aharony2000large}. A GFF algebra is fixed by two choices of functions corresponding to the two-point correlation function and the commutator of the fundamental field \cite{greenberg1961generalized}.\footnote{See \cite{furuya2023information} for a review.} In a KMS state (finite temperature state), the two functions are related and we only need to fix one function, namely the spectral density (the Fourier transform of the commutator).
In holographic GFF \cite{dutsch2003generalized}, assuming spherical symmetry, effectively, we reduce the bulk to $1+1$-dimensions and the GFF to a $0+1$ dimensional theory.\footnote{Finite temperature GFF in $0+1$-dimensions is a collection of simple harmonic oscillators of frequency $\omega$ associated with every non-vanishing spectral density $\rho(\omega^2)$ \cite{furuya2023information}.} The relation between the bulk fields and boundary fields is explicit at the level of the one-particle Hilbert space. Given a boundary function $f(t)$ we can find the bulk solution $\tilde{f}(t,r)$ to the field equations of motion on geometry $g_{\mu\nu}$ that asymptotically relates to the boundary function $f(t)$. This is sometimes called the HKLL bulk reconstruction \cite{hamilton2006holographic,hamilton2006local}. The KMS state of holographic GFFs exhibits a special phase transition, called the Hawking-Page phase transition. Below the critical temperature, the bulk is the thermal state of Anti-de Sitter space, whereas above the critical temperature, the geometry is an AdS Schwarzschild black hole. If we think in terms of the thermofield double (the canonical purification of the thermal state), below the transition point, the dual geometry is two disjoint copies of thermal AdS, whereas above the transition, the dual geometry is two black holes sewn together to form a two-sided wormhole (eternal AdS black hole) \cite{maldacena2003eternal}. Recently, \cite{leutheusser2021causal,leutheusser2021emergent} showed that above the Hawking-Page phase transition, there are boundary GFF von Neumann algebras associated with finite time intervals. In particular, there are future and past subalgebras associated with half-infinite time intervals $(t_1,\infty)$ and $(-\infty,t_2)$. In the bulk, the Hawking-Page phase transition point is associated with the emergence of the following physical phenomena:
\begin{enumerate}[(1)]
    \item {\bf Strong Mixing:} Below the transition point, all correlators are almost periodic functions of time, whereas above the transition, all connected correlators decay to zero. 

    \item {\bf Second Law:} Below the transition point, the bulk dynamics is reversible, whereas, above the transition point, the bulk dynamics becomes irreversible because excitations can fall behind the horizon. The emergence of a monotonically increasing coarse-grained entropy (second law) is a manifestation of this irreversibility.

    \item {\bf Exponential decay:} Above the transition point, not only do all connected correlators decay but also a large class of operators decay exponentially fast with fixed exponents called quasi-normal modes.

    \item {\bf Near-horizon Poincar\'e algebra:} Above the transition point, a smooth horizon is formed in the bulk, and there is an emergent local two-dimensional Poincar\'e group in the vicinity of the bifurcation surface of the eternal black hole.
    
\end{enumerate}Properties (1), (2) and (3) are aspects of the {\it black hole information loss problem}. However, they are not specific to black holes. The physics above the transition point occurs in most, if not all, thermalizing quantum systems in the strict thermodynamic limit of infinite fine-grained entropy. These properties are not independent. In fact, as we review in section \ref{sec:ergodicity}, they follow an ergodic hierarchy
\begin{eqnarray}\label{hierarch1}
    \text{Exponential Decay}\Rightarrow \text{Second Law}\Rightarrow \text{Strong Mixing}\ .
\end{eqnarray}
Each of them corresponds to an important ergodic class: Exponential Decay ({\it Anosov systems}), Second Law (Kolmogorov systems or, in short {\it K-systems}) and {\it Strong Mixing} systems. 

There is an intuitive physical justification as to why we expect properties (1)-(3) to emerge in the thermodynamic limit. As we argued in \cite{furuya2023information}, strong mixing follows from the physical picture that in a generic thermalizing system all non-conserved excitations should decay. In other words, at late times, the system forgets about its current state, as is manifest by the decay of all connected two-point correlators. The second law is a consequence of the physical requirement that, at late times, not only the system forgets about its current state, but it also becomes independent of the {\it entire past} of the system. 
Independence from the entire past (future) is what motivated Kolmogorov to introduce a class of quantum ergodic systems called K-systems which can be characterized by the existence of past (future) subalgebras.

Property (4), namely an emergent two-dimensional Poincar\'e group, is often considered to be purely gravitational. 
Our main result is that this property also emerges in modular quantum K-system. We show that all the four properties above follow from a simple feature that is known to hold at finite temperature state of many thermodynamic systems, namely the existence of the {\it future and past subalgebras}: 
\begin{eqnarray}\label{result1}
        &&\text{{\bf Future/past subalgebras}}\Rightarrow \\
        &&\text{Local Poincar\'e Group}\Rightarrow\text{Exponential Decay}\Rightarrow \text{Second Law}\Rightarrow \text{Strong Mixing}.\nn
\end{eqnarray}
The future subalgebra is a {\it proper} subalgebra of observables generated by all perturbations (or measurements) one can make from now until future infinity in time. For instance, in a type I algebra, the set of future operators generates all operators; therefore, we say that type I algebras do not admit future subalgebras.\footnote{For a review of the classification of observables algebras of quantum mechanics (von Neumann algebras) see appendix A of \cite{furuya2023information}} Similarly, the past subalgebra is the {\it proper} subalgebra of all perturbations (or measurements) that could have been made from past infinity in time until now.

Our discussion generalizes from finite temperature time-evolution in two major ways: 1) Modular flow of a general state and 2) General dynamics with a positive generator.
The results continue to hold in arbitrary out-of-equilibrium states, with Hamiltonian replaced with modular Hamiltonian. We show that the existence of {\it modular future and past subalgebras} implies the strong mixing of modular correlators, a modular second law, the exponential decay of modular correlators, and an emergent approximate modular Poincar\'e group in a certain limit similar to the one that localizes a particle near the horizon a black hole \cite{maldacena2016bound}.
We see in section \ref{sec:ergodicity}, our results also generalize to any dynamical unitary flow $e^{i Kt}$ that is state-preserving $K\ket{\Omega}=0$ and has a positive generator $K\geq 0$.

In a recent work \cite{furuya2023information}, we showed that strong (modular) mixing implies that the algebra is type III$_1$ in a state that has a trivial centralizer. This is compatible with the fact that, in an eternal black hole, the bulk dual of each boundary algebra is an AdS Rindler wedge. We proved that GFF algebra in a KMS state with a spectral density that is Lebesgue measurable is type III$_1$. Note that the bulk algebra is known to be the unique hyperfinite type III$_1$ algebra, and it is expected that the hyperfiniteness requires the spectral density to be a smooth and non-vanishing function of frequency, in accordance with the conjecture in \cite{leutheusser2021emergent}.
The half-sided modular inclusions play a key role in our work. They were also discussed in \cite{leutheusser2021emergent} in the context of the GFF algebra with the spectral density that matches the black hole two-point function. The authors pointed out the connections to the emergence of an arrow of time (second law) and calculated the action of the modular scrambling modes $G_\pm$ on the creation/annihilation operators of this GFF theory.

Our result in (\ref{result1}) applied to holography implies that in a theory of GFF at finite temperature (KMS state), the existence of future subalgebra implies a local Poincar\'e group near the horizon, the exponential decay of connected correlators of a dense set of observables, a second law of thermodynamics, and the decay of any connected correlator. However, as we remarked earlier, our results are more general, and apply to any quantum system which satisfies the following three key assumptions:
1) {\bf Symmetry:} $e^{i Kt}\ket{\Omega}=\ket{\Omega}$, 2) {\bf Positivity:} $K\geq 0$, 3) {\bf Future (Past) subalgebras:} there exists a subalgebra such that for all $t>0$ ($t<0$) we have $e^{iKt}\mA e^{-iK t}\subset \mA$.\footnote{In the language of quantum ergodic theory, the first assumption defines state-preserving dynamics, the first two assumptions define a quantum K-system and all three assumptions together are called half-sided translations, which leads to a {\it maximal modular chaos system} in the sense of maximal Lyapunov exponents.} We show that the three conditions above are equivalent to the existence of modular future and past subalgebras (see Theorem \ref{thm:HST}).

In section \ref{sec:CPT}, we review the well-known result of Bisognano-Wichmann and its generalizations to curved spacetime, showing that the modular Hamiltonian of half-space is boost, and there are future and past subalgebras with respect to the dynamical flows generated by null translations $P_\pm$. Poincar\'e algebra implies that the geometric deformations along null directions grow/decay exponentially in modular time. Our work relies on converse results that prove if there are a pair of commuting dynamical flows with positive generators $\pm G_\pm$ and an algebra $\mR$ that is a future (past) subalgebra with respect to $G_+$ ($-G_-$) respectively, the modular flow $\Delta_\mR^{it}$ and $e^{i s^\pm G_\pm}$ generate a two-dimensional Poincar\'e group. 

In Section \ref{sec:futurepast}, we give several examples of future and past algebras in physical systems and comment on the second law and the exponential decay of correlations. Section \ref{sec:ergodicity} reviews the key intuitions and ideas of quantum ergodic theory and the relevance of future and past subalgebras for maximal modular chaos, exponential decay of correlators, and the second law. There are many definitions of classical chaos and quantum chaos. We avoid the comparison of these definitions here and instead focus on the sharply defined ergodic hierarchy. Intuitively, one can think of classical Anosov systems as classically chaotic systems because they are ergodic and show exponential sensitivity to initial conditions, i.e. they have Lyapunov exponents. The surprise is that, in the quantum world, for modular dynamics, the ergodic hierarchy simplifies (see Corollary \ref{corr:modularhierarchy}). The weaker assumption of quantum K-systems (existence of future/past algebras) becomes equivalent to maximal quantum chaos. If we then take the quantum K-systems property as the definition of quantum chaos, we conclude that modular chaos is always maximal!  To explain what is special about modular dynamics, we generalize our results to arbitrary dynamical flows with a positive generator and identify the key physical insight that underlies the collapse of the ergodic hierarchy as {\it quantum detailed balance} (see Theorem \ref{thm:HST}).

This paper expands upon the contents of \cite{Gesteau:2023}, which were communicated to us by its author during a conference last July. We decided to coordinate our releases.

\section{CRT symmetry, Poincar\'e group and Modular Flow}\label{sec:CPT}

\subsection{From Poincar\'e to future/past subalgebras}\label{sec:CPTflat}

Consider QFT in Minkowski space $\mathbb{R}^{d,1}$ with the metric
\begin{eqnarray}
    &&ds^2=dx^+dx^-+dx^i dx_i\nn\\
    &&x^\pm=x^d\pm x^0\ .
\end{eqnarray}
To simplify our presentation, in most of this work, we ignore the perpendicular $x^i$ directions and focus on the two-dimensional space spanned by $x^\pm$. A pair $(z^+,z^-)$ defines a right wedge $W(z^+,z^-)=(x^+>z^+,x^->z^-,x^i)$ and a left wedge that is the causal complement $W(z^+,z^-)'=(x^+<z^+,x^-<z^-,x^i)$. We use the notation $W(z^+):=W(z^+,0)$ and $W(z^-)=W(0,z^-)$ and $W:=W(0,0)$; see Figure \ref{fig:wedgealgebras}(b) with $U$ and $V$ replaced with $x^+$ and $x^-$, respectively. We warn the reader that in our notation $P_-$ is a negative operator; see Figure \ref{fig:rindlerboost}.

QFT in Minkowski space is symmetric under a CRT transformation ($C$ is charge conjugation, $R$ is reflection, and $T$ is time-reversal).\footnote{If the dimension of spacetime $(d+1)$ is even, we can apply the transformation in each of the $d$-spatial directions to obtain the CPT symmetry $(x^\mu;q)\to (-x^\mu; -q)$. } In the Hilbert space, the CRT transformation corresponds to an anti-unitary operator $J$ that preserves the vacuum, i.e. $J\ket{\Omega}=\ket{\Omega}$, and satisfies the following commutation relation with null translation operators $e^{iz^\pm P_\pm}$ and fields:
\begin{eqnarray}
    &&J e^{i z^\pm P_\pm} J= e^{-iz^\pm P_\pm},\qquad J \phi(x^+,x^-) J= \phi^*(-x^+,-x^-)\ .
\end{eqnarray}
In Wightman's formulation of QFT, the analytic continuations of Wightman functions are invariant under the complexified connected Poincar\'e transformations. The complexified Lorentz group contains the RT transformation (reflection and time-reversal). Assuming that all fields transform as finite-dimensional representations of the Lorentz group and that every charged field is accompanied by its charge conjugate is enough to prove that CRT is a symmetry \cite{jost1957bemerkung}. Bisognano and Wichmann gave another proof of the CRT theorem by applying modular theory\footnote{See appendix A of \cite{furuya2023information} for a brief review of key definitions of modular theory.} to local algebras built out of Wightman fields:\footnote{In appendix \ref{app:Wightman}, we review the connections between Wightman's formulation of QFT and the algebraic approach to QFT (AQFT).}
\begin{theorem}[Bisognano-Wichmann CRT Theorem]
    Consider the von Neumann algebra of observables $\mA(W)$ associated with any wedge $W$ in the vacuum of QFT generated by smeared functions of a Wightman field. The modular flow is boost 
    \begin{eqnarray}
         &&\Delta_W^{it}:(x^+,x^-)\to (e^{2\pi t}x^+,e^{-2\pi t}x^-)\,
    \end{eqnarray}
    and the modular conjugation is a CRT transformation. The wedge algebra satisfies the Haag duality condition $\mA(W')=\mA'(W)$.
\end{theorem}
\begin{proof}
    See \cite{bisognano1975duality,bisognano1976duality} for a proof.
\end{proof}
Since the modular conjugation operator leaves the vacuum invariant, we obtain a proof of the CRT theorem in the vacuum. Borchers offered a proof of the above theorem using the axioms of algebraic QFT (AQFT) for the wedge algebras relaxing the assumption of Wightman fields. From the point of view of holography, the operator-algebraic approach has the advantage that if the boundary theory has algebras associated with time intervals they satisfy the same algebraic axioms as the algebras of the bulk QFT. 
\begin{theorem}[Borchers CRT Theorem]\label{Borchersthm}
    Consider the vacuum representation of Poincar\'e-covariant QFT that satisfies Haag's duality for wedges and every double cone. The modular flow of a wedge algebra is boost, and modular conjugation is the CRT transformation. As a result, 
\begin{enumerate}
    \item The positive generators\footnote{The positivity follows from the monotonicity of modular operator; see appendix A of \cite{furuya2023information}. Note that with our notation, $-P_->0$.} 
    \begin{eqnarray}\label{defGpm}
    \pm G_\pm=\frac{1}{z^\pm}\lb K_W-K_{W(z^\pm)}\rb=\pm P_\pm
    \end{eqnarray}
    are independent of $z^\pm$ and equivalent to null momenta $\pm P_\pm$.
    \item The unitary flows $\Delta_W^{it}$ and $e^{is G_\pm}$ 
    represent the Poincar\'e group:
    \begin{eqnarray}
    &&\Delta_{W}^{-it} e^{is G_\pm} \Delta_{W}^{it}=e^{ i  e^{\pm 2\pi t}sG_\pm}\label{Borch1}\\
    &&[G_+,G_-]=0\label{commutBorch}\ .
    \end{eqnarray}

    \item We obtain the relations
    \begin{eqnarray}
    &&\lim_{t\to \pm\infty}\Delta_W^{-it}\Delta_{W(z^\pm)}^{it}=e^{\pm i z^\pm G_\pm}\label{Borch2}\\
    &&J_W J_{W(z^\pm)}=e^{\pm 2i z^\pm G_\pm}\label{JJs}\ .
    \end{eqnarray}
\end{enumerate} 

\end{theorem}
\begin{proof}
    See \cite{borchers1992cpt,borchers2000pct,borchers2000revolutionizing} for a proof. For completeness, we derive the expressions for $G_{\pm}$ explicitly. The modular Hamiltonian $K_W=-\frac{1}{2\pi}\log\Delta_{W}$ is proportional to the generator of the boost transformation around $x^{\pm} = 0$:
\begin{eqnarray}
    K_{W} =  \int_{\Sigma} d\Sigma^{\mu} B^{\nu} T_{\mu\nu}  
\label{eq-boost-generator-t}
\end{eqnarray}
where $B = i(x^{+}\partial_{x^{+}} - x^{-}\partial_{x^{-}})$ is the boost Killing vector field and $T_{\mu\nu}$ is the energy-momentum tensor. This integral can be performed on any Cauchy slice, $\Sigma$. A convenient choice is to integrate on the null plane $x^{-}=0$. This yields 
\begin{eqnarray}
    K_W  =  i\int d^{d-2}x^{i} \int_{-\infty}^{\infty} dx^{+} x^{+} T_{++}(x^+,x^{-}=0,x^{i}) \ . 
\end{eqnarray}
The expression for $K_{W(z^+)}$ is the same as the expression in \eqref{eq-boost-generator-t} but the boost Killing vector field is now $B = i\lb(x^{+}-z^{+}) \partial_{x^{+}} - x^{-}\partial_{x^{-}}\rb \, $. With this change, we get
\begin{eqnarray}
    K_{W(z^{+})} = i \int d^{d-2}x^{i} \int_{-\infty}^{\infty} dx^{+} (x^{+}-z^{+}) T_{++}(x^+,x^{-}=0,x^{i} )\ .
\end{eqnarray}
It follows from the definition in \eqref{defGpm} that
\begin{eqnarray}
    G_{+} = i \int d^{d-2}x^{i} \, \int_{-\infty}^{\infty} dx^{+} T_{++}(x^{+},x^{-}=0,x^{i}) = P_{+} \ .
\end{eqnarray}
is indeed independent of $z^{+}$ and is equal to $P_{+}$. Repeating the argument for $G_{-}$ gives $-G_-=-P_-$.
The modular Hamiltonian $K_W$ is the operator in the Hilbert space that represents $B$. In this work, we often say informally that the modular Hamiltonian is a boost.
\end{proof}
Consider three different observers in two-dimensional Minkowski space: Two observers who perceive null translations by $P_+$ ($-P_-$) as time evolution, and a Rindler observer localized in the right wedge $W$ moving with finite acceleration. 
These setups provide the first examples of future and past (modular) subalgebras (see Figure \ref{fig:wedgealgebras} with $U$ and $V$ replaced with $x^+$ and $x^-$):
\begin{corollary}[Poincar\'e to Future/Past Subalgebras]\label{cor:future}
    Consider the vacuum representation of a local Poincar\'e-covariant QFT that satisfies Haag's duality for the algebra $\mA(W)$. Null translations $P_+$ and $-P_-$ are a pair of commuting positive charges that generate unitary flows satisfying
    \begin{eqnarray}
        &&\forall z^\pm>0,\qquad \mA(W(z^\pm))=e^{i z^\pm P_\pm}\mA(W)e^{-i z^\pm P_\pm}\subset \mA(W)\ .
    \end{eqnarray}
    When these properties hold we say $W$ is a future (past) subalgebra of $B(\mH)$ with respect to the unitary flow of $P_+$ ($-P_-$), respectively.\footnote{See Lemma \ref{lemma:equivfuturealgebra} for other equivalent ways of defining future and past subalgebras.}
\end{corollary}

\begin{corollary}[Poincar\'e to Modular Future/Past Subalgebras]\label{cor:futuremod}
    Consider the vacuum representation of a local Poincar\'e-covariant QFT. Then,
\begin{enumerate}
    \item With respect to the modular flow of $W$ (boost), the algebras of $W(z^+)$ and $W(z^-)$ for $z^\pm>0$ are, respectively, modular future and past subalgebras; i.e. 
    \begin{eqnarray}
        &&\forall t>0\qquad \Delta_W^{- it}\mA(W(z^+))\Delta_W^{it}\subset \mA(W(z^+))\nn\\
        &&\forall t<0\qquad \Delta_W^{- it}\mA(W(z^-))\Delta^{it}_W\subset \mA(W(z^-))\ .
    \end{eqnarray}
    \item The commutation relation $[G_+,G_-]=0$ is equivalent to the modular conjugation relation  
    \begin{equation}\label{JJcommut}
        J_{W(z^-)}J_{W(z^+)}=J_W J_{W(z^+)} J_{W(z^-)}J_W\ .
    \end{equation}
\end{enumerate}
\end{corollary}
\begin{proof}
    It follows from Theorem \ref{Borchersthm} that $\Delta_W^{it}$, $\Delta_{W(z^\pm)}^{it}$ are boosts, and $J_{W(z^\pm)}$ are CRT transformations. The first statement follows from the boost transformations of $W(z^\pm)$. The second statement is equivalent to the fact that null translations commute:
\begin{eqnarray}\label{commutetrans}
    [e^{i z^-G_-},e^{i z^+G_+}]=0\ .
\end{eqnarray}
To see this, using $J_W J_{W(z^\pm)}=e^{\pm 2i z^\pm G_\pm}$ in (\ref{JJs}) we  repackage the commutation relation in (\ref{commutetrans}) as 
\begin{eqnarray}
    (J_{W(z^-)}J_{W})(J_{W} J_{W(z^+)})=(J_W J_{W(z^+)}) (J_{W(z^-)}J_W)\ .
\end{eqnarray}
\end{proof}

Figure \ref{fig:rindleranosov} gives a geometric picture for Corollary \ref{cor:futuremod}, representing Poincar\'e group relations in \eqref{Borch1} and \eqref{Borch2}, as modular dynamics by $\Delta_W^{it}$, and a pair of commuting exponentially growing (decaying) modes $G_+$ ($-G_-$) with exponent $\lambda=2\pi$. 

\begin{figure}[t]
    \centering
    \includegraphics[width=0.8\linewidth]{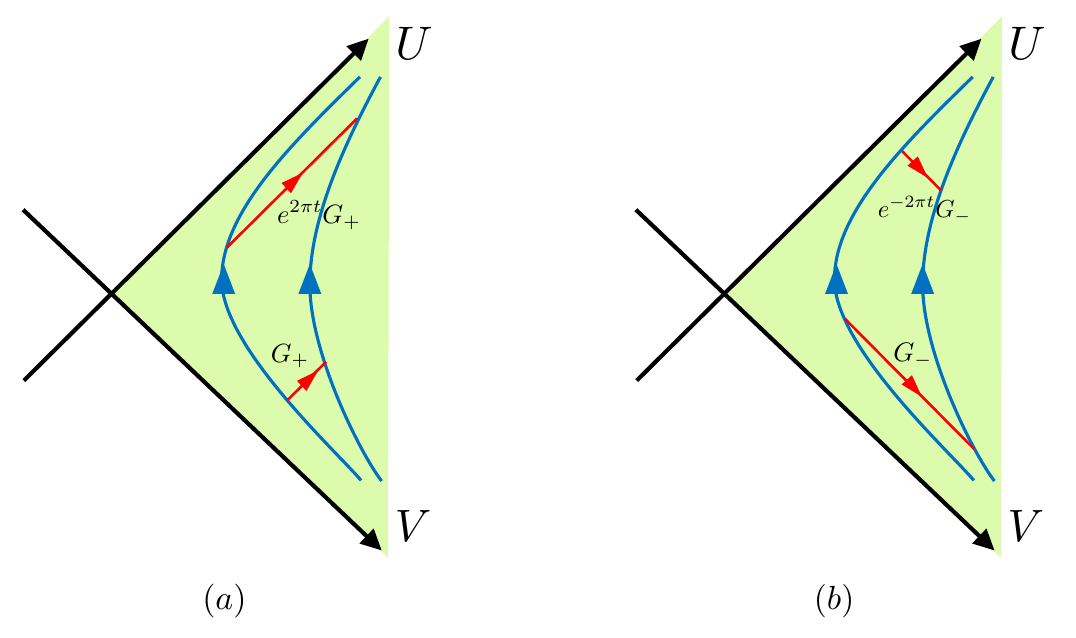}
    \caption{\small{The two-dimensional Poincar\'e algebra can be viewed as a modular Anosov system with maximal exponents $\lambda=2\pi$ (see section \ref{sec:ergodicity}). The blue lines denote the flow generated by the Killing vector field and the red lines denote commuting exponentially growing (decaying) modes $G_{+}$ ($G_{-}$). }}
    \label{fig:rindleranosov}
\end{figure}

\subsection{Modular flows in curved spacetime}\label{sec:CPTcurved}

The proofs of the CRT, Bisognano-Wichmann, and Borchers theorems in the Minkowski space crucially rely on the positivity of the Hamiltonian (the spectrum condition) because it implies the analyticity of the vacuum correlation functions in complex time domains. The lack of a globally timelike Killing vector in a general curved spacetime spoils these analytic properties. Globally hyperbolic spacetimes with a bifurcate Killing horizon are special in that null translations restricted to Killing horizons are isometries that allow us to quantize the subalgebra of QFT localized on the horizon and recover some key uniqueness and analyticity results similar to the Minkowski space \cite{kay1991theorems}. 

\begin{definition}[Spacetime with bifurcate Killing horizon]
A spacetime with bifurcate Killing horizon is a triple ($\mM$, $g_{\mu\nu}$, $B$) where
\begin{enumerate}
    \item $\mM$ is a manifold with the Lorentzian metric $g_{\mu\nu}$ describing a globally hyperbolic spacetime.
    \item The vector $B$ is an everywhere smooth Killing vector that generates a flow that is well-defined for all $t\in \mathbb{R}$.\footnote{This is sometimes called a {\it complete} Killing vector.}
    \item The Killing vector $B$ vanishes on a co-dimension two orientable submanifold called the bifurcation surface $\sigma$ that belongs to some Cauchy slice of $\mM$.\footnote{This implies that $\sigma$ is a smooth submanifold.} If the restriction of the Killing vector $B$ to this Cauchy slice is everywhere timelike, we say the Killing horizon is {\it stationary}.
\end{enumerate}
\end{definition}
Consider the spacetime $\mM^{1,1}\times X$ with a metric of the form (Kruskal-like coordinates)
\begin{eqnarray}\label{generalmetricSewel}
    ds^2=A(UV,x^i) dU dV+B(UV,x^i)dx^i dx_i
\end{eqnarray}
where $A$ and $B$ are positive-valued smooth functions with $(U,V)\in \mathbb{R}^2$. These spacetimes have a stationary bifurcate Killing horizon with a wedge reflection symmetry $(U,V)\to (-U,-V)$; see Figure \ref{fig:rindlerboost}. The Killing vector $B=i(U\p_U-V \p_V)$ generates the isometric flow $(U,V)\to (e^{2\pi t }U,e^{-2\pi t}V)$. Null codimension one-surfaces at $UV=0$ split the spacetime into four regions similar to the Minkowski spacetime. The norm of the Killing vector vanishes on $UV=0$. These null hypersurfaces are future and past Killing horizons $\mathfrak{H}^+$ and $\mathfrak{H}^-$, respectively, and we have a bifurcate surface $\sigma$ at $U=V=0$. For $z^\pm>0$ we define the right wedges $W(z^+,z^-)=(U>z^+,V>z^-,x^i)$ and their causal complements the left wedges $W'(z^+,z^-)$. Eternal black holes in both asymptotically flat and asymptotically AdS spacetimes in Kruskal coordinates have metrics of this form. For instance, the metric of an eternal AdS$_3$ black holes in Kruskal coordinates is 
\begin{eqnarray}\label{generalmetr}
    ds^2=-\frac{4dUdV}{(1+UV)^2}+\frac{(1-UV)^2}{(1+UV)^2}d\tilde{\phi}^2\ .
\end{eqnarray}
In higher dimensions, the eternal AdS$_{d+1}$ black brane has the metric
\begin{eqnarray}
    &&ds^2=\frac{1}{z^2}(-h(z)dt^2+h(z)^{-1}dz^2+dx^i dx_i)\nn\\
    &&h(z)=1-z^d/z_0^d\ .
\end{eqnarray}
We make the following changes of the variables
\begin{eqnarray}
    &&U=-e^{-\frac{d}{2z_0}(t-z_*)},\qquad V=e^{\frac{d}{2z_0}(t+z_*)},\qquad \frac{dz_*}{dz}=h(z)^{-1}\,
\end{eqnarray}
where  $z_*$ is the Tortoise coordinate. Since 
\begin{eqnarray}
    z_*=\frac{4z_0}{d}\log (-UV)
\end{eqnarray}
the metric in Kruskal coordinates takes the form in (\ref{generalmetricSewel}) for 
\begin{eqnarray}
    &&A=\frac{4}{d^2UV}g_1(UV), \qquad B=g_2(-UV)\nn\\
    &&g_1(-UV)=h(z)\lb\frac{z_0}{z}\rb^2,\qquad g_2(-UV)=\frac{1}{z^2}\ .
\end{eqnarray}
In \cite{sewell1982quantum} Sewel proved a generalization of Bisognano-Wichmann theorem to these spacetimes:
\begin{theorem}\label{thm:summerscurved}
    Consider the vacuum of QFT in a spacetime with metric (\ref{generalmetricSewel}).  Wightman fields smeared on null hypersurfaces $UV=0$ and the KMS condition define a von Neumann algebra of observables for the wedge $W$. Then, the action of the modular flow of $W$ is the Killing flow generated by $B=i(U\p_U-V\p_V)$, and the modular conjugation is $CRT$ with $RT:(U,V)\to (-U,-V)$. 
\end{theorem}
\begin{proof}
    See Theorems 4 and 8 in \cite{sewell1982quantum}.
\end{proof}
The null hypersurfaces $UV=0$ have an enhanced symmetry group because null translations $U\to U+z^+$ ($V\to V+z^-$) restricted to the hypersurface $V=0$ ($U=0$) generate isometric flows;\footnote{Note that, importantly, null translations $\p_U$ and $\p_V$ are non Killing vectors of the full geometry.} see Figure \ref{fig:dilatation}. Sewel used this to define the observables associated with the future Killing horizon $\mathfrak{H}^+$ with the null translations $\p_U$ corresponding to a positive Hamiltonian. Therefore, there is a canonical choice of vacuum on this null hypersurface with correlators that can be analytically continued to complex times \cite{kay1991theorems}. The state of the field is extended from horizons $\mathfrak{H}^\pm$ to the right wedge $W$ by the requirement that it enjoys the CRT symmetry in the full spacetime. In physics terminology, such a state is often called a Hartle-Hawking state \cite{jacobson1994note,sanders2015construction}. The physical intuition to have in mind is that modular flow is the unique state-preserving automorphism flow of a von Neumann algebra that satisfies the KMS condition \cite{stratila2020modular}. The Killing flow is manifestly an automorphism of the algebra of the right wedge. The choice of the Hartle-Hawking state ensures the KMS property. Therefore, in the Hartle-Hawking state the Killing flow is the modular flow of the right wedge.

\begin{figure}[t]
    \centering
    \includegraphics[width=0.85\linewidth]{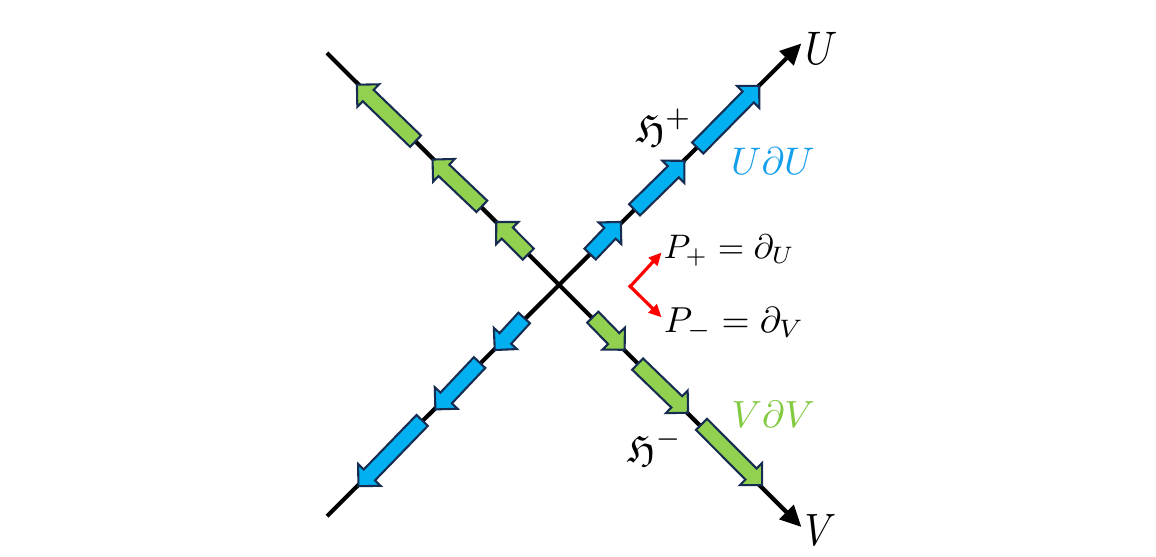}
    \caption{\small{There are enhanced isometry groups on future and past horizons $\mathfrak{H}^\pm$ generated by $P_\pm$ and the Killing flow $B$ that act as dilation and translations: $U\to e^{2\pi t}U+U_0$ and $V\to e^{-2\pi t}+V_0$.}}
    \label{fig:dilatation}
\end{figure}

Summers and Verch proved an analog of the theorem above in curved spacetimes with bifurcate Killing horizons in a Hartle-Hawking state without assuming Wightman fields:\footnote{To incorporate the enhanced symmetry on the horizon in the formalism of algebraic QFT (without Wightman fields), Summers and Verch introduced the notion of {\it symmetry-improving restrictions} \cite{summers1996modular}. They defined algebras of operators localized on the future (past) horizon $\mathfrak{h}^+$ ($\mathfrak{h}^-$), and used them to prove the above theorem without assuming Wightman fields.}
\begin{theorem}
      Consider the vacuum representation of a QFT in a curved space-time with a bifurcate Killing horizon and the metric (\ref{generalmetricSewel}). Modular flow $\Delta_W^{it}$ of the Hartle-Hawking state is the Killing flow generated by $B$, and modular conjugation is the CRT transformation with $RT:(U,V)\to (-U,-V)$. 
\end{theorem}
\begin{proof}
    See \cite{summers1996modular} for a proof.
\end{proof}

\begin{corollary}\label{cor:curvedfuturemod}
   Consider the Hartle-Hawking vacuum representation of QFT in spacetime with a bifurcate Killing horizon. With respect to the modular flow of $W$ (generated by $B$), the wedge algebras of $W(z^+)$ and $W(z^-)$ are, respectively, modular future and past subalgebras (see Figure \ref{fig:rindlerboost}). 
\end{corollary}

\subsection{From future/past algebras to local Poincar\'e}

In previous subsections, we showed that in Minkowski space QFT the Poincar\'e group implies future/past algebras with respect to null translations (Corollary \ref{cor:future}) and modular future/past algebras with respect to boosts (Corollary \ref{cor:futuremod}). In curved space-times, with bifurcate Killing horizons, in the Hartle-Hawking state, we found modular future/past algebras with respect to the Killing flow by $B$ (Corollary \ref{cor:curvedfuturemod}). Here, we present three results in quantum ergodicity, each of which is the converse of one of the corollaries of the previous section. The physics interpretation is that there is an emergent ``local" Poincar\'e algebra in quantum systems that have modular future and past subalgebras. We postpone the systematic discussion of what quantum dynamical systems that show future/past subalgebras to section \ref{sec:ergodicity}.

The converse to Corollary \ref{cor:future} is:
\begin{theorem}[Poincar\'e Group from Future/Past Algebras]\label{EmergentBorch}
    Assume $\pm G_\pm$ are a pair of commuting positive operators that kill the vacuum:  $G_\pm\ket{\Omega}=0$. If a subalgebra $\mR$ is simultaneously the future subalgebra for $G_+$ and the past subalgebra for $G_-$ then the modular flow $\Delta_\mR^{it}$ and $e^{i s G_\pm}$ generate a two-dimensional Poincar\'e group.
\end{theorem}
\begin{proof}
     The proof is based on the half-sided translation  theorem of Borchers (see Theorem \ref{thm:HST}) that says the existence of future (past) algebras for a flow with $G_\pm>0$ and $G_\pm\ket{\Omega}=0$ implies the relations
    \begin{eqnarray}\label{Borchrel}
        \Delta_\mR^{-it}e^{i s G_\pm} \Delta_\mR^{it}=e^{i e^{\pm 2\pi t}s G_\pm}, \qquad [G_\pm,K_\mR ]=\pm i G_\pm\ .
    \end{eqnarray}
   Since $[G_+,G_-]=0$, the modular flow $\Delta_\mR^{it}$ and $e^{i s G_\pm}$ generate a two-dimensional Poincar\'e group.
\end{proof}
Summers proved the converse to Corollary \ref{cor:futuremod}: 
\begin{theorem}[Poincar\'e Group from Modular Future/Past Algebras]\label{EmergentSummers}
    Consider the modular flow of a von Neumann algebra $\mR$. If $\mA^+$ and $\mA^-$ are modular future and past subalgebras, respectively, and we have 
    \begin{eqnarray}\label{Jcommut}
        &&J_{\mA^-}J_{\mA^+}=J_\mA J_{\mA^+} J_{\mA^-}J_\mA\ .
    \end{eqnarray}
    The operators $\pm G_\pm$ defined by
    \begin{eqnarray}\label{Gpmdef}
        \pm G_\pm=K_\mR-K_{\mA^\pm}
    \end{eqnarray}
    are positive, and together with $\Delta_\mR^{it}$ they generate the two-dimensional Poincar\'e algebra.
\end{theorem}
\begin{proof}
The proof of this is based on the half-sided modular inclusion theorem (see Theorem \ref{thm:HSMI}) which says that the existence of modular future (past) algebras implies that the operators $G_\pm$ defined by 
\begin{eqnarray}
     \pm G_\pm=K_\mR-K_{\mA^\pm}\geq 0
\end{eqnarray}
satisfy the relations the relations
\begin{eqnarray}\label{algebrarelGs}
    &&\forall s, t \in \mathbb{R}:\qquad \Delta_\mR^{-it}e^{i s G_\pm}\Delta_\mR^{it}=e^{i  e^{\pm 2\pi t}sG_\pm}\nn\\
    &&[G_\pm,K_\mR]=\pm i G_\pm\ .
 \end{eqnarray}
As we argued in Corollary \ref{cor:futuremod} the assumption (\ref{Jcommut}) is equivalent to $[G_+,G_-]=0$, therefore we obtain a two-dimensional Poincar\'e group.
For more details, see \cite{summers2005tomita,borchers2000revolutionizing}.
\end{proof}
In both the Minkowski vacuum and the Hartle-Hawking state in a space-time with a bifurcate Killing horizon, the modular flow of the right wedge $W$ is local. However, a key difference is that $P_\pm$ are Killing vectors of the full Minkowski spacetime; therefore, the modular flow of every wedge $W(z^+,z^-)$ is also local. However, in curved spacetime, it is only the restriction of the action of the modular flow of $W(z^\pm)$ to the Killing horizons $\mathfrak{H}^\pm$ that is local and geometric. 
In general, the modular flows of $W(z^+,z^-)$ are nonlocal. For $z^\pm>0$, we have 
\begin{eqnarray}
    &&\Delta^{-it}_{W(z^+)}\mO(U,V=0)\Delta_{W(z^+)}^{it}=\mO(e^{2\pi t}U,V=0)\nn\\
    &&\Delta^{-it}_{W(z^-)}\mO(U=0,V)\Delta_{W(z^-)}^{it}=\mO(U=0,e^{-2\pi t}V)\ .
\end{eqnarray}
The action of $\Delta_{W(z^\pm)}^{it}$ away from the horizon, even though nonlocal, still satisfies the group relations $t_0\to e^{2\pi t}t_0+s$ generated by the local flow $\Delta_W^{it}$ and the non-local flow $e^{is G_\pm}$ where $G_\pm$ are defined as in (\ref{Borchrel}) as implied by the half-sided modular inclusion Theorem \ref{thm:HSMI}. 
Finally, we write a theorem that resembles a converse to Corollary \ref{cor:curvedfuturemod} because it applies in situations in curved spacetimes where there is only a local Poincar\'e group:
\begin{theorem}[Emergent Local Poincar\'e Group]\label{newthmapprox}
    Consider a pair of von Neumann subalgebras $\mA^\pm \subset \mR$ such that $\mA^+$ and $\mA^-$ are modular future and past subalgebras, respectively. Define 
    \begin{eqnarray}
        &&\mA^\pm(s):=\Delta_\mR^{-is}\mA^\pm \Delta_\mR^{is}\label{shiftedalg}\,\nn\\
        &&\pm G_\pm(s)= K_\mR-K_{\mA^\pm(s)}
    \end{eqnarray}
    and consider the algebra $\mA^+(-s)$ and $\mA^-(s)$, and correspondingly $G^+(-s)$ and $G^-(s)$ for $s\gg 1$.
    Then, in the scaling limit of large enough $s$ such that $z^+z^- e^{-4\pi s}\ll 1$ we have an emergent  Poincar\'e algebra 
    \begin{eqnarray}
    &&\Delta_\mR^{-it}e^{i z^\pm G_\pm} \Delta_\mR^{it}=e^{i e^{\pm 2\pi t}z^\pm G_\pm}\,\nn\\
    &&e^{i z^+G_+(-s)}e^{i z^- G_-(s)} e^{-i z^+G_+(-s)}=e^{i z^- G_-(s)+O(z^-z^+ e^{-4\pi s})}\ .
    \label{approxcomm}
    \end{eqnarray}
\end{theorem}
\begin{proof}
It follows from assumptions and the half-sided modular inclusion theorem (Theorem \ref{thm:HSMI}) that we have the algebraic relations in (\ref{algebrarelGs}) for $G_\pm(0)$. 
In general, $[G_+(0),G_-(0)]\neq 0$, however, the half-sided modular inclusion theorem also implies that for the algebras $\mA^\pm(s)$ defined in (\ref{shiftedalg}) we have
\begin{eqnarray}
    \pm G_\pm(s):= K_\mR-K_{\mA^\pm(s)}=\Delta_\mR^{-is}(\pm G_\pm) \Delta_\mR^{is}=e^{\pm 2\pi s}(\pm G_\pm)\ .
\end{eqnarray}
    Therefore, we have 
    \begin{eqnarray}
        [G_+(-s),G_-(s)]=e^{-4\pi s}[G_+(0),G_-(0)]
    \end{eqnarray}
    which means that at large $s$ they almost commute. 
    Using the Baker-Campbell-Hausdorff  expansion we find (\ref{approxcomm}).
\end{proof}

\section{Future/Past Subalgebras}\label{sec:futurepast}

In this section, we discuss two examples of physical systems with past and future subalgebras and comment on the ergodic properties of their modular flows, namely strong mixing, the second law, the exponential decay of correlators, and the emergence of an approximate Poincar\'e algebra.

\subsection{Large N theories}\label{sec:largeN}

Consider a large $N$ theory in a normalization where the action has the form $S=N \text{tr}(L)$ with no explicit $N$-dependence in the Lagrangian $L$.
In the strict $N\to\infty$ limit of a large $N$ theory, in a KMS state (canonical ensemble) of inverse temperature $\beta$, the thermal one-point function removed single-trace operators $\mO-\braket{\mO}_\beta$ generate an algebra of GFF \cite{leutheusser2021emergent}. Above the Hawking-Page phase transition, the expectation value of the Hamiltonian and its fluctuations in this state are 
\begin{eqnarray}
    \braket{H}_\beta=O(N^2), \qquad \braket{(H-\braket{H}_\beta)^2}_\beta= O(N^2)\ .
\end{eqnarray}
Therefore, the time evolution operator $e^{i Ht}$ cannot be included in the GFF algebra because of its large fluctuations \cite{witten2022gravity}.\footnote{It generates an outer automorphism of the algebra. This already implies that the observable algebra cannot be type I.} The operator $(H-\braket{H}_\beta)/N$ can be included in this algebra, but since it commutes with all the other GFF operators it forms a center. 
We denote the von Neumann algebra of noncentral single-trace operators by $\mR$. For every time interval $I_{12}=(t_1,t_2)$ we define a von Neumann GFF algebra $\mA_{(t_1,t_2)}$ with a KMS state \cite{leutheusser2021causal,leutheusser2021emergent,furuya2023information}. They describe all the events (perturbations) that can occur in that time interval; see Figure \ref{fig:timealgebra}. Time evolution shifts the interval and their corresponding algebras in time:
\begin{figure}[t]
    \centering
    \includegraphics[width=0.8\linewidth]{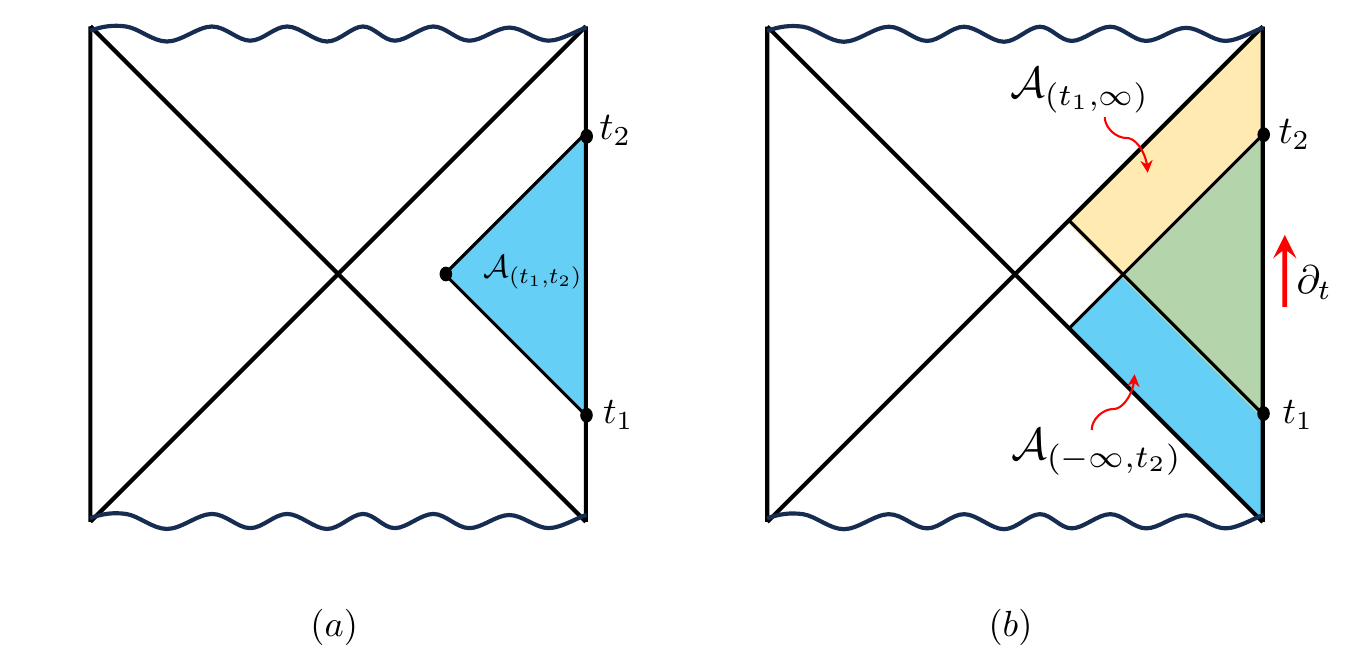}
    \caption{\small{(a) The time interval algebras of GFF above the Hawking-Page phase transition. (b) The algebra $\mA_{(t_1,\infty)}$ is a future subalgebra and $\mA_{(-\infty,t_2)}$ is a past subalgebra.}}
    \label{fig:timealgebra}
\end{figure}
\begin{eqnarray}
    \mA_{(t_1,t_2)}(s):=e^{i Hs}\mA_{(t_1,t_2)}e^{-i Hs}=\mA_{(t_1+s,t_2+s)}\ .
\end{eqnarray}
The events and perturbations that can occur in the entire future (past) form the future (past) von Neumann subalgebras (see Figure \ref{fig:timealgebra}):
\begin{eqnarray}\label{futurepastdef}
\mA_{(t_1,\infty)}&=&\lb\vee_{s> 0}\mA_{(t_1,t_2)}(s)\rb''\subset \mA_{(-\infty,\infty)}\nn\\
\mA_{(-\infty,t_2)}&=&\lb\vee_{s<0}\mA_{(t_1,t_2)}(s)\rb''\subset \mA_{(-\infty,\infty)}\ .
\end{eqnarray}
The future and past subalgebras $\mA_{(t_1,\infty)}$ and $\mA_{(-\infty,t_2)}$, respectively, have the following properties:
\begin{enumerate}
\item {\bf Ergodicity:} The algebraic union of the future and the past subalgebras includes all observables on the right boundary:
\begin{eqnarray}
    \mA_{(-\infty,t_2)}\vee \mA_{(t_1,\infty)}=\vee_{s\in \mathbb{R}}\mA_{(t_1,t_2)}(s)=\mA_{(-\infty,\infty)}\ .
\end{eqnarray}
The orbit of the subalgebra gets us almost everywhere \footnote{Note that if we define the algebraic union as the closure in operator norm topology we obtain C$^*$-algebras, whereas if we close them in weak operator topology, i.e. take the double commutant, we obtain von Neumann algebras.}. 

\item {\bf Strong Mixing (Information loss):} The infinity limit of future algebras is trivial:
\begin{eqnarray}\label{mixingGFF}
    \mA_\infty=\wedge_{s>0}\mA_{(t,\infty)}(s)=\lim_{t\to \infty}\mA_{(t,\infty)}=\lambda 1\ .
\end{eqnarray}
This property is crucial for the correlators to cluster in time (information loss). 

    \item {\bf Half-sided Inclusions (Second law):} The future and past semigroup of time evolution is the restriction map (partial trace):
\begin{eqnarray}
  &&  \forall s>0: \mA_{(t,\infty)}(s)\subset \mA_{(t,\infty)}\nn\\
   && \forall s<0: \mA_{(-\infty,t)}(s)\subset \mA_{(-\infty,t)}\ .
\end{eqnarray}
This property is crucial for the emergence of a second law. 
\end{enumerate}
We will see in section \ref{sec:ergodicity} that the above properties for a general quantum dynamical system define a class of ergodic systems called {\it quantum K-system}.

\paragraph{Strong mixing:} In a KMS state, time evolution is the modular flow. The mixing property in (\ref{mixingGFF}) implies that there are no conserved charges in the algebra (the centralizer of modular flow is trivial). In \cite{furuya2023information}, we used this property to prove that the observable algebra $\mA_{(-\infty,\infty)}$ is a type III$_1$ factor. The same argument implies that general quantum K-systems have no Poincar\'e recurrences (quasiperiodic orbits) at all. See also \cite{narnhofer1989quantum}.

\paragraph{Second law:} To prove a second law of thermodynamics, we consider the thermofield double GFF dual to the eternal black holes. Consider the mutual information between the future algebra of the right boundary $\mA_{(t,\infty)}$ and any subalgebra of the left boundary $\mB \in \mA_{(-\infty,\infty)}'$ as an {\it entropy} function:\footnote{For our discussion of the second law, we only demand the entropy function to be monotonic in time. One may demand the entropy function to satisfy stronger `entropy-like' constraints, such as subadditivity. The investigation of these properties for our proposed entropy functions is beyond the scope of this work.}
\begin{eqnarray}
    \mS(t):=I(\mA_{(t,\infty)}:\mB)\ .
\end{eqnarray}
The half-sided inclusion property implies that forward time evolution with $s>0$ on the right corresponds to the restriction map sending $\mA_{(t,\infty)}$ to its subalgebra $\mA_{(t+s,\infty)}$; see Figure \ref{fig:secondlaw}. Then, it follows from strong subadditivity of entanglement entropy (the monotonicity of mutual information under partial trace) that 
\begin{eqnarray}\label{secondlawmono}
\forall t_1\leq t_2:\qquad \mS(t_1)=I(\mA_{(t_1,\infty)}:\mB)\geq I(\mA_{(t_2,\infty)}:\mB)=\mS(t_2)\ .
\end{eqnarray}
This is a second law of thermodynamics. In fact, any relative entropy of the bulk QFT is monotonic under restriction. This was the idea behind the proof of the generalized second law in \cite{wall2012proof}.

\begin{figure}[t]
    \centering
    \includegraphics[width=.7\linewidth]{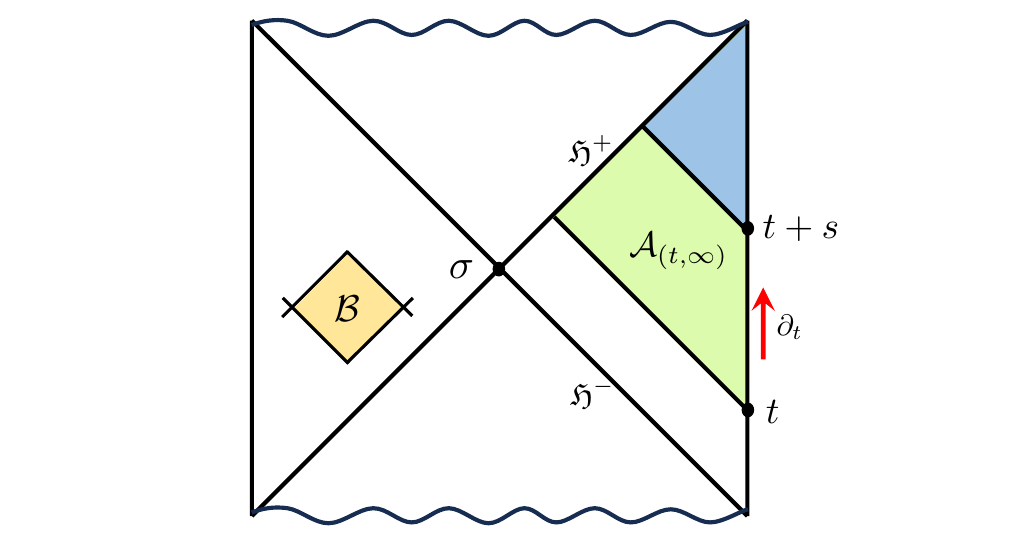}
    \caption{\small{Under half-sided time evolution in the boundary, we have a net of subalgebras $\mA_{(t+s,\infty)}\subset \mA_{(t,\infty)}$ for $s>0$. Correspondingly in the bulk, the mutual information between the future subalgebra $\mA_{(t,\infty)}$ and any subsystem $B$ in the commutant decreases monotonically in time, constituting a second law of thermodynamics.}}
    \label{fig:secondlaw}
\end{figure}

\paragraph{Exponential decay:} The thermal one-point function of GFF in this background vanishes because the expectation value of bulk fields in an eternal black hole is zero. The connected thermal two-point functions of GFF are found from retarded Green's functions of the wave-equation on the black hole background \cite{festuccia2007arrow,Dodelson:2023vrw}. Perturbations can be expanded in the basis of quasinormal modes $\mO_\omega$ whose connected correlators decay exponentially fast at late times:
\begin{eqnarray}
    \braket{\mO_\omega(0)\mO_\omega(t)}_\beta\sim e^{(i \omega_R -\omega_I) t}\ .
\end{eqnarray}
Therefore, for a dense set of observables, the connected correlator is expected to decay exponentially fast.

\paragraph{Local Poincar\'e group:} We saw in Theorem \ref{thm:summerscurved} that in the Hartle-Hawking state of QFT in an eternal AdS black hole, the modular flow of the wedge $W$ is the Killing flow generated by $B=i(U\p_U-V\p_V)$. Near the boundary, this Killing flow becomes the generator of the boundary time evolution $\beta\p_t$. 

With respect to the modular flow of the right wedge $W$, the algebras of regions $W(z^+)$ and $W(z^-)$ are the modular future and past subalgebras, respectively. In the near horizon limit, we are considering the regions $W(e^{-s}z^\pm)$ with large $s$; see Figure \ref{fig:localpoincare}. These algebras are dual to the boundary GFF time interval algebras $(-s \beta,\infty)$ and $(-\infty, s\beta)$. Therefore, Theorem \ref{newthmapprox} implies that in the long time limit $s\gg 1$ we obtain an approximate Poincar\'e group on the boundary that is reminiscent of the local Poincar\'e group near the bifurcate Killing horizon.

\begin{figure}[t]
    \centering
    \includegraphics[width=.9\linewidth]{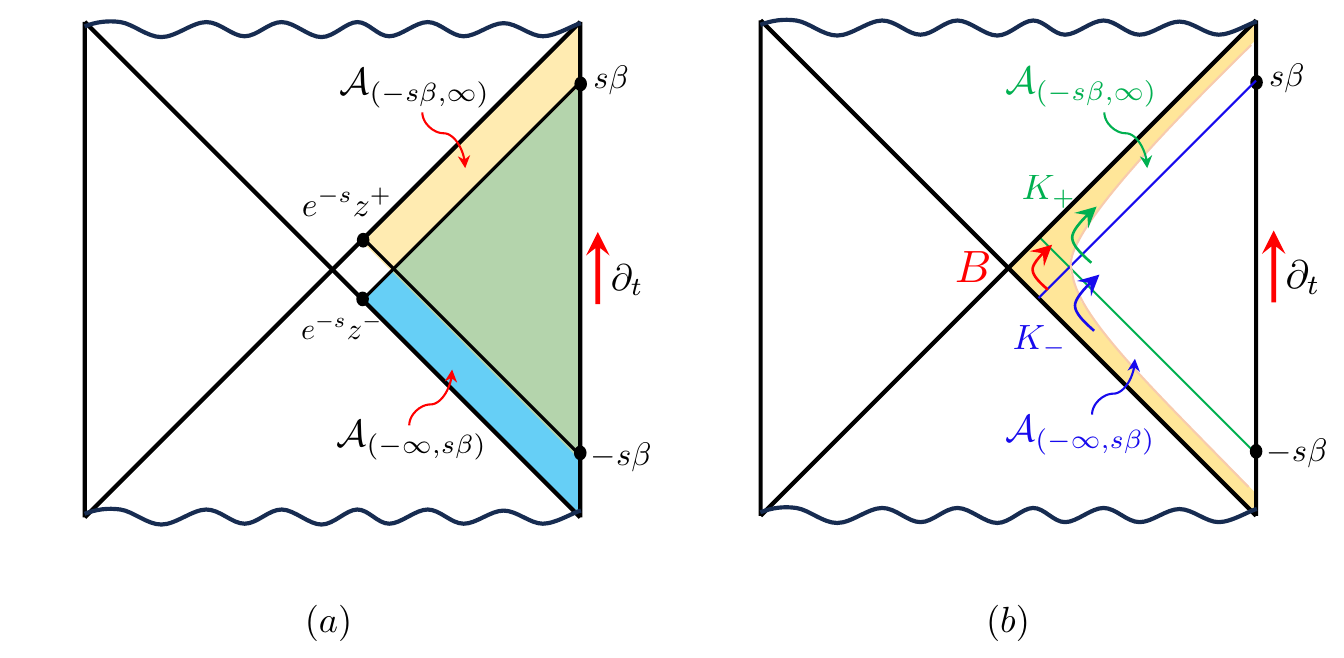}
    \caption{\small{(a) Near-horizon limit in the Hartle-Hawking state of an eternal AdS black hole dual to boundary GFF time interval algebras. (b) In the vicinity of the bifurcation surfaces, we obtain an emergent Poincar\'e group which has a corresponding analog in the boundary.}}
    \label{fig:localpoincare}
\end{figure}

Consider a time interval on the boundary. Without loss of generality, we will choose it to be $I=(-t,t)$. The von Neumann subalgebras of GFF $\mA_{(-t,t)}$ are dual to the algebra of half-spaces defined by a point (sphere in higher dimensions)  $(z^+(t),z^-(t))$ that is null separated from both the boundary points at $\pm t$; see Figure \ref{fig:epsilontimealgebra}. The Theorem \ref{newthmapprox} does not apply directly because the modular flow of the region $W(z^+(t),z^-(t))$ is no longer geometric.
We choose boundary time intervals $I^+_\ep=(-t(1-\ep),t)$ and $I^-_\ep=(-t,t(1-\ep))$ for some $\ep\ll 1$. In the vicinity of the edge $(z^+(t),z^-(t))$, we expect the modular flow of $W(z^+,z^-)$   to be well-approximated by a local boost around this bulk point; see Figure \ref{fig:epsilontimealgebra}. 
We have an approximate notion of half-sided modular inclusion relations, but with respect to the subspace of observables localized near $(z^+(t),z^-(t))$.
This suggests that there might be a generalization of Theorem \ref{newthmapprox} that will explain the emergence of the local Poincar\'e algebra near any point $(z^+,z^-)$ in the geometry.
We postpone further exploration of this to future work.

\begin{figure}[t]
    \centering
    \includegraphics[width=.6\linewidth]{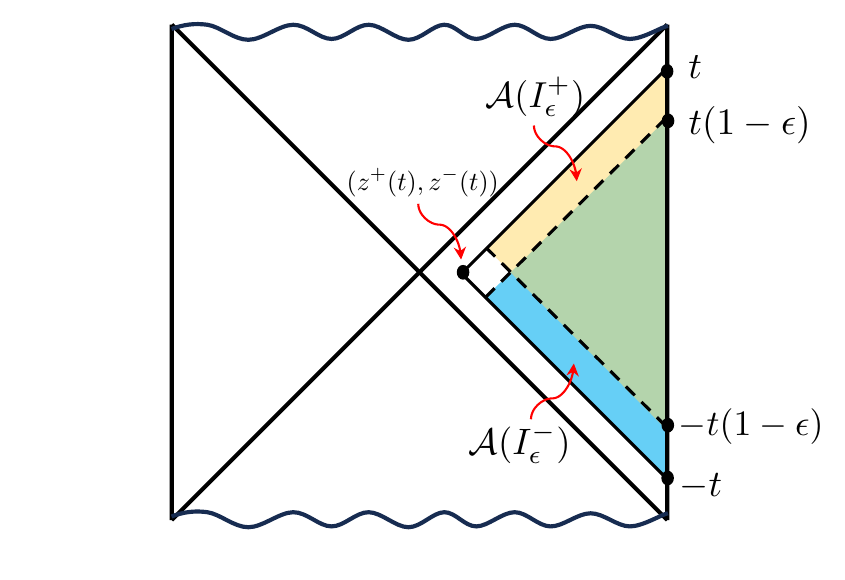}
    \caption{\small{The modular future and past subalgebras, dual to boundary time intervals $I^\pm_\ep$, for wedge subregions away from the Killing horizon.}}
    \label{fig:epsilontimealgebra}
\end{figure}

\subsection{Forward lightcone}

In any QFT, there are von Neumann observable algebras associated with causal developments of ball-shaped regions $D$ that we denote by $\mA(D)$. Every ball-shaped region is defined in terms of its past and future tips; $x_1^\mu$ and $x_2^\mu$ such that $x_2^\mu-x_1^\mu$ is time-like. We denote this algebra by $\mA_{D(x_1^\mu,x_2^\mu)}$; see Figure \ref{fig:forwardlightcone}. 

\paragraph{Second Law:} The future and past algebras of $\mA(D)$ with respect to time evolution are C$^*$-subalgebras 
\begin{eqnarray}
    &&\mA_{D(x_1^\mu,\infty)}=\vee_{t>0}\mA_{D(x^\mu_1,x^\mu_2)}(t)\nn\\
    &&\mA_{D(-\infty,x_2^\mu)}=\vee_{t<0}\mA_{D(x^\mu_1,x^\mu_2)}(t)\nn\\
    &&\mA_{D(x_1^\mu,x_2^\mu)}(t)=e^{i Ht}\mA_{D(x_1^\mu,x_2^\mu)} e^{-i Ht}
\end{eqnarray}
where the algebraic union is defined by taking the closure in operator norm topology. The future/past subalgebras above are C$^*$-subalgebras.
Then, by the same argument as in (\ref{secondlawmono}) the mutual information between the future algebra and any subalgebra in its commutant constitutes a second law; see Figure \ref{fig:forwardlightcone}. 

\begin{figure}[t]
    \centering
    \includegraphics[width=\linewidth]{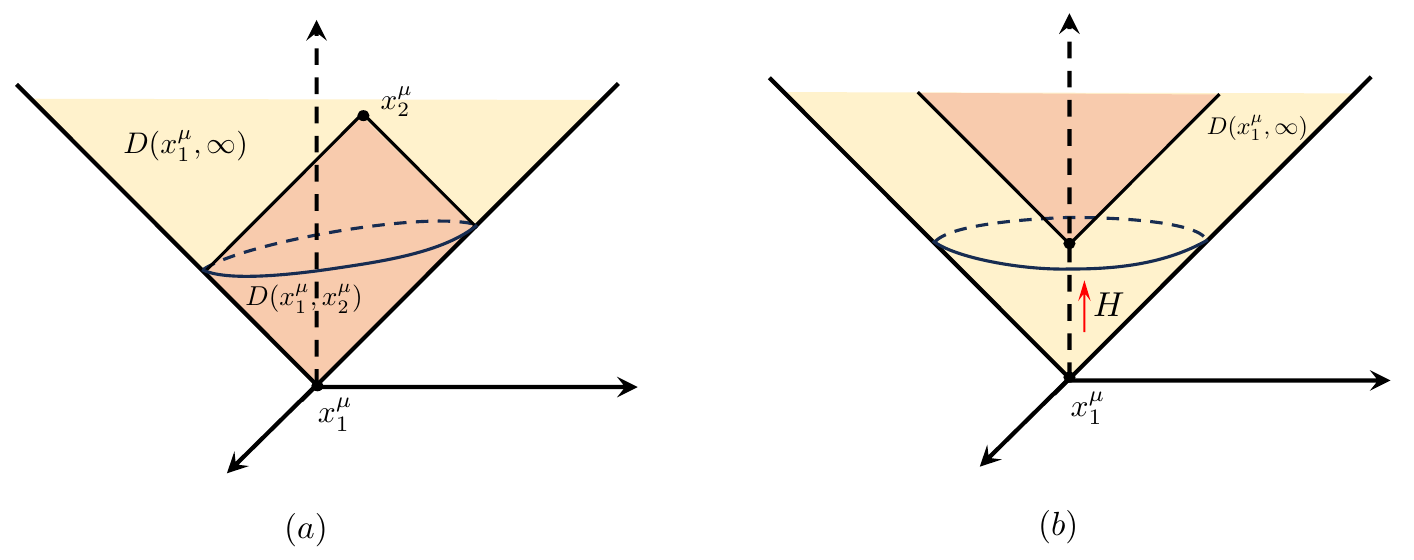}
    \caption{\small{(a) The future algebra of the ball-shaped region $D(x^\mu_1,x^\mu_2)$ with respect to time-evolution. (b) Time evolution is the restriction map on the future subalgebra. At finite temperature this implies that the mutual information between the subalgebra of $D(x_1^\mu,\infty)$ and any subalgebra in the canonically purified copy decreases monotonically in time.}}
    \label{fig:forwardlightcone}
\end{figure}

\begin{figure}[b]
    \centering
    \includegraphics[width=.8\linewidth]{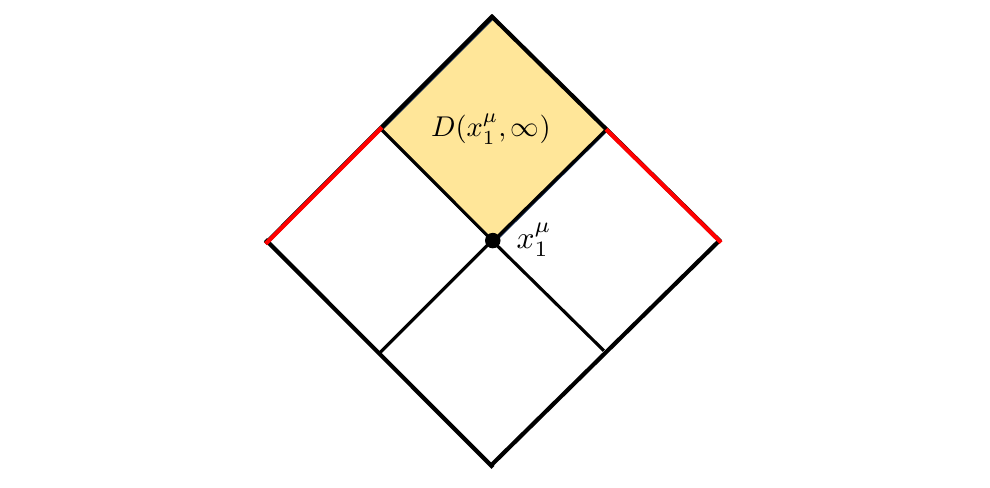}
    \caption{\small{Conformal diagram of the Minkowski spacetime showing that the commutant of the  C$^*$-algebra $D(x_1^\mu,\infty)$ in Minkowski space are operators localized on null infinity. In massive QFT, there are no such operators, therefore taking the double commutant of this C$^*$-algebra gives the von Neumann algebra of all observables on the Cauchy slice: $B(\mH)$.}}
    \label{fig:conformallightcone}
\end{figure}

In a theory of massless free fields, the future and past C$^*$-subalgebras defined above are von Neumann algebras. In this case, the commutant of the future algebra $\mA_{D(0,\infty)}$ is the past algebra $\mA_{D(-\infty,0)}$ simply because for massive free fields time-like and spacelike commutators both vanish. In this case, the modular flow of the future algebra is dilatation $D$, and they satisfy the algebra \cite{brunetti1993modular}
\begin{eqnarray}
    e^{i t D}e^{i sH}e^{-i t D}=e^{ i e^{2\pi t} s H}\ .
\end{eqnarray}
These relations describe a mode $H$ that grows exponentially under modular dynamics $e^{i t D}$, analogous to one of the conditions in (\ref{Borchrel}).
In a massive theory, the double commutant of the future (past) C$^*$-subalgebras is the observable algebra of the whole spacetime, $B(\mH)$. Physically, this happens because the only operators that can be in the commutant of the future algebra $\mA_{(0,\infty)}$ will have to be localized on segments of null infinity; see Figure \ref{fig:conformallightcone}. However, in a massive theory, there are no such excitations. Therefore, the double commutant of the future algebra $\mA_{(0,\infty)}$ in the free field theory is the whole algebra $B(\mH)$.

\paragraph{Local Poincar\'e group:}
It follows from Theorem \ref{newthmapprox} that every modular flow with future and past subalgebras leads to an emergent Poincar\'e algebra. Causal developments of ball-shaped regions in the vacuum of conformal field theory provide an example. Consider the algebra $\mA_{D(x_1^\mu,x_2^\mu)}$ of a ball-shape region $A$ of radius $R$ centered at the origin $x^\mu=0$ in a CFT with the past and future tips are at $x^\mu_1=(-R, \vec{0})$ and $x^\mu_2=(R,\vec{0})$.
 The modular flow of this algebra is local and generated by the conformal Killing vector \cite{brunetti1993modular,haag2012local,Casini:2011kv}: 
\begin{eqnarray}
   K=\frac{i}{2R}\lb\lb R^2-(x^0)^2 -|\vec{x}|^2\rb \p_0-2x^0 x^i \p_i\rb\ .
\end{eqnarray}
The algebra of any ball-shaped region with the same future (past) tip and its past (future) tip inside $A$ is a modular future (past) subalgebra. For simplicity, we consider the case of $1+1d$ CFT, but generalization to arbitrary dimensions is straightforward. We introduce radial null coordinates $x^\pm= x^1 \pm x^0$. In these coordinates, the modular flow is given by
\begin{figure}[t]
    \centering
    \includegraphics[width=\linewidth]{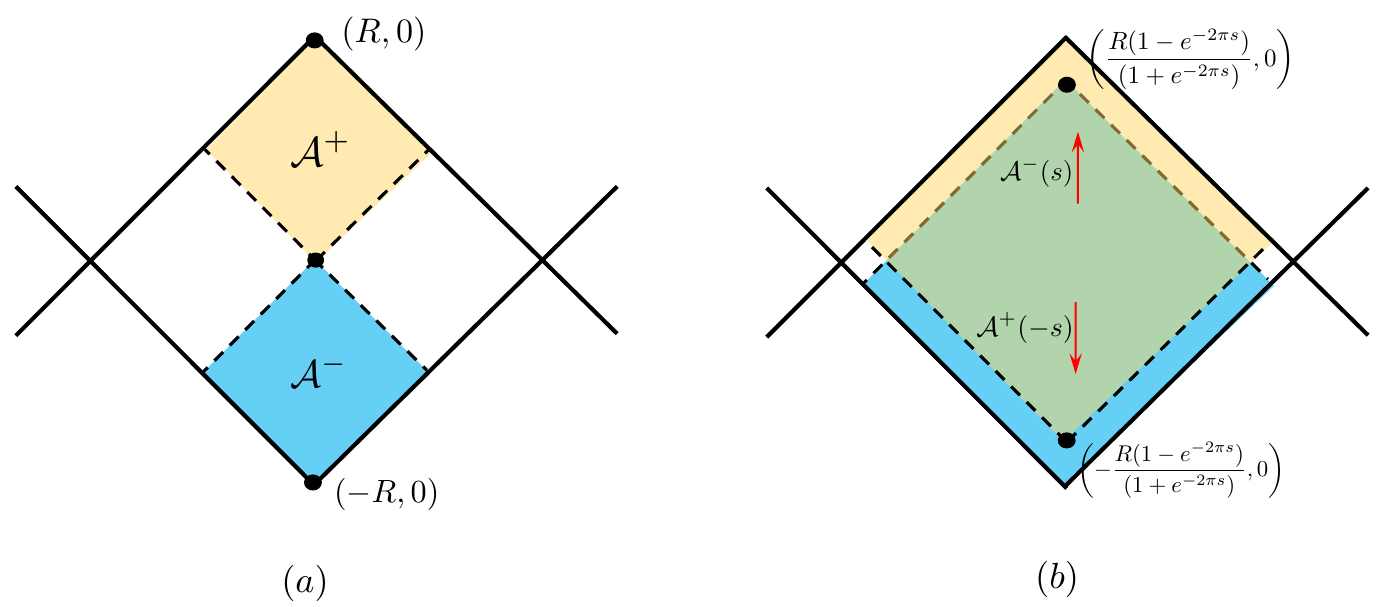}
    \caption{\small{(a) The modular future/past subalgebras $\mA^\pm \subset\mA_{D(x_1^\mu,x_2^\mu)}$ of the causal diamond. (b) The modular flow of the subalgebras $\mA^\pm(\mp s)$ in the large $s$ limit for the emergent local Poincar\'e algebra. The coordinates are shown in the $(x^0, x^1)$ basis.} }
    \label{fig:causaldiamond}
\end{figure}
\begin{eqnarray}
    &&x^\pm(s)=R\frac{\lb 1-e^{\mp2\pi s} \frac{(R-x^\pm)}{(R+x^\pm)}\rb}{\lb 1+e^{\mp2\pi s}\frac{(R-x^\pm)}{(R+x^\pm)}\rb}\nn\\
    && K =\frac{i}{2R}\left(\lb R^2-(x^+)^2 \rb \p_+ - \lb R^2-(x^-)^2 \rb\p_-\right)  \ .
\end{eqnarray}
Now consider another ball shaped region $\Tilde{A}$ centered at the origin with future tip $ (R/2,0)$ and past tip $(-R/2,0)$. The modular future/past subalgebras $\mA^\pm$ correspond to the algebra of $\Tilde{A}$ translated by $\pm R/2$ in the $x^0$-direction as shown in Figure \ref{fig:causaldiamond}. The generators of the modular flow for $\mA^\pm$ are 
\begin{eqnarray}
  K_{\pm} &=&  \frac{i}{R}\left[\lb \lb \frac{R}{2}\rb^2-\lb x^+ \mp\frac{R}{2} \rb^2 \rb \p_+ - \lb \lb \frac{R}{2}\rb^2-\lb x^- \pm \frac{R}{2}\rb^2 \rb\p_-\right]\nn\\
  &=& \frac{i}{R}\left(\lb \pm Rx^+ - (x^+)^2 \rb \p_+ - \lb \mp Rx^- - (x^-)^2 \rb\p_-\right)\ .
\end{eqnarray}
The corresponding modular scrambling modes
\begin{eqnarray}
    \pm G_\pm &=& K - K_{\pm}= \frac{i}{2R} \lb (x^+ \mp R)^2 \p_+ - (x^-\pm R)^2 \p_- \rb
\end{eqnarray}
satisfy the commutation relations
\begin{eqnarray}
    [G_-,G_+] &=&  2iK, \qquad [G_\pm, K]=\pm i G_\pm \ .
\end{eqnarray}
Corresponding to the action of the modular flow of $A$ on $A^\pm$
\begin{eqnarray}
    A^\pm(s) =  e^{ i 2\pi s K } A^\pm e^{- i 2\pi s K},
\end{eqnarray}
we have the generators of the modular flow for $A^\pm(s)$ given by
\begin{eqnarray}
    K_\pm (s) = e^{ i 2\pi s K } K_\pm e^{- i 2\pi s K}\ .
\end{eqnarray}
Following Theorem \ref{newthmapprox}, we consider the modular flow $A^\pm(\mp s)$ such that in the large $s$ limit, both regions are almost equal to $A$; see Figure \ref{fig:causaldiamond}. Under this modular flow, the $(x^+,x^-)$ coordinates of the centers of the causal developments $D(A^\pm(s))$ change according to $\lb\frac{\pm R}{1+e^{2\pi s}}, \frac{\mp R}{1+e^{2\pi s}}\rb$ and their radii change as $\frac{R}{1+e^{-2\pi s}}$. 
Therefore,
\begin{eqnarray}
    K_{\pm}(\mp s) = && \frac{i (1 + e^{-2\pi s})}{2R}\lb \lb \frac{R}{1+ e^{-2\pi s}}\rb^2 - \lb x^+ \mp \frac{R}{1+e^{2\pi s}} \rb^2 \rb \p_+ \nn \\ &&- \frac{i (1 + e^{-2\pi s})}{2R}\lb \lb \frac{R}{1+ e^{-2\pi s}}\rb^2 - \lb x^- \pm \frac{R}{1+e^{2\pi s}} \rb^2 \rb\p_- \nn \\ 
     = && \frac{i}{2R} \lb \lb 1 - e^{-2\pi s}\rb R^2 \pm 2Re^{-2\pi s}x^+ - (1+e^{-2\pi s}) (x^+)^2 \rb \p_+ \nn \\ &&- \frac{i}{2R}\lb \lb 1 - e^{-2\pi s}\rb R^2 \mp 2Re^{-2\pi s}x^+ - (1+e^{-2\pi s}) (x^+)^2 \rb\p_-\ .
\end{eqnarray}
The corresponding modular scrambling modes evolve according to
\begin{eqnarray}
    \pm G_\pm(\mp s) = && K - K_{\pm}(\mp s)\nn\\
    = && \frac{i}{2R} \lb(R^2 - (x^+)^2\p_+ - (R^2-(x^-)^2 \p_-)\rb - K_\pm(\mp s) \nn\\
    = && \frac{ie^{-2\pi s}}{2R} \lb (x^\pm \mp R)^2 \p_+ - (x^-\pm R)^2 \p_- \rb \nn\\
    = &&  e^{-2\pi s}(\pm G_\pm)\ . 
\end{eqnarray}
The commutator of the modular scrambling modes evolves according to
\begin{eqnarray}
    [G_+(-s),G_-(s)] = e^{-4\pi s}[G_+,G_-]\ .
\end{eqnarray}
Thus in the large $s$ limit, they approximately commute and there is an emergent Poincar\'e algebra as expected from Theorem \ref{newthmapprox}.

\section{Quantum Ergodicity and Future/Past Algebras}\label{sec:ergodicity}

Up to now, we argued for the emergence of a local Poincar\'e group in systems with modular future/past algebras. In the remainder of this paper, we elaborate on the ergodic properties of quantum dynamical systems with future/past subalgebras and how they lead to an exponential decay of correlators, a second law of thermodynamics, and maximal modular chaos.

\subsection{Classical ergodic hierarchy}

Drop a handful of blueberries in a glass of water and stir with a straw. We refer to the glass of water as our system $R$, and denote by $A$ the subregion where we initially dropped $A$. The stirring is a dynamical map $T_t:R\to R$ that describes the motion of blueberries in $R$. We say that the system is {\it ergodic} if the blueberries visit every point in the glass.
Ergodic systems mix on average, meaning that from time to time there are moments where a few blueberries are localized in a small corner of the glass, but such fluctuations average out over time. There is no particular state to which the system settles, and if we wait long enough, we will arbitrarily come close to any particular blueberry configuration in the glass. This system is {\it almost periodic}.

If instead of blueberries, we drop a droplet of ink, it gradually spreads until it becomes uniformly distributed over the whole glass. We refer to this final state as an equilibrium state. If the equilibrium state is unique, all states evolve to this unique equilibrium state, and we say that we have {\it strong mixing}.\footnote{The mathematical distinction between the example of blueberries and ink is that the particles of blueberry have finite volume, whereas we take the particles of ink to be pointlike (not a measurable set). If $\rho(x)$ is the density function, in reversible dynamics, $\sup_x\rho(x)$ never changes. However, we only keep track of the total volume of ink in an open neighborhood of a point (a measurable set). Physically, this means that our observables are always coarse-grained on an open set. The decay of coarse-grained observables is not in contradiction with the reversibility of the dynamics.} To keep track of how the system mixes, we use the {\it connected correlation} measure\footnote{Note that $A$ and $B$ are measurable subsets of $R$ and are not to be confused with the Killing field $B$.}
\begin{eqnarray}\label{strongmixing}
    C(A:B_t)=\mu(A\cap B_t)-\mu(A)\mu(B_t)=\mu(A\cap B_t)-\mu(A)\mu(B)
\end{eqnarray}
with $\mu$ the Haar measure in $R$ and $A_t = T_t(A)$. See appendix \ref{app:correlation} for the motivation for this definition.

Physically, strong mixing means that events such as dropping a drop of ink in a particular corner of the glass become increasingly irrelevant to the state of the system in the distant future. In other words, the connected correlators decay at late times. Of course, since the glass of water has finite volume, even though the connected correlators decay, the eventual equilibrium cannot truly forget about the initial droplet because the total amount of ink in the initial droplet decides the final color of the water. This is because the one-point function $\frac{1}{\vol(R)}\int_R \rho(x)$ is invariant with time evolution. Strictly speaking, if one wants the final state to be independent of the initial amount of ink as well, one needs an infinite water reservoir to forget about the entire past. Initially, the drop of ink spreads exponentially in time, and the correlations decay as $e^{-\lambda t}$. In finite volume, we can expect independence from the entire past only in the approximate sense for $1\ll e^{\lambda t}\ll \text{vol}(R)$. For larger times, the ink reaches the walls of the glass and starts folding back. 

Kolmogorov introduced the class of dynamical systems, called {\it Kolmogorov systems}, or in short, {\it K-systems}, in which the system forgets its {\it  entire} far past history (the K-mixing property). As we argued above, forgetting about about the {\it entire past} can only emerge in the thermodynamic limit where we first send $\vol(R)\to \infty$ and then send $t\to\infty$. To understand what independence from the entire past means, we drop two droplets of ink at times $0$ and $t$. We expect the state in the far future to forget both events. In other words, it follows from the definition in (\ref{strongmixing}) that
\begin{eqnarray}
    \lim_{s\to\infty}C(A^{(1)}\cap A^{(2)}_t:B_{t+s})=0\ .
\end{eqnarray}
One can ask what happens if we wait for an infinite amount of time in between the two droplets? The system forgets about both droplets in the future only if 
\begin{eqnarray}
    \lim_{s\to \infty}\lim_{t\to \infty}C(A^{(1)}\cap A^{(2)}_{t}:B_{t+s})=0\ .
\end{eqnarray}
This is equivalent to saying that the three-region correlation measure decays away:
\begin{eqnarray}
    &&\lim_{s\to \infty}\lim_{t\to \infty} C(A^{(1)}: A^{(2)}_t: B_{t+s})=0\nn\\
    &&C(A^{(1)}:A^{(2)}:A^{(3)})=\mu(A^{(1)}\cap A^{(2)}\cap A^{(3)})-\mu(A^{(1)})\mu(A^{(2)})\mu(A^{(3)})\ .
\end{eqnarray}
The strong mixing in (\ref{strongmixing}) is sometimes called {\it strong 2-mixing}, and the property above is called {\it strong 3-mixing}.\footnote{In the zoo of mixing dynamical systems, there are dynamical systems that are strong 2-mixing, but not even strong 3-mixing. See \cite{Tao} for some examples.} It is straightforward to generalize the definition above to {\it strong n-mixing}:
\begin{eqnarray}
&&   \lim_{t_n\to \infty}\cdots \lim_{t_1\to \infty} C(A_{t_1}^{(1)},A_{t_1+t_2}^{(2)},\cdots ,A_{t_1+\cdots+t_n}^{(n)})=0,\nn\\
  &&  C(A^{(1)},A^{(2)},\cdots A^{(n)})=\mu(A^{(1)}\cap A^{(2)}\cap \cdots A^{(n)})-\mu(A^{(1)})\cdots \mu(A^{(n)})\ .
\end{eqnarray}
We expect that dropping ink droplets is an $n$-mixing system for all $n\in \mathbb{N}$.\footnote{A dynamical system that is strong $n$-mixing for all $n$ is called {\it strong mixing of all order}.}
This still does not mean that the system is independent of its {\it entire past}. It is possible that the effect of a perturbation at time $t=0$ can be mimicked, arbitrarily well, by a countably infinite set of events in the far past. In many ergodic systems given any pair of subsystems $A$ and $B$ as we evolve $A$ forward and backward in time, we will eventually end up overlapping with $B$. K-systems emerge when the evolution of $A$ restricted to the past explores only a subset of all possible events. 

The defining property of K-systems is the notion of {\it past subalgebra} (or {\it future subalgebra}). Past subalgebra is formed by the set of all perturbations (events) that can be created in the entire far past. We define the entire past as any countable union of events in the past, i.e. $\cup_{t<0} A_t$. More formally, we define $B(t_0)$ the $\sigma$-algebra of the past as the set generated by a set of events $A_t$ with $t<t_0$ that includes all countable unions and countable intersections of $A_t$ and is closed under complements.\footnote{In ergodic theory, one often considers a further generalization of this definition by choosing an initial set of subsystems $\{A^{(1)}, A^{(2)},\cdots A^{(r)}\}$ and the $\sigma$-algebra generated by this set $B(t_0,r) = \{A^{(i)}_t| t< t_0, i=1,\cdots,r\}$ replaces $B(t_0)$.} In K-systems, the past subalgebra is a strict subalgebra of all perturbations. Famous examples of K-systems include Sinai billiards and the ideal free gas of particles with elastic collisions \cite{szasz1993ergodicity,kramli1991k}.

We are now ready to present some definitions:
\begin{definition}
    A classical dynamical system is a measure space $(X, \Sigma, \mu)$ with a space $X$, a $\sigma$-algebra $\Sigma$\footnote{Not to be confused with the notation for a Cauchy slice.} of measurable sets, a measure $\mu$ and a measure-preserving flow $T_t$ on $X$.
\end{definition} 

\begin{definition}
    We say a classical dynamical system is
    \begin{itemize}
        \item {\bf Ergodic:} if all correlations averaged over time decay to zero:
        \begin{eqnarray}
            \forall A,B\in \Sigma:\qquad \lim_{T\to \infty}\frac{1}{T}\int_0^Tdt\: C(A:B_t)=0\ .
        \end{eqnarray}
        \item {\bf Strong Mixing:} if the correlation functions decay to zero in the far future:
\begin{eqnarray}
    \forall A,B\in \Sigma:\qquad \lim_{t\to \infty}C(A,B_t)=0\ .\qquad\qquad
\end{eqnarray}  
\item {\bf K-mixing:} if the correlations with the entire past decay to zero:
\begin{eqnarray}\label{ksystemclass}
    \forall A, \tilde{A}\in \Sigma:\qquad \lim_{t\to -\infty}\sup_{B\in B(t)}|C(A:B)|=0
\end{eqnarray}
where $B(t)$ is the $\sigma$-algebra generated by $\{\tilde{A}_{t'}|t'< t\}$.
    \end{itemize}
\end{definition}
For more details on classical ergodic theory see appendix \ref{app:classical_dynamics}.
The definition of K-mixing above is equivalent to the definition of a K-system:
\begin{definition}
    Consider a classical dynamical system $(X, \Sigma, \mu)$ and a measure-preserving flow ${T_t}$. It is a classical K-system if there exists a $\sigma$-algebra of measurable sets $\Sigma_0\subset \Sigma$ such that
    \begin{enumerate}
        \item $\vee_{t\in \mathbb{R}} T_t \Sigma_0=\Sigma$.
        \item $\wedge_{t\in\mathbb{R}} T_t \Sigma_0=\{\emptyset, X\}$.
        \item For all $t>0$ we have $T_t\Sigma_0\subset \Sigma_0$.
    \end{enumerate}
    
\end{definition}

In K-systems, all correlations cluster in time, but the decay can be arbitrarily slow. It is an observed fact that in many dynamical systems in physics, the correlators of relevant observables decay exponentially fast.
For example, when we ring a bell, perturbations can be expanded in a basis of {\it quasi-normal modes} that decay exponentially fast. In classical models with discrete time evolution, $T_n=(T_1)^n$ the exponentially decaying (growing) modes correspond to the eigenvectors of $T_1$ with eigenvalues less (more) than one. In continuous time evolution, they correspond to the eigenfunctions of $\p_t$ operator with negative (positive) eigenvalues. Informally, one can think of classical {\it Anosov systems} as a special class of K-systems for which a dense set of correlators decay exponentially fast. Since the Anosov systems are ergodic, show exponential sensitivity to perturbations, and forget about their entire past we will refer to them as {\it classically chaotic systems}.\footnote{In dynamical systems, there are various definitions of classical chaos, however, according to most definitions classical Anosov systems are chaotic.}

Dynamical systems of classical physics are described in terms of flows on a phase space $X$.\footnote{See appendix \ref{app:classical_dynamics} for a brief review.} For example, the phase space is the space of the position and momenta of blueberries or particles of ink.\footnote{The cotangent bundle of $R$.} Bounded measurable functions on the phase space are classical observables of the system. To every smooth observable (function) on the phase space, we can associate a flow along its integral curves that preserves the volume (measure) on the phase space.  Time evolution corresponds to a distinguished smooth function $h$ called the Hamiltonian. The observables $g_\pm$ that satisfy the Poisson bracket relation
 \begin{eqnarray}\label{PoissonHam}
     \{g_\pm, h\}=-\lambda_\pm g_\pm
 \end{eqnarray}
grow/decay exponentially in time. If we denote by $G_\pm$ and $H=\p_t$ the vectors that generate flows along the integral curves of $g_\pm$ and $h$, respectively, we have the commutation relation
\begin{eqnarray}
    [G_\pm,H]=-\lambda_\pm G_\pm\ .
\end{eqnarray}
We obtain a flow on the phase space with expanding/contracting directions; see Figure \ref{fig:lyapunov}.

\begin{figure}[t]
    \centering
    \includegraphics[width=0.85\linewidth]{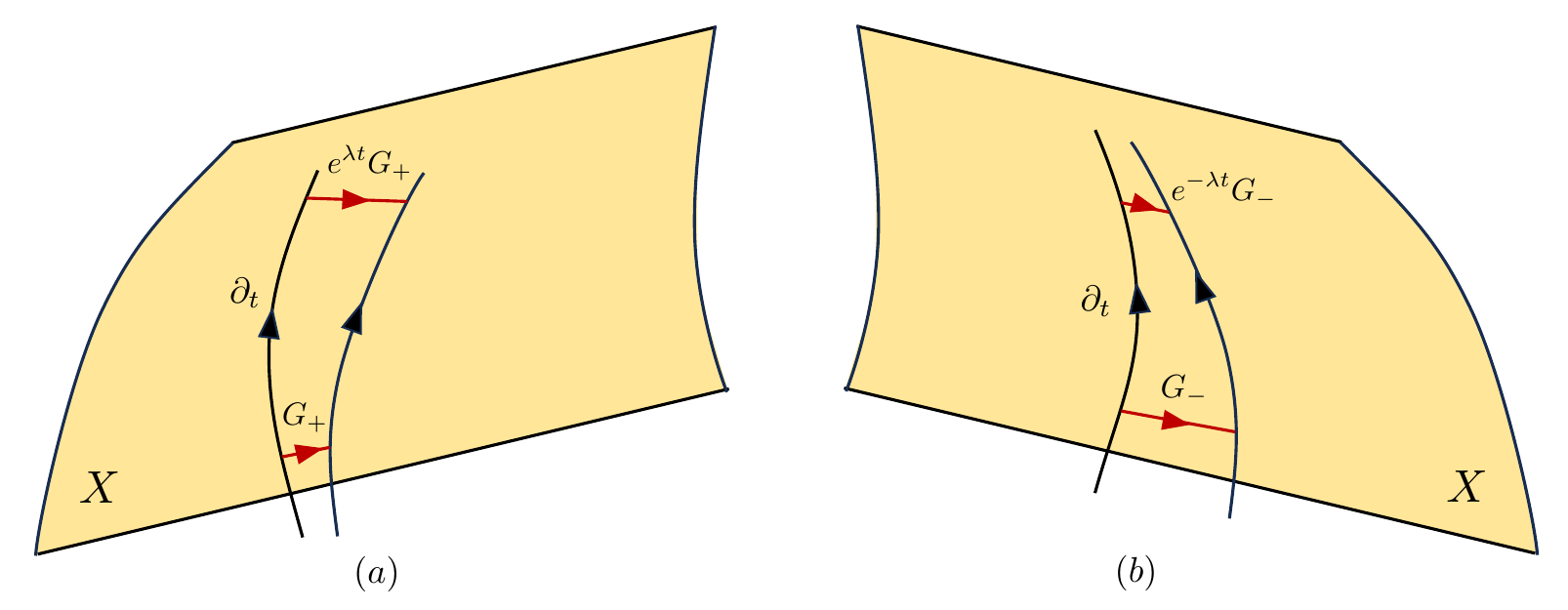}
    \caption{\small{In the phase space $X$ and a dynamical flow $\p_t$, there can be (a) exponentially
    growing modes and (b) exponentially decaying modes.}}
    \label{fig:lyapunov}
\end{figure}
A special class of Anosov systems that play an important role for us are those with $\lambda=2\pi$. 
The connection between the Poincar\'e group and the Anosov systems comes from the commutation relations
\begin{eqnarray}
    &&[B,G_\pm]=\pm 2\pi G_\pm 
\end{eqnarray}
with $B$ playing the role of $H$.
Compare Figure \ref{fig:lyapunov} to Figure \ref{fig:rindleranosov}. Motivated by the modular chaos results in \cite{de2020holographic,faulkner2019modular} we call them {\it maximally chaotic systems}. 
In summary, we have a hierarchy of ergodic systems\footnote{For a review on K-systems see \cite{savvidy2020maximally}.}
\begin{eqnarray}
\text{Maximal Chaos}   \subset\text{Anosov} \subset \text{Kolmogorov}\subset \text{Strong Mixing}\subset 
   \text{Ergodic}\ .
\end{eqnarray}

\subsection{Quantum ergodic hierarchy}

As a first example of a quantum dynamical system, we start with a quantum analog of the example of glass and blueberries. We consider a $d$-dimensional lattice $R$ with a qubit at each site. The lattice is to be compared to a pixelated glass of water and an excitation on a site corresponds to blueberries in a pixel. In the case of a lattice, to every subsystem $A\subset R$ we associate an observable algebra, which is the quantum (non-commutative) analog of the algebra of bounded functions on $A$. Here, the subalgebra on $A$ corresponds to the observable algebra $B(\mH_A)$ where $\mH_A=\otimes_{x\in A}\mH_x$ and $\mH_x$ is the Hilbert space of the qubit at the site $x\in R$. Time evolution is given by the unitary flow $e^{i Ht}$ where $H$ is some Hamiltonian on the lattice that, for example, couples the nearest neighbors. The quantum analog of the example of ink droplets is the continuum limit of the lattice model above, a finite temperature quantum field that lives in some background manifold $\mM$. The von Neumann algebra of observables of the QFT is $\mathcal{R}$ with a state $\omega$ which is represented in the GNS Hilbert space $\mH_\omega=\overline{\mathcal{R}\ket{\Omega}}$. At finite temperature, the set of all observables in a Cauchy slice form a von Neumann algebra $\mR$ in a KMS state represented in the double copy Hilbert space as the thermofield double vector $\ket{\Omega}_{RR'}$. If we take $\mR$ to be the algebra of the right wedge, then $\ket{\Omega}$ is the vacuum or any other CRT symmetric state of QFT. We can choose dynamics to be modular flow, or more generally, any symmetry group $\mathbb{R}$ ($\mathbb{Z}$) of state-preserving transformations defines a  continuous (discrete) quantum dynamical system.

Similar to the classical case, we define connected correlation functions for every pair of observables $a,b\in \mR$ in the state $\ket{\Omega}$
\begin{eqnarray}
    f^{conn}_{ab}(t)=\braket{a\Omega| \Delta^{1/2+it}_\mR b\Omega}-\braket{a\Omega|\Omega}\braket{\Omega|b\Omega}\ .
\end{eqnarray}
Here, we are using a convention with a factor  $\Delta_\mR^{1/2}$  in the definition of the correlator to make the analytic extension of the modular correlators symmetric. They can be interpreted as Left-Right correlators \cite{furuya2023information}.
In a thermalizing system, the expectation is that all connected correlators of operators (except for conserved charges) decay to zero in time. This property is called {\it quantum strong $2$-mixing}: 
\begin{eqnarray}
    \forall a,b \in \mR:\qquad\lim_{t\to \infty}f_{ab}^{conn}(t)=0\ .
\end{eqnarray}
We say a system is quantum strong $n$-mixing if 
\begin{eqnarray}\label{strongquntumnnmixing}
   \lim_{t_1, t_2, \dots, t_{n-1} \to \infty}\braket{a_1(t_1)a_2(t_1+t_2)\cdots a_{n-1}(\sum^{n-1}_{i=1}t_i)\Omega|\Delta^{1/2}_\mR b\Omega}=\braket{\Omega|b\Omega}\prod_{i=1}^{n-1}\braket{a_i\Omega|\Omega}\ .
\end{eqnarray}
If we want the late time observable to be independent of the entire past of the system, we need to generalize the notion of K-systems in (\ref{ksystemclass}) to the quantum realm. 
In the quantum case, the algebraic union replaces the notion of $\sigma$-algebra for classical subregions.\footnote{An important distinction between the classical and quantum dynamical systems is that in the quantum case, under the algebraic union, two finite-dimensional subalgebras that do not commute can generate an infinite-dimensional algebra.} Therefore, we define the future and past subalgebras as:
\begin{definition}[Future/Past subalgebras]
Consider $\mA$ a proper subalgebra of a von Neumann algebra $\mR$ in the standard representation $\mH=\overline{\mR\ket{\Omega}}$. Assume that the dynamics is given by an automorphism of $\mR$ realized by unitaries $e^{iHt}$ that leaves the state $\ket{\Omega}$ invariant, and flows the subalgebra according to $\mA_s=e^{i Hs}\mA e^{-i Hs}$.
\begin{enumerate}
    \item {\bf Future/Past C$^*$-subalgebras:}
If the closure of $\vee_{s< t}\mA_s$ in norm operator topology is a proper subalgebra of $\mR$ we call it the past C$^*$-algebra of $\mA$ in $\mR$. The future C$^*$-algebra of $\mA$ is defined the same with $s< t$ replaced with $s> t$. 
\item {\bf Future/Past von Neumann subalgebras:} If the closure of $\vee_{s< t}\mA_s$ in weak operator topology is a proper subalgebra of $\mR$ we call it the past von Neumann algebra of $\mA$ in $\mR$:
\begin{eqnarray}
    \mA_{(-\infty,t)}=\lb \vee_{s< t}\mA_s\rb''\ .
\end{eqnarray}
The future von Neumann algebra of $\mA$ is defined the same with $s< t$ replaced with $s> t$:
\begin{eqnarray}
       \mA_{(t,\infty)}=\lb \vee_{s>t}\mA_s\rb''\ . 
\end{eqnarray}
\end{enumerate}
\end{definition}
For an ergodic quantum dynamical system, we have
    \begin{eqnarray}
        \mR=\lb \mA_{(-\infty,t)}\vee \mA_{(t,\infty)}\rb''\ .
    \end{eqnarray}
We are now ready to define the following ergodicity classes: 
\begin{definition}
    We say a quantum dynamical system $\mR$ with a dynamical group that is a subgroup of $\mathbb{R}$ is 
    \begin{itemize}
        \item {\bf Quantum Ergodic:} if all correlations averaged over time decay to zero:
        \begin{eqnarray}\label{qergodic}
            \forall a,b\in \mR:\qquad \lim_{T\to \infty}\frac{1}{T}\int_0^Tdt\: f_{ab}^{conn}(t)=0\ .
        \end{eqnarray}
        \item {\bf Quantum Strong Mixing:} if the correlation functions decay to zero in far future:
\begin{eqnarray}\label{qstrongmix}
    \forall a,b\in \mR:\qquad \lim_{t\to \infty}f_{ab}^{conn}(t)=0\ .
\end{eqnarray}    
\item {\bf Quantum K-Mixing:} if the correlations with the entire future decay to zero:
\begin{eqnarray}
\forall \mA\subset \mR, \forall a\in \mR:\qquad \lim_{t\to \infty}\sup_{b\in \mA_{(t, \infty)}}|f_{ab}^{conn}|=0\ .
\end{eqnarray}
    \end{itemize}
\end{definition}
Similar to the classical case, we define quantum K-systems using the three properties we saw in the example of GFF in section \ref{sec:largeN}:
\begin{definition} \label{qksystem}
    Consider a von Neumann algebra $\mR$ represented in a Hilbert space in the standard representation $\mH=\overline{\mR\ket{\Omega}}$, and a unitary flow $e^{i Ht}$ that preserves the vacuum $e^{iH t}\ket{\Omega} =\ket{\Omega}$. We say this quantum dynamical system is a quantum K-system if it has a proper subalgebra $\mA\subset\mR$ with the following properties:
    \begin{enumerate}
        \item {\bf Ergodicity:} $\lb\vee_{s\in \mathbb{R}}\mA_s\rb''=\mR$.
        \item {\bf Strong Mixing:} $\wedge_{s\in\mathbb{R}}\mA_s=\lambda 1$.
        \item {\bf Half-sided translation:} For all $s>0$ (or $s<0$ but not both) we have $\mA_s\subset \mA$.
    \end{enumerate}
\end{definition}
A subtlety that arises is that, as opposed to the classical case, quantum K-systems, and quantum K-mixing, while intimately related, are not equivalent.\footnote{See appendix \ref{app:Ksystem} for more detail.} In this work, we focus on quantum K-systems. As we will show in Lemma \ref{lemma:equivfuturealgebra}, the assumption of half-sided translation above is equivalent to the statement that $\mA$ is a future subalgebra.

Future subalgebras $\mA_{(s,\infty)}$ correspond to all observables we can measure from time $s$ until eternity. They are very special, as when they exist, forward time evolution acts on them as the restriction map which is a unital completely positive (CP) map (the Heisenberg picture of a quantum channel) \cite{furuya2022real}
\begin{eqnarray}
&&\forall t>0:\qquad e^{i Ht}\mA_{(s,\infty)}e^{-iHt}\subset \mA_{(s,\infty)}\ .
\end{eqnarray}
The spectrum of every unital CP map is inside the unit disk:\footnote{This follows from the fact that in the GNS Hilbert space, any unital CP map is represented by a contraction \cite{furuya2022real}.}
\begin{eqnarray}
    \forall t>0, \forall a\in \mA_{(s,\infty)}:\qquad \|e^{i Ht}ae^{-i t H}\|\leq\|a\|
\end{eqnarray}
and the map is uniquely fixed in terms of its spectrum $\lambda\in \mathbb{C}$ such that $\Im(\lambda)\geq 0$
\begin{eqnarray}
    e^{i Ht}a_\lambda e^{-i Ht}=e^{i \lambda t}a_\lambda\ .
\end{eqnarray}
Any mode with $\Im(\lambda)>0$ decays exponentially fast towards the future, and corresponds to a quasi-normal mode.
They satisfy the algebra $[H,a_\lambda]=\lambda a_\lambda$ with the Hamiltonian. 

As opposed to the classical case, in the quantum case, there is no simple way of associating ``smooth" observables with state-preserving unitary flows in the Hilbert space. In defining quantum Anosov systems, we focus on state-preserving flows (symmetries):\footnote{In the quantum world, the volume of the phase space is replaced with the trace on the observable algebra. For general operator algebras, there might not exist a trace. More generally, the quantum analog of measure-preserving flows are unitary flows with a symmetry generator $G\ket{\Omega}=0$ where $\ket{\Omega}$ is the GNS vacuum.}
\begin{definition}
We define a quantum Anosov system as a quantum K-system with expanding and contracting unitary flows that preserve the state, i.e. $e^{i s G_\pm}\ket{\Omega}=\ket{\Omega}$, such that
\begin{eqnarray}
    U(t) e^{i s G_\pm} U^\dagger(t)= e^{ie^{\pm \lambda t}sG_\pm}\ .
\end{eqnarray}
\end{definition}
In the quantum world, we have a similar ergodic hierarchy 
\begin{eqnarray}
    \text{Q. Anosov} \subset \text{Q. Kolmogorov}\subset \text{Q. Strong Mixing}\subset \text{Q. Weak Mixing}\subset \text{Q. Ergodic}\nn
\end{eqnarray}
where the letter Q stands for quantum. In this work, we focus on two classes of measure-preserving quantum dynamical systems: 
\begin{enumerate}
    \item Modular flow.
    \item Flow with a positive generator: $G\geq 0$.
\end{enumerate}
We will see in Theorem \ref{thm:HST} that modular flow K-systems are equivalent to K-systems with positive generators. In the language of ergodic theory, this means that half-sided translations imply half-sided modular inclusions and vice-versa.

\subsection{Algebra types and quantum ergodicity}\label{sec:algebratypes}

The set of (modular) time-invariant operators generates a subalgebra of observables that we refer to as the {\it algebra of conserved charges} $\mathcal{R}^\rho\subset \mR$. Of course, correlators of conserved charges can never decay. To focus on ergodic properties such as strong mixing we need to project to a sector where all the charges are fixed. Fixing all charges corresponds to applying a minimal projection $p\in \mR^\rho$ to the Hilbert space $\mH\to p\mH$. For example, if we have a set of commuting charges $Q_1,\cdots Q_n$ the projection  $\ket{q_1,\cdots q_n}\bra{q_1,\cdots q_n}$ projects to the sector with fixed charges $Q_ip\mH=q_i p\mH$. In other words, we can replace the state with $\rho(p \cdots p)$ so that $p\mR p$ contains no non-trivial time-invariant operators. From now on, we will always project out states so that there are no conserved charges: $\mR^\rho=\lambda 1$.\footnote{In modular theory terminology, we say that the state $\rho_p$ has a trivial centralizer.}

It follows immediately that there cannot be any ergodicity in any finite-dimensional or type I quantum system. In such systems, the Hamiltonian has eigenvectors $\ket{E_n}$, and projections $\ket{E_n}\bra{E_n}$ generate an algebra of time-invariant operators.  If we project to a sector with fixed energy we are left with operators that are also time-invariant. To have ergodicity, we need the Hamiltonian spectrum to be continuous with no point spectrum (normalizable eigenvectors). This can be seen from the definition of quantum ergodicity in (\ref{qergodic}).\footnote{In a quantum ergodic system there is no nontrivial proper subalgebra that is left invariant by the dynamics. In other words, time evolution acts on the algebra irreducibly.} In finite-dimensional quantum systems, or more generally type I von Neumann algebras, infinite time-average decoheres the operators in the energy eigenbasis, and one cannot satisfy the condition (\ref{qergodic}) for all operators. In fact, a stronger statement holds:
\begin{lemma}\label{lemma:Qergodicity}
    Consider a von Neumann algebra $\mR$ and reversible quantum dynamics $U(t)=e^{i Gt}$ that leaves the normalizable GNS (cyclic and separating) vector $\ket{\Omega}$ invariant. The system is quantum ergodic only if both conditions below hold:
    \begin{enumerate}
        \item Either the algebra is type III or type II$_1$ in the tracial state.
        \item The generator $G$ is not a positive operator. 
    \end{enumerate}
\end{lemma}
\begin{proof}
{\bf 1:} Quantum ergodicity in (\ref{qergodic}) means that\footnote{We have defined quantum ergodicity in (\ref{qergodic}) in terms of $f_{ab}^{conn}(t)$ in which the inner product is computed with a factor of $\Delta^{1/2}$. This factor of $\Delta^{1/2}$ is equivalent to replacing operator $b$ with its mirror image $b_J$. However, this replacement of $b$ to $b_J$ is not necessary for the following discussion. }
\begin{eqnarray}
    \forall a, b \in \mR\qquad \lim_{T\to \infty}\frac{1}{T}\int_0^T dt\:\braket{a\Omega|e^{i G t}b}=\braket{a\Omega|\Omega}\braket{\Omega|b\Omega}\ .
\end{eqnarray}
We argued before that quantum ergodicity can happen only if the algebra is infinite-dimensional.
Since the integral on the left-hand-side projects to the subspace of the Hilbert space that is invariant under $e^{i Gt}$, and $\ket{\Omega}$ is cyclic and separating we learn that $\lambda \ket{\Omega}$ is the only ray in the Hilbert space that is invariant under $e^{i G t}$. Then, it follows from Theorem 1 of \cite{longo1979notes} that $\mR$ is either type III or $\ket{\Omega}$ corresponds to a tracial state on $\mR$ (see also Theorem 1 of \cite{Longo:1982zz}). Infinite-dimensional type II algebras are either type II$_1$ or II$_\infty$, however, the trace in type II$_\infty$ is not normalizable. Hence, we find that either the algebra is type III or it is type II$_1$ in a tracial state. 

{\bf 2:} Consider the spectrum of $G$ as a set of $\lambda \in \mathbb{C}$.  Then $e^{i \lambda t}$ is in the spectrum of $U(t)=e^{i G t}$ and $e^{-i\lambda t}$ is in the spectrum of $U(t)^{-1}$. The intersection of the spectrum of $U(t)$ and $U(t)^{-1}$ is trivial if and only if for all $\lambda_1,\lambda_2 \in \text{Spec}(G)$ we have
 \begin{eqnarray}
     e^{i \lambda_1t}\neq e^{-i\lambda_2 t}\ .
 \end{eqnarray}
 We prove by contradiction.
 Assume $G\geq 0$ and we have quantum ergodicity. It follows from positivity of the spectrum of $G$ that $e^{i\lambda_1t}\neq e^{-i \lambda_2 t}$ for all $\lambda_1,\lambda_2 \in \text{Spec}(G)$. We say that $U(t)$ has asymmetric spectrum meaning $\text{Spec}(U)\cap \text{Spec} (U^{-1})=\{1\}$. Then, it follows from Lemma 1 of \cite{Longo:1982zz} that $U(t)$ is an inner automorphism. However, this is in contradiction with the separating condition because an inner automorphism $U(t)$ would imply
 \begin{eqnarray}
     (U(t)-1)\ket{\Omega}=0
 \end{eqnarray}
 and $U(t)-1\in \mR$. As a result, $G$ is not a positive operator.
\end{proof}
It is instructive to discuss the implications of the theorem above in local many-body quantum systems. Consider an infinite lattice with a boson degree of freedom on each site (lattice scalar field theory) in a finite-temperature state of a local Hamiltonian. The vector $\ket{\Omega}$ is the thermofield double. It is invariant under $e^{i (H_L-H_R)t}$ which is expected to induce ergodic and strongly mixing dynamics due to clustering in time. Then, the theorem above implies that the algebra is type III, and indeed the $H_L-H_R$ operator is not positive. If we take the state to be the vacuum of QFT in Minkowski space and the algebra to be a wedge $W$. Then, both null translations $P_\pm>0$ preserve the vector $\ket{\Omega}$, and we have strong mixing of the correlators, however, they do not generate an automorphism of the algebra. 

The above result suggests that if we want to insist on ergodic dynamics with a positive Hamiltonian ($G\geq 0)$ we need to relax the assumption of invertibility. 
Motivated by the case of null translations and the algebra of wedges in QFT we consider the case of semigroup of positive or negative half-line:
\begin{eqnarray}
    \forall t>0,\qquad U(t)\mR U(t)^\dagger\subset \mR\ .
\end{eqnarray}
The definition above naturally motivates future and past subalgebras. We will see in section \ref{sec:HST/HSMI} that such a dynamical system is equivalent to modular dynamics. 

The failure of ergodicity means that there exists a proper subalgebra $\mA\subset \mR$ that is invariant under the flow. In the case of modular flow of $\mR$ this invariance implies that there exists a normalizable conditional expectation from $\mE:\mR\to \mA$. Physically, one can interpret the failure of quantum ergodicity as an exact quantum error correction code \cite{furuya2022real}.

Quantum strong mixing (the assumption in (\ref{qstrongmix})) is much stronger than quantum ergodicity. To make this concrete, we specialize our discussion to modular flows. Strong mixing in a KMS state (or more generally, in modular time) can occur only in type III$_1$ von Neumann algebras:
\begin{theorem}\label{thm:typeiii1}
Consider the modular flow of a von Neumann algebra $\mR$ as quantum dynamical system. Strong quantum mixing implies that the algebra is type a III$_1$ in a state with a trivial centralizer.
\end{theorem}
\begin{proof}
    See \cite{furuya2023information} for a proof.\footnote{Note that it is also true that every type III$_1$ algebra admits states with trivial centralizer besides many other states with centralizers that are factors \cite{marrakchi2023ergodic} (We thank Hong Liu for pointing out this reference to us).  For example, Powers factors type III$_1$ is in a state with a centralizer that is a type II$_1$ factor.}
\end{proof}
KMS states (modular dynamics) is also special in that strong mixing implies the decay of the commutator, the {\it Weak Asymptotic Abelianness}: 
\begin{eqnarray}
    \forall a,b\in \mR:\qquad \lim_{t\to \infty}\braket{\Omega|[a(t),b]|\Omega}=0\ .
\end{eqnarray}
In a quantum system with non-modular dynamics generated by a non-positive generator, we can have type III and type II$_1$ algebras with strong mixing. Examples include translations in lattice quantum systems. This dynamics is strong mixing if we have clustering in space.
We introduced the notion of future/past algebras in connection to quantum K-systems. The motivation was to make rigorous the intuition that late-time observables are independent of the entire past of the system. Independence does not necessarily mean that the late-time observable $a(t)$ commutes (or almost commutes) with the whole past algebra. Such an assumption is called the {\it Norm Asymptotic Abelianness}:
\begin{eqnarray}
    \lim_{t\to \infty}\|[a(t),b]\|=0\ .
\end{eqnarray}
The connection between K-systems and asymptotic Abelianness is not fully understood, however, in the case of type II$_1$ algebras, K-system property ensures {\it Strong Asymptotic Abelianness}
\begin{eqnarray}
    \lim_{t\to \infty}\|[a(t),b]\ket{\Omega}\|=0
\end{eqnarray}
which guarantees the decay of the four-point functions, including the out-of-time-ordered ones \cite{benatti1991strong}. For a discussion of the connection between quantum K-systems and strong asymptotic Abelianness, see appendix \ref{app:Ksystem}.

\subsection{Second law in quantum K-systems}

In the thermodynamic limit, the expectation is that genuinely interacting quantum systems thermalize. In particular, besides strong mixing, we expect the emergence of a second law of thermodynamics that postulates the existence of a non-negative function of the state called {\it entropy} that grows monotonically in time.

We saw that every KMS state with strong modular mixing corresponds to a type III$_1$ algebras. Therefore, in quantum thermalizing systems, the fine-grained entropy of the system diverges. The entropy in the second law must be a coarse-grained notion.
Intuitively, we can justify the emergence of the second law as follows: consider the subspace of all the observables that we can access from time $t$ to eternity and denote it by $\mS_{(t,\infty)}$. This is a {\it future operator system}.\footnote{An operator system is a $*$-closed subspace of observables.} If there are no Poincar\'e recurrences, this provides only partial knowledge about the state of the system, and the coarse-grained entropy in the second law is some information-theoretic measure of the amount of information the observer is missing. Forward time evolution is the restriction map on the future observables:
\begin{eqnarray}
    \forall s>0:\qquad  e^{is H}\mS_{(t,\infty)}e^{-is H}\subset \mS_{(t,\infty)}\ .
\end{eqnarray}
We say time-evolution is a {\it half-sided translation} of the future operator system $\mS_{(t,\infty)}$.
Any information-theoretic measure that satisfies the data-processing inequality decreases monotonically over time.\footnote{See \cite{furuya2023monotonic} for a large class of quantum measures that satisfy the data processing inequality.}  Multiplying any such measure by a minus sign, we obtain a monotonically increasing coarse-grained entropy; otherwise known as a second law of thermodynamics. 

As a first example of an entropy function, consider
\begin{eqnarray}
    D(t)=1-\sup_{\substack{a\in \mS_{(-\infty,0)}\\
    b\in \mS_{(t,\infty)}}}\frac{|\braket{a|\Delta^{1/2}|  b}^{conn}_\beta|}{\|a\|\|b\|}
\end{eqnarray}
for any $t>0$, and the past and future operator systems $\mS_{(-\infty,0)}$ and $\mS_{(t,\infty)}$, respectively.
As we increase $t$ the set of observables $\mS_{(t,\infty)}$ shrink, therefore the supremum decreases, and $D$ increases. In the asymptotic limit $t\to \infty$ if there are no conserved charges we expect $D\to 1$. More generally, for any pair of non-overlapping time intervals $A$ and $B$ we can define the correlation\footnote{If there is a tensor product between $\mS_A$ and $\mS_{B}$ we can rewrite the measure above in terms of trace distance of the reduced states 
\begin{eqnarray}
    D(A,B)=1-\|\psi_{AB_J}-\psi_A\otimes \psi_{B_J}\|
\end{eqnarray}
where we consider the canonical purification of the state and the mirror operator $B_J = JBJ$ in the purifying copy. $J$ is the modular conjugation operator. $D(A,B)$ is monotonic under partial trace, or more generally, satisfies data processing inequality for any quantum channel. }
\begin{eqnarray}
   D(A,B)=1-\sup_{\substack{a\in \mS_A\\
    b\in \mS_B}}\frac{|\braket{a|\Delta^{1/2} b}^{conn}_\beta|}{\|a\|\|b\|}\ .
\end{eqnarray}

The coarse-grained measures we defined above are smooth functions of the state and are always bounded by one. In thermodynamics applications, it is desirable that the coarse-grained entropy is {\it extensive} so that it can grow forever in infinite systems. When the observables $\mS_A$ and $\mS_B$ are C$^*$-algebras, a simple example of such a measure is the mutual information: 
\begin{eqnarray}
    I(A:B)=S(A)+S(B)-S(AB)=S(\psi_{AB}\|\psi_A\otimes \psi_B)
\end{eqnarray}
where $S(\rho_{AB}\|\psi_A\otimes \psi_B)$ is the relative entropy of a C$^*$-algebra defined by \cite{uhlmann1977relative} and \cite{belavkin1982c}, or more generally by \cite{araki1973relative} for von Neumann algebras. Once again, mutual information is monotonically decreasing under restriction. This is the celebrated strong subadditivity of von Neumann entropy:\footnote{Recall that the strong subadditivity inequality is $I(A:BC)\geq I(A:B)$.}
\begin{eqnarray}
&&\forall t>s:\qquad B_s\supseteq B_t\Rightarrow I(A:B_s)\geq I(A:B_t)\ .
\end{eqnarray}
In summary, we find that when the future and past subalgebras exist their fine-grained mutual information gives a coarse-grained entropy that satisfies a second law dictated by strong subadditivity. 
Note that the mutual information we defined above is a generalization of {\it mutual information in time} defined in \cite{hosur2016chaos}.

More generally, we can say any quantum K-system has a second law. 
The following lemma clarifies the equivalence of the assumption of future/past subalgebras, half-sided translations, and the semigroup of quantum channels that we used to prove the second law:
\begin{lemma}\label{lemma:equivfuturealgebra}
    Consider a dynamical flow $e^{i G t}$ acting on the Hilbert space $\mH$ and a proper subalgebra $\mA\subset B(\mH)$.\footnote{Note that the generator of the flow need not be positive.} The following three properties are equivalent and can be used as the definition of a future algebra:
    \begin{enumerate}
        \item {\bf Future subalgebra of $\mB$:} the algebra $\mA$ is a proper subalgebra of $B(\mH)$ such that there exists a subalgebra $\mB$ satisfying $\mA=\vee_{t>0}\mB_t$.
        \item {\bf Half-sided translations:} For all $s>0$ we have $\mA_s \subset \mA$.
        \item {\bf Semigroup of Quantum Channels:} For all $s>t>0$ we have $\mA_s\subset \mA_t$.
    \end{enumerate}
\end{lemma}
\begin{proof}
    {\bf 1$\to$ 2,3:} By definition a future/past subalgebra $\mA$ is a proper subalgebra of $B(\mH)$ such that there exists a subalgebra $\mB$ satisfying $\mA=\vee_{t>0}\mB_t$. Statements (2) and (3) follow from the definition.
    {\bf 2,3 $\to$ 1:} This follows from the observation that if either (2) or (3) are satisfied, we have $\mA=\vee_{t>0}\mA_t$. Therefore, $\mA$ is a future algebra of itself.
\end{proof}

\subsection{Exponential decay of correlators in quantum Anosov systems}

We motivated quantum Anosov systems as examples of quantum K-systems where a dense set of correlators decay exponentially. In this section, we review the proof of this result due to Narnhofer. We start with the following Lemma:
\begin{lemma}[Theorem 3.3 of \cite{emch1994anosov}]\label{Anosoventirespec}
    Consider an Anosov system with two automorphism groups: first,  $U(s)=e^{i G s}$ with respect to which the algebra $\mA$ is a future algebra, and $V(t)=e^{iKt}$ that satisfies the Anosov relation
    \begin{eqnarray}\label{Anosovrel}
    V(t) U(s)V(-t)=U(e^{\lambda t}s)
\end{eqnarray}
for some real $\lambda$. Then, 
\begin{enumerate}
    \item The spectrum of $G$ splits the Hilbert space into three parts $\mH=\mH_0\oplus\mH_+\oplus\mH_-$ corresponding to the subspace of invariant states, positive and negative parts of the spectrum. This decomposition is stable under the flow of $V(t)$ and $U(s)$.
    
    \item Restricting the spectrum of $V(t)$ and $U(s)$ to $\mH_\pm$ and denoting their corresponding generators with $K_\pm$ and $G_\pm$. The spectrum of all four generators $K_\pm$ and $\Lambda_\pm=\log(\pm G_\pm)$ is the entire real line $\mathbb{R}$.
\end{enumerate}
\end{lemma}
\begin{proof}
   {\bf 1:} By definition, we know that $\mH_\pm$ and $\mH_0$ are stable under the flow by $U(s)$: $U(s)\mH_\pm= \mH_\pm$. It follows from the Anosov relations that for any range $I$ in the spectrum of $G$ we have
\begin{eqnarray}
    V(t)E_G(I)V(-t)=E_G(e^{-\lambda t} I)
\end{eqnarray}
where the projections $E_G(I)$ are 
\begin{eqnarray}
        &&E_G(I)=\int_{\tau\in I} dP_G(\tau)\nn\\ 
    &&G=\int_{-\infty}^\infty \tau dP_G(\tau)\ .
\end{eqnarray}
Therefore, $\mH_\pm$ and $\mH_0$ are also stable under the flow by $V(t)$.

{\bf 2:} On $\mH_\pm$ the operators $K_\pm$ and $\pm G_\pm$ are strictly positive, respectively. On these subspaces the operators $\Lambda_\pm=\pm\log(\pm G_\pm)$ satisfy the algebra
\begin{eqnarray}
    V(t)\Lambda_\pm V(-t)=\Lambda_\pm\pm \lambda t\ .
\end{eqnarray}
It follows from the Stone-von Neumann theorem that $K_\pm$ and $\Lambda_\pm$ satisfy the canonical commutation relations 
\begin{eqnarray}
    [\Lambda_\pm, K]=\pm i\lambda
\end{eqnarray}
and the spectrum of $\Lambda_\pm$ is the entire real line.
\end{proof}
Equipped with this lemma, we are now ready to prove exponential clustering in quantum Anosov systems:
\begin{theorem}[Theorem 3.6 of \cite{emch1994anosov}]
   Consider an Anosov system with two automorphism groups $U(s)=e^{i G s}$ with respect to which the algebra $\mA$ is a future algebra, and $V(t)=e^{i tK}$ that satisfies the Anosov relations in (\ref{Anosovrel}) for some $\lambda$.
Then, every operator $a,b\in \mA$ with $a$ in the domain of $G^r$ with $r>0$ and any $\epsilon>0$ we have
\begin{eqnarray}
   |\braket{a|V(t)b}| \lesssim(e^{\lambda t}/\ep)^r\|G^r \ket{a}\|\|\ket{b}\|\ .
\end{eqnarray}
\end{theorem}
\begin{proof}
Consider the range $I=(-\infty,-\ep)\cup (\ep,\infty)$ in the spectrum of $G$. For any $r>0$ we have
\begin{eqnarray}
    V(t)E_G(I)V(-t)=E_G(e^{-\lambda t}I)\leq (e^{\lambda t}G/\ep)^{2r}
\end{eqnarray}
where we have used the operator statement:
\begin{eqnarray}
    E_G(I)\leq (G/\ep)^{2r}\ .
\end{eqnarray}
Now, consider the correlator
\begin{eqnarray}
    |\braket{a|V(t)E_G(I) b}|&=&|\braket{E_G(e^{-\lambda t}I)a|V(t)b}|\nn\\
    &\leq& |\braket{E_G(e^{-\lambda t}I)a|E_G(e^{-\lambda t}I)a}|^{1/2}|\braket{b|b}|^{1/2}\nn\\
    &\leq &(e^{\lambda t}/\ep)^{r} \|G^r \ket{a}\|\|\ket{b}\|\ .
\end{eqnarray}
where we have used Cauchy-Schwarz inequality in the second line.
Moreover for every $\epsilon>0$, there exists some $\delta>0$ such that
\begin{eqnarray}
    \|E_G(-\epsilon,\epsilon)\ket{a}\| \leq \delta \|\ket{a}\|\ .
\end{eqnarray}
Therefore we can remove the projection $E_G(I)$ from the correlator at the cost of a small error controlled by $\delta$:\footnote{A systematic way to do this is to replace each operator $a$ with $a(f)=\int dt f(t) e^{i t G}ae^{-i t G}$
with $f(t)$ that is highly peaked around $t=0$. This way, by choosing appropriate $f(t)$ we can make sure that
$\|E_G(-\ep,\ep)a\ket{\Omega}|^2\ll |a\ket{\Omega}|^2$. For more detail, see Corollary 3.4 and Lemma 3.5 of \cite{emch1994anosov}.} 
\begin{eqnarray}
|\braket{a|V(t) b}| &\leq & |\braket{a|V(t)E_G(I) b}| + |\braket{a|V(t) E_G(I') b}| \nn \\
&\leq& (e^{\lambda t}/\ep)^{r} \|G^r \ket{a}\|\|\ket{b}\| + \delta \|\ket{a}\| \|\ket{b}\|\ .
\end{eqnarray}
Hence, as long as $a$ is in the domain of $G^r$ the correlator decays faster than the exponent $e^{\lambda r t}$. 
\end{proof}

\subsection{Half-side translations/modular inclusions}\label{sec:HST/HSMI}

 A modular quantum K-system is a quantum dynamical flow with a modular future proper subalgebra $\mA\subset \mR$. Then, an important result of quantum ergodic theory is that for modular flow the ergodic hierarchy simplifies (see Corollary \ref{corr:modularhierarchy}). This is based on the key theorem below:
\begin{theorem}[Half-sided modular inclusions]\label{thm:HSMI}
  Consider a von Neumann algebra in a standard GNS representation $\{\mH,\ket{\Omega},\mR\}$. If we have an ergodic modular future subalgebra $\mA\subset \mR$, then the positive operator \begin{eqnarray}
      G=K_\mR- K_\mA\geq 0
  \end{eqnarray}
  generates a unitary flow $U(s)=e^{is G}$ such that: 
  \begin{enumerate}
      \item {\bf Maximal Chaos:} The flow by $U(s)$ corresponds to a growing Anosov mode with a maximal Lyapunov exponent $\lambda=2\pi$, i.e. 
  \begin{eqnarray}
      \forall s, t \in \mathbb{R}:\qquad \Delta_\mR^{-it}e^{i s G}\Delta_\mR^{it}=\Delta_\mA^{-it}e^{i s G}\Delta_\mA^{it}=e^{i  e^{2\pi t}sG}\ .
  \end{eqnarray}
  \item {\bf Future algebra with $G>0$:} The algebra $\mR$ is a future algebra with respect to the dynamics $U(s)$, i.e.
  \begin{eqnarray}
      \forall s>0:\qquad U(s) \mR U(s)^\dagger\subset \mR\ .
  \end{eqnarray}
  \item  {\bf Quantum detailed balance:} We also have 
  \begin{eqnarray}\label{detailedassu}
       \forall s\in \mathbb{R}:\qquad J_\mR e^{i s G}J_\mR=J_\mA e^{i s G}J_\mA=e^{-i sG}\ .
  \end{eqnarray}
  \end{enumerate}
  \begin{proof}
      The proof is standard and can be found in \cite{borchers2000revolutionizing,araki2005extension}.
  \end{proof}
\end{theorem}
\begin{corollary}\label{corr:modularhierarchy}
Consider a quantum dynamical system given by the modular flow of a von Neumann algebra $\mR$. Then, 
\begin{eqnarray}
\text{Maximal Modular Chaos}\equiv   \text{Modular K-system}\ .
\end{eqnarray}
\end{corollary}
Next, we explain why we call the third property in Theorem \ref{thm:HSMI} quantum detailed balance. We start by reviewing detailed balance in classical physics.

The assumption of {\it detailed balance} plays an important role in classical statistical mechanics. It was used by Boltzmann to prove a second law of thermodynamics: entropy production is positive \cite{boltzmann2022lectures}. It says that if the classical equilibrium is described by the probability vector $p_i= e^{-\beta E_i}/Z$ where $Z=\sum_i e^{-\beta E_i}$, and a dynamical process is the matrix $T_{ij}$ then
\begin{eqnarray}
    e^{-\beta E_i}T_{ij}=e^{-\beta E_j}T_{ji}\ .
\end{eqnarray}
Written as a matrix equation this is $T e^{-\beta H}=e^{-\beta H}T^T$.
In other words, it is the assumption that the matrix $e^{\beta H/2}Te^{-\beta H/2}$ is symmetric. 

In quantum systems, the equilibrium distribution is given by the Gibbs state $\rho_\beta\sim e^{-\beta H}$ that in the GNS Hilbert space is represented by the thermofield double $\ket{1}_\beta$ (the canonical purification of the Gibbs state). A dynamical process is a general unital CP map $\Phi:\mA\to \mA$ (the Heisenberg analog of a quantum channel) that acts as a linear contraction on the Hilbert space:\footnote{The GNS Hilbert space represents the algebra $\mA$ as a Hilbert space $\mH$: $a\to \ket{a}$. The superoperators $\Phi:\mA\to \mA$ are represented by linear operators on the GNS Hilbert space; see \cite{furuya2022real} for a review.}
\begin{eqnarray}
    \Phi(a)\ket{1}_\beta=Fa\ket{1}_\beta\ .
\end{eqnarray}
It is natural to define {\it quantum detailed balance} as the condition 
$\Delta^{1/2}F=F^\dagger\Delta^{1/2}$ or the self-adjointness of $\Delta^{1/4}F \Delta^{-1/4}$.\footnote{This is equivalent to saying that the dynamical map $\Phi$ is self-adjoint with respect to the {\it alternate inner product}: $(a,b)_\beta:=\braket{a|\Delta^{1/2}b}_\beta$.} In other words, quantum detailed balance is the assumption
\begin{eqnarray}
    \braket{a|\Delta^{1/2}\Phi(b)}_\beta=\braket{\Phi(a)|\Delta^{1/2}b}_\beta\ .
\end{eqnarray}
 
We are interested in quantum channels that correspond to unitary flows $\Phi_t(a)=e^{i G t}ae^{-iG t}$ and correspond to a symmetry (state-preserving) i.e., $G\ket{1}_\beta=0$. Our detailed balance condition is
\begin{eqnarray}
    U(-t)\Delta^{1/2}=\Delta^{1/2}U(t)\ .
\end{eqnarray}
or equivalently
\begin{eqnarray}
    JU(t)J=U(-t)
\end{eqnarray}
as in (\ref{detailedassu}). 

Finally, we connect modular future/past algebras (half-sided modular inclusion) and future/past algebras with a positive generator (half-sided translations):
\begin{theorem}[Half-sided translations]\label{thm:HST}
    Consider a dynamical flow $U(t)=e^{i Gt}$ that leaves the vacuum invariant $U(t)\ket{\Omega}=\ket{\Omega}$. If $\mA$ is a future algebra ($\mA_t\subset \mA$ for all $t>0$) then the following are equivalent:
    \begin{enumerate}
        \item The generator is positive: $G\geq 0$.
        \item The quantum detailed balance condition $\Delta^{1/2}U(t)=U(-t)\Delta^{1/2}$
        \item The flow satisfies Borchers' relations $\Delta^{-is}U(t)\Delta^{is}=U(e^{2\pi s}t)$.
        \item The algebra $\mA_t$ is a modular future algebra of $\mA$.
    \end{enumerate}
\end{theorem}
\begin{proof}
{\bf 1} $\to$ {\bf 2} $\&$ {\bf 3:} This is the half-sided translation theorem of Borchers. The proofs are standard, for instance, see \cite{borchers2000revolutionizing} or the original paper \cite{borchers1992cpt}.

{\bf 3$\to$ 4:} We write
\begin{eqnarray}
   \forall s,t\geq 0: \Delta_\mA^{-is}\mA_t\Delta_\mA^{is}&=&\Delta_\mA^{-is}U(t)\mA U(t)^\dagger \Delta_\mA^{is}\nn\\
    &=&U(e^{2\pi s}t)\Delta_\mA^{-is}\mA\Delta_\mA^{is}U(e^{2\pi s}t)^\dagger=\mA_{e^{2\pi s}t}\subset \mA_t\,
\end{eqnarray}
where in going to the second line we have used {\bf 3}. 

{\bf 4$\to$ 1:} This was the second statement in Theorem \ref{thm:HSMI}.

{\bf 2} $\to$ {\bf 1:} The proof is in \cite{narnhofer2007quantum}, but we reproduce it here for completeness. 
The idea is to first write the logarithm of the modular operator as 
\begin{eqnarray}\label{logexpress}
K=-\log\Delta^{1/2}=\int_0^\infty d\beta \lb\frac{1}{\Delta^{1/2}+\beta}-\frac{1}
{1+\beta}\rb
\end{eqnarray}
so that 
\begin{eqnarray}
    [G,-\log\Delta^{1/2}]=\int_0^\infty d\beta\, [G,\frac{1}{\Delta^{1/2}+\beta}]\ .
\end{eqnarray}
To compute this commutator we make the following formal manipulations
\begin{eqnarray}
    [G,\frac{1}{\Delta^{1/2}+\beta}]&=&-2\frac{1}{(\Delta^{1/2}+\beta)}G \frac{\Delta^{1/2}}{(\Delta^{1/2}+\beta)}\nn\\
    &=&-2G\frac{\Delta^{1/2}}{(\Delta^{1/2}+\beta)^2}+2[G,\frac{1}{\Delta^{1/2}+\beta}]\frac{\Delta^{1/2}}{\Delta^{1/2}+\beta}
\end{eqnarray}
where in the first line we have used $[G,(\Delta^{1/2}+\beta)^{-1}(\Delta^{1/2}+\beta)]=0$. Therefore,
\begin{eqnarray}
     [G,\frac{1}{\Delta^{1/2}+\beta}]&=&\lb -2G\frac{\Delta^{1/2}}{(\Delta^{1/2}+\beta)^2}\rb \lb 1-2\frac{\Delta^{1/2}}{\Delta^{1/2}+\beta}\rb^{-1}\nn\\
     &=&2G\frac{ \Delta^{1/2}}{(\Delta^{1/2}+\beta)(\Delta^{1/2}-\beta)} = G\left(\ \frac{1}{\Delta^{1/2}+\beta} + \frac{1}{\Delta^{1/2}-\beta} \right) \ .
\end{eqnarray}
The integral expression in (\ref{logexpress}) requires a choice of contour for $\beta$ that includes the spectrum of $\Delta^{1/2}$. The spectrum of $\Delta$ is real and positive. Therefore we choose a contour parallel to the real line slightly shifted down to the lower half plane and close the contour in the upper half plane at infinity. Since the integrand is even under $\beta\to-\beta$,
\begin{eqnarray}
    [G,-\log\Delta^{1/2}] &=& \frac{G}{2} \int_{-\infty-i\epsilon}^{\infty-i\epsilon}d\beta \left(\ \frac{1}{\Delta^{1/2}+\beta} + \frac{1}{\Delta^{1/2}-\beta} \right) \nn\\
    &=& \frac{G}{2} \int_{-\infty}^{\infty}d\beta \left(\ \frac{1}{\Delta^{1/2}+\beta - i\epsilon} + \frac{1}{\Delta^{1/2}-\beta + i\epsilon} \right)\ .
\end{eqnarray}
Using the spectral projection $\Delta^{1/2} = \int dP_\lambda\,  \lambda^{1/2}$ we write
\begin{eqnarray}
    [G,-\log\Delta^{1/2}] = \frac{G}{2} \int dP_\lambda \int_{-\infty}^{\infty}d\beta \left(\ \frac{1}{\lambda^{1/2}+\beta - i\epsilon} + \frac{1}{\lambda^{1/2}-\beta + i\epsilon} \right)\ .
\end{eqnarray}
Applying the residue theorem, only the second term has poles in the upper half plane that contribute a residue of $2\pi i$ to the integral. Hence
\begin{eqnarray}
    [G,K] = 2i\pi G\nn\\
    K=-2\log\Delta^{1/2}\ .
\end{eqnarray}
It follows from the Baker-Campbell-Hausdorff expansion that\footnote{An useful version of the Baker-Campbell-Hausdorff expansion is
\begin{eqnarray}
    e^X Y e^{-X} = Y + [X,Y]+ \frac{1}{2!} [X,[X,Y]] + \frac{1}{3!}[X,[X,[X,Y]]]+\cdots\ .
\end{eqnarray}
In the special case when $[X,Y]=\alpha X$, $\alpha\in\mbC$ we have the simplification
\begin{eqnarray}
    e^X e^Y e^{-X} = e^{Y+\alpha X}\ .
\end{eqnarray}}
\begin{eqnarray}
    &&e^{i Gs}K e^{-i Gs}=K-2\pi sG\ .
\end{eqnarray}
Therefore,
\begin{eqnarray}
    \log\Delta_s^{1/2}=\log (e^{iGs}\Delta^{1/2}e^{-iGs})=\log\Delta^{1/2}+\pi s G\ .
\end{eqnarray}
As a result
\begin{eqnarray}
    G=\frac{1}{\pi s}(-\log \Delta^{1/2}+\log \Delta^{1/2}_s)\geq 0\ .
\end{eqnarray}

\end{proof}

We find the following corollary:
\begin{corollary}
    Consider a quantum dynamical system $\{\mA,\mH_\Omega,\ket{\Omega},U(t)\}$ with strongly continuous dynamical flow generated by a Hamiltonian $H\geq 0$ which is ergodic (i.e. $\ket{\Omega}$ is the unique invariant state of $U(t)$). Then there are past (or future) subalgebras. We have a quantum K-system and the algebra $\mA$ is type III$_1$.
\end{corollary}
   \begin{proof}
   This follows from Theorem \ref{thm:HSMI}, Theorem \ref{thm:HST} and Theorem \ref{thm:typeiii1}. Also, see Theorem 4 of \cite{Longo:1982zz}.
\end{proof}

\section{Discussion}

In this work, we formulated an aspect of bulk locality, namely sharp horizons in spacetimes with bifurcate Killing horizons, in terms of universal ergodic properties of von Neumann observable algebras of quantum gravity. We showed that the local Poincar\'e symmetry near the horizon emerges in a certain scaling limit of any quantum system with modular future and past subalgebras (modular quantum K-system). In particular, we found that any modular quantum K-system is maximally chaotic. 

In quantum K-systems late-time observables are independent of the entire past subalgebra. 
This implies quantum strong mixing of all orders (see \eqref{strongquntumnnmixing}). In a theory of GFF above the Hawking-Page phase transition, we have a maximally ergodic system.  The strong two-mixing property implies that at large but finite $N$, in the limit $1\ll t/\beta\ll \log N$, the connected correlators decay. The quantum strong n-mixing says that the connected correlator 
    $\braket{ab_1(t)b_2(2t)\cdots b_n(nt)}^{conn}_\beta$ in the scaling limit $1\ll t/\beta\ll \log N$ vanishes. This fact was used in \cite{chandrasekaran2022large} to argue for a discrete-time second law of thermodynamics at large but finite $N$.

From our point of view, the advantage of the operator-algebraic approach is that in holographic GFF above the Hawking-Page phase transition, the boundary von Neumann algebras associated with time intervals also satisfy the assumptions of Theorem \ref{newthmapprox}.
Therefore, there is an exact emergent Poincar\'e algebra in the boundary GFF. In general, the modular flow of the boundary future and past algebras is highly non-local. It is only in a particular ``near-horizon" limit that this emergent Poincar\'e algebra matches the approximate local Poincar\'e algebra in the bulk.

It is worth commenting on the three key properties we assumed in our general quantum dynamical system that led to our results. We have an observable algebra $\mA$ that evolves unitarily according to $\mA(s)=e^{i Ks}\mA e^{-i Ks}$. The {\bf symmetry} assumption ensures that the unitary flow $U(s)=e^{-i Ks}$ preserves the state. In the language of ergodic theory, we say that time evolution (dynamics) is state-preserving. The second assumption of {\bf future/past subalgebras} implies that the forward time evolution of these subalgebras from any initial time $t_1$ to final time $t_2$ is given by a unital CP map $\Phi_{t_1,t_2}$ (the Heisenberg picture of a quantum channel), namely restriction to a subalgebra. We have a one-parameter family $\Phi_{t_1,t_2}$ that satisfy the following semi-group property $\Phi_{t_1,s}\circ \Phi_{s,t_2}=\Phi_{t_1,t_2}$ for $t_1\leq s\leq t_2$. 
In other words, the forward time evolution of future subalgebras is Markovian (memoryless), and we have a semigroup of quantum channels.\footnote{In our case, time evolution has a time-independent generator $H$:, i.e. $\mA(s)=e^{iKs}\mA e^{-iKs}$. Therefore, the quantum channel $\Phi_{t_1,t_2}$ is only a function of $t_2-t_1$.} The assumption of the {\bf positivity} of the generator of the flow corresponds to the quantum analog of detailed balance. Even though the modular Hamiltonian is not positive, quantum detailed balance is always satisfied by modular flows. It follows from Theorem \ref{thm:HST} that modular dynamics can be generated by a flow with the positive generator: $G\sim K(\mA)-K(\mA(s))\geq 0$.

Finally, we would like to make a comment in connection with cosmological spacetimes. Consider an expanding universe with a (big bang) singularity at $t=0$. Let us denote the past domain of dependence of a co-moving at time $t$ by $W^{-}(t)$. Then it is clear that the algebra of observables in $W^{-}(t)$ are past subalgebras of the time evolution since $W^{-}(t_2) \supset W^{-}(t_1)$ for $t_2 > t_1$. Moreover, these spacetimes satisfy a second law in terms of the monotonic increase of the area of holographic screens \cite{Bousso:2015mqa, Bousso:2015qqa}. It will be interesting to investigate the connection between the second law, the observer-dependence of holographic screens \cite{Bousso:2016fia}, and the existence of past algebras in cosmological spacetimes.

\paragraph{Acknowledgments}

We thank Thomas Faulkner and Elliott Gesteau who pointed us to the literature on quantum Anosov systems. We also thank Yidong Chen, Hong Liu, and Sasha Zhiboedov for insightful conversations. SO is also thankful to Amit Vikram for pointing him to literature on ergodicity. The authors are grateful
to the DOE that supported this work through grant DE-SC0007884 and the QuantISED
Fermilab consortium.

\appendix

\section{From Wightman Axioms to Local Algebras}\label{app:Wightman}

In this appendix, we review the construction of local observable algebras from Wightman fields.

We start with Wightman axioms: 
\begin{enumerate}
\item {\bf Poincar\'e representation:} There is a Hilbert space $\mH$ and a unitary representation of the Poincar\'e group $U(\Lambda^\mu_\nu,a^\mu)$\footnote{Since the Poincar\'e group is non-compact this representation is infinite-dimensional.} with generator that satisfies $P_0\geq 0$ and $P_\mu P^\mu\geq 0$ and a unique Poincar\'e-invariant vector $\ket{\Omega}\in \mH$ called vacuum.
    \item {\bf Wightman fields:} Wightman fields $\varphi_a$ are operator-valued tempered distributions on spacetime, and the Hilbert space is spanned by the field polynomials acting on vacuum (cyclicity).
    \item {\bf Covariance:} Wightman fields transform covariantly under Poincar\'e transformations according to
    \begin{eqnarray}
    U(\Lambda^\mu_\nu,a^\mu)\varphi_a(x^\mu) U(\Lambda^\mu_\nu, a^\mu)^\dagger=[D(\Lambda^\mu_\nu)]^b_a\varphi_b(\Lambda^\mu_\nu x^\nu+a^\mu)
    \end{eqnarray}
    where $D(\Lambda^\mu_\nu)$ is a finite-dimensional representation of the Lorentz group. Note that we are assuming that fields have finite spin.

    \item {\bf Microscopic Causality:} The (anti-)commutator of Wightman fields vanishes when they are spacelike separated.
\end{enumerate}
Wightman axioms imply the axioms of local QFT in the following way:
Consider any bounded open region in spacetime $A$ and consider the subspace of all test functions supported on $A$. Given any such test function $f_A$ we associate the Wightman fields $\varphi_a(f_A)=\int_{x\in A} f_A(x) \varphi_a(x)$. This operator is inside the $*$-algebra generated by $\varphi_a$ therefore it has a densely defined adjoint and it can be closed. We define a von Neumann algebra associated to a particular field $\varphi_a$ and a particular function $f_A$ as the double commutant of the algebra generated by 
\begin{eqnarray}
    \mA_{a,f_A}=\{ \varphi_a(f_A)^\dagger, (\varphi_a(f_A)^\dagger)^\dagger\}''\ .
\end{eqnarray}
It is the smallest von Neumann algebra to which the closure of the $*$-algebra of $\varphi_a(f_A)$ is affiliated \cite{buchholz1987universal}. Finally, we define two von Neumann algebras for a region $A$, the minimal algebra $\underline{\mA}(A)$ and the maximal algebra $\overline{\mA}(A)$ as
\begin{eqnarray}
    &&\underline{\mA}(A)=\lb \vee_{a,f_B,B\subset A}\:\mA_{a,f_B}\rb''\nn\\
    &&\overline{\mA}(A)=\lb \wedge_{a,f_B,A\subset B}\:\mA_{a,f_B}\rb''\ .
\end{eqnarray} 
These von Neumann algebras both satisfy the {\it isotony} property
\begin{eqnarray}
    \forall A\subset B:\qquad \mA(A)\subset \mA(B)\ .
\end{eqnarray}
Since the support of a function is a closed set the algebra $\underline{\mA}(A)$ defined above are continuous from the inside that is to say for all sequence of inclusions $\{ A_1\subset A_2\subset \cdots\}$ with $\cup_i A_i=A$ we have
\begin{eqnarray}
    \underline{\mA}(A)=\vee_i \underline{\mA}(A_i)\ .
\end{eqnarray}
The minimal algebras are additive 
\begin{eqnarray}
    \underline{\mA}(A_1\vee A_2)=\left(\underline{\mA}(A_1)\vee \underline{\mA}(A_2)\right)''\,
\end{eqnarray}
whereas the maximal algebras $\overline{A}(A_i)$ satisfy Haag's duality:
\begin{eqnarray}
    \overline{\mA}(A')=\left(\overline{\mA}(A)\right)'\ .
\end{eqnarray}
Similarly, the set of algebras $\overline{\mA}(A)$ are continuous from the outside that is to say for all sequences of inclusions $\{A_1\supset A_2\supset \cdots\}$ with $\cap_i A_i=A$ we have
\begin{eqnarray}
    \overline{\mA}(A)=\wedge_i \overline{\mA}(A_i)\ .
\end{eqnarray}
For $A$ that are causal developments of balls or wedges, Haag's duality is believed to hold \cite{witten2018aps} meaning that $\underline{\mA}(A)=\overline{\mA}(A)$. We will focus on this case. We also find that the algebras transform covariantly under Poincar\'e transformations
\begin{eqnarray}
    U(\Lambda^\mu_\nu, a^\mu)\mA(A)U(\Lambda^\mu_\nu, a^\mu)^\dagger= \mA( (\Lambda_\nu^\mu,a^\mu) A)\ .
\end{eqnarray}
Every operator in $a\in\mA(A)$ can be canonically decomposed into a Bose part $a_+$ and a Fermi part $a_-$. The observable algebra of $A$ is the von Neumann algebra generated above using bosonic generators ($a_+$ and even powers of $a_-$). Since the Wightman fields are unbounded, we need to assume some regularity conditions to ensure that microscopic causality implies that the observable algebras of spacelike separated regions $A$ and $B$ commute:\footnote{See \cite{driessler1986connection} for a discussion of such regularity conditions.}
\begin{eqnarray}
    [\mA(A),\mA(B)]=0\ .
\end{eqnarray}

We can now collect all the properties satisfied by the algebras $\mA(A)$ and define the axioms of {\it algebraic QFT}. To every bounded open region of spacetime $A$ we associate a C$^*$-algebra $\mA(A)$ that satisfies the following axioms \cite{hollands2018entanglement}:
\begin{enumerate}
    \item {\bf Isotony:} If $A\subset B$ then $\mA(A)\subset \mA(B)$.
    \item {\bf Causality:} If $A$ and $B$ are spacelike separated, we have $[\mA(A),\mA(B)]=0$.
    \item {\bf Vacuum:} There is a unique state $\ket{\Omega}$ that is invariant under the covering of the Poincar\'e group\footnote{We need the covering to describe half-integer spin particles.}  called the vacuum and the GNS representation of the algebra $\mA(A)$ on the vacuum generates the Hilbert space $\mH$. The Poincar\'e transformations are represented using unitary operators $U(\Lambda^\mu_\nu,a^\mu)$ and the spectrum of the generator satisfies : $P^0\geq0$ and $P_\mu P^\mu\geq0$.
    \item {\bf Relativistic Covariance:} For any covering of the Poincar\'e group, the algebra $\mA(A)$ transforms covariantly: $ U(\Lambda^\mu_\nu, a^\mu)\mA(A)U(\Lambda^\mu_\nu, a^\mu)^\dagger= \mA( (\Lambda_\nu^\mu,a^\mu) A)\ .$
\end{enumerate}
In order to define states that are localized in a region, some authors also consider an additional axiom:
\begin{enumerate}
\setcounter{enumi}{4}
    \item {\bf Buchholz-Wichmann Nuclearity:} For an operator $a\in\mA(A)$, the representation $\pi(a) = e^{-\beta P_0}a\ket{\Omega}$ on the GNS Hilbert space satisfies $\|\pi\|_1 \leq e^{(c/\beta)^n}$ for some $n,c>0$.
    
\end{enumerate}

\section{Classical Indepedence and Correlations}
\label{app:correlation}

Consider the example of a glass $R$ of water and a drop of ink in section \ref{sec:ergodicity}.
In statistical physics, measurable functions on $R$ such as the density of ink molecules are observables, and the state at any time associates (expectation) values to observables.\footnote{In the language of probability theory, observables are random variables and states are probability distributions on $R$. Continuous complex functions on $R$ form a C$^*$-algebra, whereas measurable functions with respect to a measure $d\mu$ form a von Neumann algebra of observables.} An observer who has access only to a corner of the glass $A\subset R$ has access to a subalgebra of observables supported in this region. The indicator function on region $A$: 
\begin{eqnarray}
    1_A(x)=\begin{cases}
    1 & \text{if }  x\in A\\
    0 & \text{if }  x\notin A
    \end{cases}
\end{eqnarray}
is called a {\it conditional expectation} because it projects the algebra of all observables in $R$ to the subalgebra of region $A$. If two subregions $A$, $B$ do not overlap we say they have {\it independent subalgebras} because they satisfy:
\begin{eqnarray}
    1_A 1_B=1_{\emptyset}\ .
\end{eqnarray}
Arbitrary regions depend on each other through their overlap
\begin{eqnarray}
    1_A1_B=1_{A\cap B}\ .
\end{eqnarray}
If some other region $C$ satisfies $A\cap C=B\cap C=A\cap B$, intuitively, we say that $A$ and $B$ are {\it independently conditioned on $C$}, and we have (see Figure \ref{fig:conditional})\footnote{We warn the reader that the {\it independence/conditional independence of the subalgebras} is not the same as the statistical independence/conditional independence of two subalgebras in a particular measure (unnormalized state) that we will introduce in the following. For example, for independent subalgebras $\mA_1$ and $\mA_2$, any state $\mu$ with some connected correlators $\mu(a_1,a_2)\neq \mu(a_1)\mu(a_2)$ is not independent.
}
\begin{eqnarray}
    1_{A}1_{B\cup C}=1_A 1_C = 1_B 1_C\ .
\end{eqnarray}

\begin{figure}[t]
    \centering
    \includegraphics[width=0.6\linewidth]{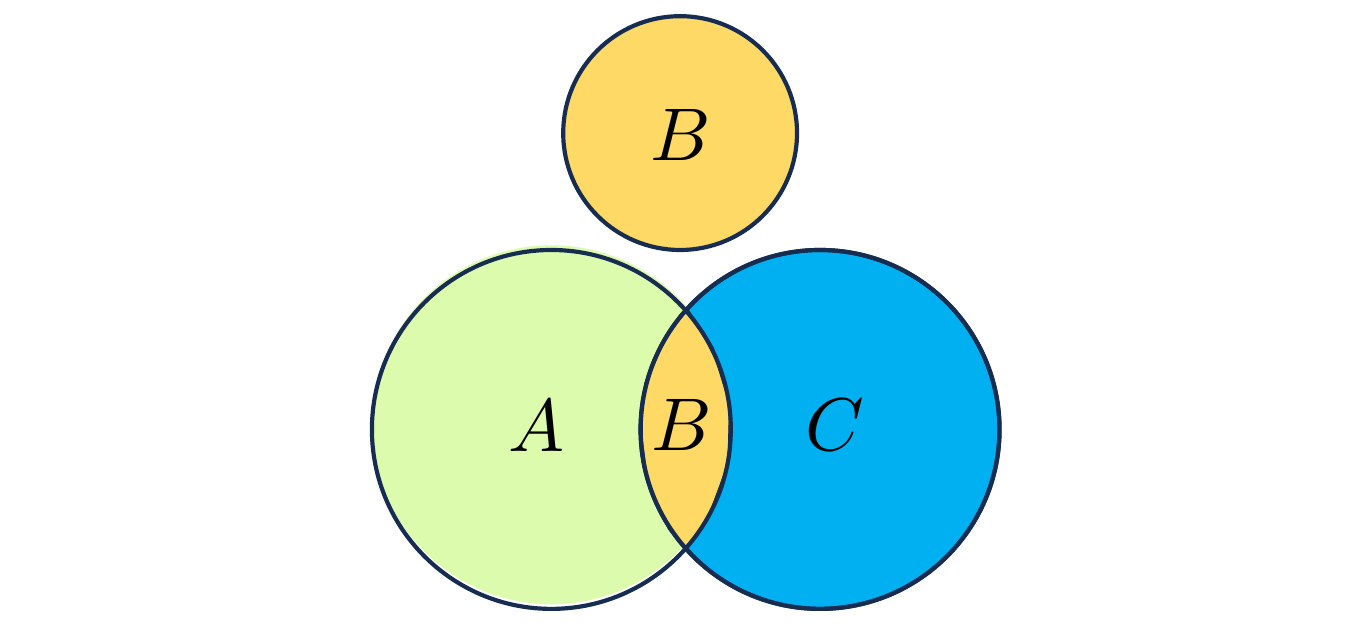}
    \caption{\small{An example of $A$ and $B$ being independently conditioned on $C$. $A$ and $C$ are the two larger circles. $B$ is the smaller circle with the overlapping region between $A$ and $C$. The overlapping regions between any two of them are the same.}}
    \label{fig:conditional}
\end{figure}


In the Heisenberg picture, instead of evolving the state, we evolve observables backward in time, and at time $t$ the ink molecules are spread over the subregion $A_{-t}:=T_{-t}(A)$.\footnote{In the appendices, we focus on $T_{-t}$ instead of $T_t$. This has advantages in a more general setup where $T$ may not be a flow or an automorphism, for example, an endomorphism.} If $A_{-t}\cap B\neq \emptyset$ for some $t$ it means that some ink molecules have made it to the subregion $B$, and the volume of $A_{-t}\cap B$ is a measure of the percentage of the initial amount of ink that made it to $B$:
\begin{eqnarray}
\frac{\vol(A_{-t}\cap B)}{\vol(A)}\ .
\end{eqnarray}
The unique equilibrium state is a uniform distribution, which means the probability that an ink droplet that was added in the far past is found in any other subregion $B$ is proportional to the volume of $B$:
\begin{eqnarray}
    \lim_{t\to \infty}\vol(A_{-t}\cap B)\sim\vol(A)\vol(B)\ .
\end{eqnarray}

For some state $\mu$, let us normalize $\mu(R)=1$ and $\mu(A)=\frac{\vol(A)}{\vol(R)}$ such that $\mu(A)$ is a measure of the relative size of $A$ to $R$. We find that the proportionality constant above is fixed to give
\begin{eqnarray}
    \lim_{t\to \infty}\mu(A_{-t}\cap B)=\mu(A)\mu(B)\ .
\end{eqnarray}

To keep track of how the system mixes, we use the {\it correlation} measure 
\begin{eqnarray}\label{strongmixing2}
    C(A:B_t)=\mu(A\cap B_{-t})-\mu(A)\mu(B_{-t})=\mu(A\cap B_{-t})-\mu(A)\mu(B)
\end{eqnarray}
where we have used the fact that $\mu$ is invariant under time evolution.\footnote{In other words, we are assuming that the transformation is measure-preserving. In classical Hamiltonian systems, this is guaranteed by the Liouville theorem.} 
Note that 
\begin{eqnarray}
    \mu(A|B):=\frac{\mu(A\cap B)}{\mu(A)}
\end{eqnarray}
has the interpretation of conditional probability. It vanishes if and only if the subregions (events) $A$ and $B$
 are statistically independent with respect to measure $\mu$. In other words, $A$ and $B$ are uncorrelated in state $\mu$:
\begin{eqnarray}
    C(A:B)=0\Rightarrow \mu(A\cap B)=\mu(A)\mu(B)\ .
\end{eqnarray}
Two events $A$ and $C$ might be correlated but only through a mediating event called $B$. This occurs if
\begin{eqnarray}
    \mu(A|B\cup C)=\mu(A|B) \,
\end{eqnarray}
and we say that $A$ and $C$ are {\it independent conditioned on $B$}.

\section{Classical Dynamical Systems}
\label{app:classical_dynamics}

In a classical dynamical system, the state of a system is a point in some configuration space $X$, a state is a measure $\mu$ on $X$, and the observables are bounded measurable functions on $X$: $L^\infty(X,\mu)$. A state is a positive continuous map from the observables $L^\infty(X,\mu)$ to complex numbers (expectation values). We assume $X$ is a locally compact manifold so that we have an invariant volume (Haar measure $\mu$) on $X$. Dynamics is given by a discrete or continuous measurable transformation $T_g:X\to X$ where $g\in G$ some group (or monoid if the dynamics is not invertible).\footnote{Ergodic theory on a measure space $X$ is often formulated such that the theorems are true almost everywhere except for a set of measure zero points. This might lead to unwelcome pathologies. Thus we work with equivalence classes of the $\sigma$-algebras which are equivalent up to sets of measure zero. This ensures that only the zero operator has measure 0 and only the identity operator has normalized measure $1$.} In this work, we will always consider $G=\mathbb{R}$ and $T_t$ is a continuous flow on $X$. 
We say that it the dynamics is measure-preserving if it preserves the Haar measure on $X$: $\mu\circ T_t=\mu$.

\subsection{Dynamical systems of classical physics}

In classical physics, the configuration space is the phase space (an even-dimensional symplectic manifold $X$). The symplectic form endows the phase space with a volume form that defines a measure $\mu$. The observables are bounded complex functions on the phase space, i.e. $L^\infty(X,d\mu)$, and the expectation values in the state $d\mu$ are
\begin{eqnarray}
    \mathbb{E}(f)=\int_{x\in X} d\mu(x) f(x)\ .
\end{eqnarray}
Normal states are absolutely continuous measures with respect to $d\mu$. Radon-Nikodym derivative connects normal states and functions in $L^1(X,d\mu)$. For example, consider a theory of $n$ particles on a line with generalized coordinates $\{q_i, p_i\}$ with $i=1,\cdots, n$. The phase space is $(\vec{q},\vec{p})\in \mathbb{R}^{2n}=X$, and the symplectic form is $\omega=\sum_i dq_i\wedge dp_i$. The volume form is $\text{vol}=\omega^n$ which gives the expectation values 
\begin{eqnarray}
    \mathbb{E}(f)=\int f(p_i,q_i)dp_1\cdots dp_n dq_1\cdots dq_n\ .
\end{eqnarray}
It corresponds to the Haar measure on the phase space. The observables are complex continuous functions of $p_i$ and $q_i$, for example, the Hamiltonian $h(\vec{q},\vec{p})=\frac{
1}{2}\sum_i (p_i^2+q_i^2+V(p_i,q_i))$ is an observable. More generally, if the particles live on the configuration space $Y$ spanned by $q_i$, then phase space is the cotangent bundle $X=T^* Y$. 

Dynamics is described by a continuous flow in phase space: $T_t:X\to X$ that evolves observables. We often choose the flow so that it preserves the measure $d\mu$. This means that the state corresponding to $d\mu$ is invariant under the flow (it is in equilibrium). In Hamiltonian dynamics of $n$ particles, the dynamics is defined using a distinguished real function on the phase space, namely the Hamiltonian $h:X\to \mathbb{R}$. Using the symplectic two form, the one-form $dh$ can be identified with the Hamiltonian vector field $H=\vec{v}_h$ on $X$:
\begin{eqnarray}
    &&dh=\sum_i\lb \frac{\p h}{\p q_i} dq_i+\frac{\p h}{\p p_i} dp_i\rb\nn\\
    &&H=\vec{v}_h=\sum_i\lb \frac{\p h}{\p p_i} \frac{\p}{\p q_i}- \frac{\p h}{\p q_i} \frac{\p}{\p p_i}\rb\ .
\end{eqnarray}
Hamilton's equations of motion set 
\begin{eqnarray}
    &&\frac{\p h}{\p p_i}=\frac{dq_i}{dt}\nn\\
    &&\frac{\p h}{\p q_i}=-\frac{dp_i}{dt}\ .
\end{eqnarray}
Therefore, Hamilton's equations of motion will ensure that $H\sim \frac{d}{dt}$. The integral curves of the Hamiltonian vector field solve the Hamilton's equations of motion keeping the energy constant (conservation of energy). We find
\begin{eqnarray}
    &&\frac{df}{dt}=(df)(\vec{v}_h)=\{f,h\}\nn\\
    &&\{f,g\}=(df)(\vec{v}_g)=\omega(\vec{v}_f,\vec{v}_g)\nn\\
    &&\vec{v}_f(\cdot)=\omega(\cdot, df)\ .
\end{eqnarray}
In this formalism, the Lie derivative of $\omega$ is preserved along the Hamiltonian vector field, therefore the flow preserves the volume form $\omega^n$. This is Liouville's theorem. 
More generally, since we have 
\begin{eqnarray}
    \frac{df(h)}{dt}=\{f(h),h\}=0
\end{eqnarray}
any state that corresponds to $f(h)\omega^n$ is also invariant. 
Conserved charges correspond to the functions $q:X\to \mathbb{C}$ that are invariant under time evolution 
\begin{eqnarray}
    \frac{dq}{dt}=\{q,h\}=0\ .
\end{eqnarray}
The algebra of functions on $X$ with the Poisson bracket forms a Poisson algebra, which is a Lie algebra with respect to the Poisson bracket and satisfies the Leibniz rule for any three observables $f,g,h$
\begin{eqnarray}
    \{f, g h\}=g\{f,h\}+h\{f,g\}\ .
\end{eqnarray}
The Lie bracket of the vector fields is fixed by the Poisson bracket of their corresponding functions
\begin{eqnarray}
    [\vec{v}_f,\vec{v}_g]=-\vec{v}_{\{f,g\}}\ .
\end{eqnarray}
Consider a Hamiltonian that has the form 
\begin{eqnarray}\label{Hamilgedo}
    h(\vec{p},\vec{q})=\sum_i \gamma^{ij}(\vec{q})p_i p_j\ .
\end{eqnarray}
It describes the motion of a free particle on the space $Y$ spanned by $\vec{q}$.
Here, $\gamma^{ij}$ is the inverse of the metric tensor on $Y$. The Hamilton equations of motion are
\begin{eqnarray}
    &&\frac{dq^i}{dt}=\gamma^{ij}(\vec{q})p_j=\frac{\p h}{\p p_i}\nn\\
    &&\frac{dp^i}{dt}=-\frac{1}{2}\frac{\p \gamma^{jk}}{\p q_i}p_j p_k=-\frac{\p h}{\p q_i}\ .
\end{eqnarray}
which are the same as the geodesic equation
\begin{eqnarray}
    \frac{d^2q^i}{dt^2}+\Gamma^i_{jk}q^j q^k=0\ .
\end{eqnarray}
Therefore, the dynamics of free particles on space $Y$ is the same as the geodesic flow on this space.

\subsection{Some key results in classical ergodicity}

A dynamical flow in a measure space induces a flow on the functions (observables) of the measure space.
Consider an observable $f \in L^\infty(X,\mu)$. The dynamics flows $f_t(x) = f(T_t x)$ for $x\in X$. We define a linear map $U:L^\infty(X,\mu) \to L^\infty(X,\mu)$. For measure-preserving flows, the $L^1$-norm is invariant under $U$: $\|Uf\|_1 = \|f\|_1$. Thus $U$ is an isometry in $L^1(X,\mu)$ and in $L^2(X,\mu)$. Assuming that the dynamics is invertible,  $U$ is an invertible isometry on the Hilbert space $L^2(X,\mu)$, otherwise known as a unitary operator \cite{koopman1931hamiltonian,halmos1956ergodic}. To simplify our notation, we suppress the measure $\mu$ in our vector spaces $L^p(X,\mu)$.

\begin{theorem}[Von Neumann ergodic theorem]
If $U$ is an isometry on a complex Hilbert space and P is the projection on to the space of all vectors invariant under $U$, then $\frac{1}{T}\int_0^T U_t f \,dt$ converges to $Pf$.
\end{theorem}
\begin{theorem}[Birkhoff's ergodic theorem]
If $T$ is a measure-preserving (but not necessarily invertible) transformation on a space $X$ (with possibly infinite measure) and if $f\in L^1(X)$, then $\frac{1}{T}\int_0^T dt\, f(T_tx)$ converges almost everywhere. The limit function $f^*$ is integrable and invariant almost everywhere. Moreover, if $\mu(X)<\infty$, then $\int_X f(x) dx = \int_X f^*(x) dx$.
\end{theorem}
The proofs of the above two theorems are standard and can be found in \cite{halmos1956ergodic}.
The following two Lemmas define ergodicity in the case of finite and infinite measure spaces:
\begin{lemma}[Ergodicity in a finite measure space]
For a measure-preserving dynamical system on a finite measure space, the following are equivalent: 
\begin{enumerate}
\item The system is ergodic (or irreducible): if $A$ is a subregion invariant under the flow, then $\mu(A)=0$ or $\mu(X-A)=0$.
\item For an integrable function $f\in L^1(X)$, the ergodic limit of the flow projects to a constant function
\begin{eqnarray}
    f^* = \lim_{T\to \infty}\frac{1}{T}\int_0^T dt\, U_tf = \frac{\braket{f}_\mu}{\mu(X)}\ .
\end{eqnarray}
\item For two functions $f,g\in L^1(X)$, the two point function clusters under the ergodic limit
\begin{eqnarray}
    \lim_{T\to\infty} \frac{1}{T}\int_0^T dt \braket{U_tfg}_\mu = \frac{\braket{f}_\mu \braket{g}_\mu}{\mu(X)}\ .
\end{eqnarray}
\item In the ergodic limit, any two subregions $A$ and $B$ become stochastically independent
\begin{eqnarray}
    \lim_{T\to\infty} \frac{1}{T}\int_0^T dt\,\mu(T_{-t}A\cap B) =\frac{\mu(A)\mu(B)}{\mu(X)}\ .
\end{eqnarray}
\item It is metrically transitive: For any two non-trivial subregions $A$ and $B$, there exits some $t>0$ such that $\mu(T_{-t}A\cap B)>0$.
\end{enumerate}    
\end{lemma}

\begin{proof}
{\bf (1$\to$2)} Using Birkhoff's ergodic theorem, the ergodic limit of the flow $f^*$ exists and is an integrable function that is invariant under the flow. Since the flow is ergodic, the only functions that are invariant are constant functions. In the case of finite measure space, Birkhoff's theorem also says that $\int_X dx\, f^*(x) = \int_Xdx\, f(x)$. Thus the constant has to be equal to $\frac{\braket{f}_\mu}{\mu(X)}$.

{\bf (2$\to$3)} Using Fubini's theorem and dominated convergence
\begin{eqnarray}
    \lim_{T\to\infty} \frac{1}{T}\int_0^T dt \braket{U_tfg}_\mu =  \left\langle\lim_{T\to\infty} \frac{1}{T}\int_0^T dt U_tfg\right\rangle_\mu =\frac{\braket{f}_\mu \braket{g}_\mu}{\mu(X)}\ .
\end{eqnarray}

{\bf (3$\to$4)} Let us choose $f$ and $g$ to be the indicator functions $f=1_A$ and $g= 1_B$. Then $\braket{1_A}_\mu = \mu(A)$ and $\braket{1_B}_\mu = \mu(B)$. On the left hand side, the product of the functions $1_A(T_tx)1_B(x)$ is non-zero only when there are points in $B$ which flow to $A$ after some time $t$. Thus
\begin{eqnarray}
    \braket{1_A(T_tx)1_B(x)}_\mu = \mu(A\cap T_tB) = \mu(T_{-t}A\cap B)\ .
\end{eqnarray}

{\bf (4$\to$5)} The average of the integral over $\mu(T_{-t}A\cap B)$ is positive since for any two non-trivial subregions $A, B$, the measure $\mu(A),\mu(B)>0$. This implies that for some $t>0$, we have $\mu(T_{-t}A\cap B)>0$.

{\bf (5$\to$1)} Suppose there exists a non-trivial subregion $A$ that is invariant under the ergodic flow. Then the complement of the region $A'$ is also invariant. Thus $\mu(T_{-t}A\cap A') = \mu(A\cap A')=0$ which is a contradiction.
\end{proof}

In physics, we are often interested in systems with a non-compact phase-space $X$ and a probability distribution $\omega(x)$ (a state) over the phase space. Thus, we can consider the weighted measurement space $(X,\omega)$ and a corresponding state $\ket{\omega^{1/2}}$ such that for any function $f\in L^\infty(X,\omega)$, we have $\braket{f}_\omega = \int_X dx\,\omega(x)f(x)$. For a $\omega$-preserving flow, $U_t^\dagger \ket{\omega^{1/2}} = \ket{\omega^{1/2}}$. Therefore, $\|U_t\|_{2,\omega}=1$ and is an isometry in the Hilbert space $L^2(X,\omega)$. Since it is also invertible, it is a unitary.

\begin{lemma}[Ergodicity in an infinite measure space]
For a faithful state $\omega$ and $\omega$-preserving dynamical system on an infinite measure space, the following are equivalent:
\begin{enumerate}
\item The system is ergodic (or irreducible): if $A$ is a subregion invariant under the flow, then $\omega(A)=0$ or $\omega(X-A)=0$.
\item For observables $f\in L^\infty(X,\omega)$, the ergodic limit of the flow projects to a constant function
\begin{eqnarray}
    \lim_{T\to \infty}\frac{1}{T}\int_0^T dt\, (U_tf) = \braket{f}_\omega\ .
\end{eqnarray}
\item For two observables $f,g\in L^\infty(X,\omega)$, the time-averaged point function connected correlators vanish
\begin{eqnarray}
    \lim_{T\to\infty} \frac{1}{T}\int_0^T dt \braket{(U_tf)g}_\omega = \braket{f}_\omega\braket{g}_\omega\ .
\end{eqnarray}
\end{enumerate}    
\end{lemma}
\begin{proof}
{\bf (1$\to$2)} Since $U_t$ is a unitary in the Hilbert space $\mH_\omega$, the von Neumann ergodic theorem says that the ergodic limit projects to the set of invariant vectors. The only invariant vector is proportional to the GNS vacuum which is the GNS representation of the identity function. The proportionality constant can be obtained by taking the expectation value in the state $\omega$:
\begin{eqnarray}
    \left\langle\lim_{T\to \infty}\frac{1}{T}\int_0^T dt\, U_tf  \right\rangle_\omega &=& \lim_{T\to \infty}\frac{1}{T}\int_0^T dt\, \braket{U_tf}_\omega  \\ \nn
    &=& \lim_{T\to \infty}\frac{1}{T}\int_0^T dt\, \braket{f}_\omega = \braket{f}_\omega
\end{eqnarray}
where in the first line we used Fubini's theorem and dominated convergence.

{\bf (2$\to$3)} The result follows by using Fubini's theorem and dominated convergence.

{\bf (3$\to$2)} Choosing $g = \braket{f}_\omega$, we get
\begin{eqnarray}
    \left\langle\lim_{T\to \infty}\frac{1}{T}\int_0^T dt\, U_tf - \braket{f}_\omega \right\rangle_\omega \braket{f}_\omega = 0\ .
\end{eqnarray}
Since $\omega$ is a faithful state and $f\neq0$, the desired result follows.

{\bf (2$\to$1)} Assume there exists a subregion $A$ that is invariant under the flow. Then the complement $A'$ is also invariant. If we choose $f = 1_A(x)$, for some $x\in A'$:
\begin{eqnarray}
     \lim_{T\to \infty}\frac{1}{T}\int_0^T dt\, U_t1_A(x) = \lim_{T\to \infty}\frac{1}{T}\int_0^T dt\, 1_A(x) = 0\ .
\end{eqnarray}
Since $\omega$ is a faithful state and $1_A$ is a non-zero function, this is a contradiction.  
\end{proof}

\subsection{Examples of classical Anosov systems}

 To make the ergodic hierarchy less abstract, we work out several examples of dynamical systems that are Anosov systems building up towards ``maximal" Lyapunov exponents.

\paragraph{Discrete translations on a torus (Arnold's cat map):} The simplest example is Arnold's cat map. It is a dynamical system in which the configuration space (phase space) is a torus, and we have a discrete dynamics given by the transformation 
\begin{eqnarray}
    (x_{n+1},y_{n+1})=(x_n,y_n)\cdot \begin{pmatrix}
        1&& 2\\
        1 && 1
    \end{pmatrix}
\end{eqnarray}
where $(x_n,y_n)$ are computed mod one. This transformation is area-preserving (its determinant is one) and it is mixing. 
Its unique fixed point is $(0,0)$, the corner of the square.
It is an Anosov system with two eigenvalues (Lyapunov exponents) $\lambda_\pm=\frac{1}{2}(3\pm \sqrt{5})$. The corresponding orthogonal eigenvectors are: stable mode $(x,x\phi)$ and unstable mode $(x,-x\phi^{-1})$ where $\phi=\frac{1}{2}(1+\sqrt{5})$ is the golden ratio. More generally, any $SL(2,\mathbb{Z})$ matrix $T$ with $\text{tr}(T)>2$ gives an Anosov system. It is hyperbolic, meaning that its eigenvalues are $0\leq \lambda_-<1<\lambda_+<\infty$, and it is mixing. The eigenvectors (the principle directions) have irrational slop $y_n/x_n$, therefore, their integral curves are dense in $T^2$. 
Any point $z=(x,y)\in T^2$ viewed as a vector can be decomposed in terms of $v_\pm$ so that 
\begin{eqnarray}
    (z_+ \vec{v}_++ z_-\vec{v}_-)\cdot T^n=z_+ (\lambda_+)^n+z_-(\lambda_-)^n\ .
\end{eqnarray}
We have three automorphisms of the observable algebra $L^\infty(T^2)$. They correspond to $T$ and $\sigma^\pm$ which are the diffeomorphisms that move points along stable and unstable modes: $\sigma^\pm(s)=e^{i s v_\pm}$ and $v_\pm=\p_x\pm \sqrt{5}\p_y$. 



\paragraph{Free nonrelativistic particle on a line:}
The phase space is a two-dimensional manifold of $(x,p)$ where $x$ is the location of the particle and $p$ is the momentum. There is a canonical choice of symplectic form $\omega=dx\wedge dp$. The observables are real Schwartz functions ($C^\infty_0$ functions) in the phase space. To every observable $f$ we associate a $1$-form $df$ and a vector $\vec{v}_f=\omega(df,\cdot)$. 
We have a Poisson bracket between functions that naturally decides the commutator of the integral flows of their corresponding functions
\begin{eqnarray}
  [\vec{v}_f,\vec{v}_g]=-\vec{v}_{\{f,g\}}\ .
\end{eqnarray}
Consider the Hamiltonian $h=\frac{p^2}{2m}$. Functions evolve according to
\begin{eqnarray}
    \lb \p_t-\frac{p}{m}\p_x\rb f(x,p)=0\ .    
\end{eqnarray}
This dynamical system generates a spatial translation $\p_t=\frac{p}{m}\p_x$. The transformation preserves the Haar measure (volume of the phase space) but it is not state-preserving, because the volume of the phase space is infinite. 

Since translations have an entirely continuous spectrum, if the translation-invariant state is unique, strong mixing is guaranteed. There is a subalgebra of observables $\mA_A$ on every subregion $A$ of the phase space. There are {\it future and past subalgebras} associated with these regions as the union of all their translations. This is a K-system. If correlation length is finite, the correlators fall off exponentially fast. If instead of translations we consider the scaling $\sigma_t(f(x))=f(e^{-t} x)$ as the dynamical system we obtain the Anosov structure because scaling and translations  $\tau_s(f(x))=f(x+s)$ satisfy the algebra:
\begin{eqnarray}
    (\sigma_{-t}\circ \tau_s\circ \sigma_{t}f)(x)=f(x+e^t s)\ .
\end{eqnarray}
We have the symmetry group $t\to a t+b$. 
 
\paragraph{Free nonrelativistic particle on hyperbolic spaces:} Consider a nonrelativistic particle on a Riemannian manifold $\mM$ with metric $h_{ab}$. The Hamiltonian function on the phase space $(x^a,p^a)$ is
\begin{eqnarray}
    H(x,p)=\frac{1}{2}h_{ab}(x)p^ap^b\ .
\end{eqnarray}
As we saw earlier, the equation of motion is the geodesic flow $T_t(x^a)=x^a(t)$ with 
\begin{eqnarray}
    \frac{d^2x^a(t)}{dt}+\Gamma^a_{bc}\frac{
    dx^b}{dt}\frac{dx^c}{dt}=0\ .
\end{eqnarray}
Consider $\mM=\mathbb{H}$, the upper half-plane $\Im(z)\geq 0$. At any point $t$ on a geodesic we have two horocyclic flows $T_t$ that satisfy the Anosov relations
\begin{eqnarray}
    T_t S_s^\pm T_{-t}=S^\pm_{e^{\mp t}s}\ .
\end{eqnarray}

The growing and decaying modes with Lyapunov exponents $\lambda=\pm 1$. It follows that the motion of a free nonrelativistic particle on any manifold $X=(\mathbb{H}^2/\Gamma)\times \mM$ for some other manifold $\mM$ and $\Gamma$ for any Fuchsian group, is also Anosov system with, at least, one pair of growing and decaying modes. See appendix \ref{app:horocycle} for more details on this example.

\paragraph{Rindler observer:} Consider Minkowski spacetime $\mathbb{R}^{1,1}$ and a point particle moving with constant acceleration along the path $(x^+(t),x^-(t))$
\begin{eqnarray}
    x^\pm(t)=x_0 e^{\pm t}\ .
\end{eqnarray}
From the point of view of this observer, time evolution is generated by boost $K=x^+\p_+-x^-\p_-$, and the null translations $P_\pm=\p_\pm$ satisfy the Anosov relations 
\begin{eqnarray}
    &&e^{-t K}e^{-s P_\pm} e^{t K}=e^{-s e^{\mp t}P_\pm}\nn\\
    &&[K,P_\pm]=\mp P_\pm\ .
\end{eqnarray}
In this Anosov system, the growing and decaying modes commute.

\paragraph{AdS$_2$ Rindler observer:} AdS$_2$ written in global coordinates is 
\begin{eqnarray}
    ds^2=\frac{-dT^2+d\sigma^2}{\sin^2\sigma}, \qquad \sigma\in [0,\pi]\ .
\end{eqnarray}
where the isometries are time translations generated by $H$, radial translations generated by $P$ and boosts generated by $B$
\begin{eqnarray}
    H=\p_T, \qquad P=-(\sin T \cos\sigma \p_T+\cos T\sin\sigma \p_\sigma)\nn\\
    B=(-\cos T\cos\sigma\p_T+\sin T\sin \sigma \p_\sigma)\ .
\end{eqnarray}
They satisfy the algebra
\begin{eqnarray}
    [H,P]=B, \qquad [B,H]=P, \qquad [B,P]=H\ .
\end{eqnarray}
The time evolution of an AdS Rindler observer is generated by $B$, the AdS$_2$ analog of boost, and the growing and decaying modes are $G_\pm=H\pm P$
\begin{eqnarray}
    &&[G_\pm,B]=\mp G_\pm\nn\\
    &&[G_+,G_-]=2 B\ .
\end{eqnarray}
This is the $sl(2,\mathbb{R})$ Lie algebra. We have an Anosov system with the growing and decaying modes corresponding to $G_\pm$.

\paragraph{Local Poincar\'e group in AdS$_2$:} In the example of the AdS$_2$ Rinlder observer, the growing and the decaying modes do not commute. However, if the observer is highly boosted (localized near the bifurcation surface), it perceives the local geometry to be Rindler space, with almost commuting generators for the growing and the decaying modes. 

Near the bifurcation surface we have $T=t\ep$ and $\sigma=\pi/2-r\ep$ for some small $\ep$.  The generators of the isometries can be expanded as
\begin{eqnarray}
&&P= -\p_\sigma+O(\ep^2)\nn\\
&&B= (-r\p_t+t\p_r)+O(\ep^2)\nn\\
&&H=-\ep^{-1}\p_t\ . 
\end{eqnarray}
In the near horizon limit, the boost $B$ and the generators $G_\pm =\ep H\pm P$ satisfy an approximate two-dimensional Poincar\'e algebra.

\subsection{Geodesic and horocycle flows}\label{app:horocycle}


Consider the upper half-plane $\mathbb{H}$ and the hyperbolic metric
\begin{eqnarray}
    ds^2=\frac{du d\bar{u}}{(\Im u)^2}
\end{eqnarray}
and $d(u_1,u_2)$ the hyperbolic distance between a pair of points $u_1,u_2\in \mathbb{H}$.
Any differentiable path between $u_1$ and $u_2$ can be parameterized 
 as $f:[0,d(u_1,u_2)]\to \mathbb{H}$ so that the generator of the flow along the path $(f(t),f'(t))$ is in the tangent bundle of $\mathbb{H}$. 
The length of the path is given by 
\begin{eqnarray}
    L(f)\int_0^1 \frac{|f'(t)|}{\Im(f)}dt\ .
\end{eqnarray}
Geodesics minimize the length between two points. 
All geodesics of $\mathbb{H}$ are vertical half-lines $\Re(u)=y_0=\text{const}$ or upper half-circles orthogonal to the real line. 

Every differentiable path between two points can be identified with its generating vector field. Vector fields in $\mathbb{H}$ are elements of its tangent bundle $T\mathbb{H}\simeq \mathbb{H}\times \mathbb{C}$. The natural inner product in the tangent plane is The isometry group of the upper half-plane is
\begin{eqnarray}\label{innerprodtangent}
    \braket{\vec{v}|\vec{w}}_{u\in H}=\frac{1}{\Im u}\braket{\vec{v}|\vec{w}}_{\mathbb{R}^2}\ .
\end{eqnarray} $PSL(2,\mathbb{R})$ that acts on $u$ as $g(u)=\frac{a u+b}{cu+d}$ with $g=\begin{pmatrix}
a& b\\
c& d
\end{pmatrix}$ with $ad-bc=1$ and $g\sim g\cdot (-I)$. Its action can be extended to the tangent bundle as 
\begin{eqnarray}
    &&g(z,\vec{v})=(g(z),g'(z)\vec{v})\nn\\
    &&g'(z)=\frac{1}{(cz+d)^2}\ .
\end{eqnarray}
The inner product in (\ref{innerprodtangent}) has the property that under the action of $PSL(2,\mathbb{R})$ on the tangent bundle the norm of the vectors remains unchanged
\begin{eqnarray}
\braket{\vec{v}|\vec{v}}_z=\braket{g'(z)\vec{v}|g'(z)\vec{v}}_{g(z)}\ .
\end{eqnarray}
The action of $PSL(2,\mathbb{R})$ sends a unit norm vector at $z$ to another unit norm vector at $g(z)$. Furthermore, for any two points on the tangent bundle $(z_1,\vec{v}_1)$ and $(z_2,\vec{v}_2)$ with unit norm vectors there is a unique $g\in PSL(2,\mathbb{R})$ such that $z_2=g(z_1)$ and $\vec{v}_2=g'(z)\vec{v}_1$. Therefore, there is a one-to-one map between $PSL(2,\mathbb{R})$ and the unit norm tangent bundle  $T^1\mathbb{H}$, and one can identify them.\footnote{In other words, the action of $PSL(2,\mathbb{R})$ on $T^1\mathbb{H}$ is simply transitive.} For example, we can identify the identity operator in $PSL(2,\mathbb{R})$ with the element $(i,\vec{i}\,)\in T^1\mathbb{H}$, and $g\in PSL(2,\mathbb{R})$ with $(g(i),g'(\vec{i}\,))$. 

Consider the points $u_1=i$ and $u_2=e i$. The path $f(t)=e^t i$ is the unique geodesic of unit speed that passes between them. Therefore,  $(f(t),f'(t))\in T^1\mathbb{H}$. It is clear that $PSL(2,\mathbb{R})$ sends one geodesic to another. 

\begin{lemma}
Given two points $u_1,u_2\in \mathbb{H}$ there is a unique geodesic path $f:[0,d(u_1,u_2)]\to \mathbb{H}$ with unit speed between them. 
There is also a unique element of $PSL(2,\mathbb{Z})$ such that $f(t)=g(e^t i)$. 
\end{lemma}
\begin{proof}
    To show this first, we use $g\in PSL(2,\mathbb{R})$ to map $g(u_1)=i$ and $g(u_2)=e i$. Then, the element of $PSL(2,\mathbb{R})$ that is $g(e^ti)$ generates the unique geodesic flow of unit speed between $u_1$ and $u_2$. 
\end{proof}

There is a one-to-one correspondence between the points in $(u,\vec{v})\in T^1\mathbb{H}$ and the unit speed geodesics that go through $u$ in the direction of $\vec{v}$. Since moving along a geodesic $(u(t),\vec{v}(t))$ preserves the points in $T^1\mathbb{H}$ it corresponds to a flow on $PSL(2,\mathbb{R})$ that is multiplication on the right by
\begin{eqnarray}
    g\to g\cdot \begin{pmatrix}
        e^{-t/2}&0\\
        0& e^{t/2}
    \end{pmatrix}
\end{eqnarray}
At the point $(i,\vec{i}\,)\in T^1\mathbb{H}$, in the direction tangent to the geodesic flow, there are two isometric flows called the 
{\it horocycle flows}. Since they move points on $T^1\mathbb{H}$ they also correspond to flows on $PSL(2,\mathbb{R})$.
The stable horocycle corresponds to $(i,\vec{i}\,)\to (i+t,\vec{i}\,)$ that corresponds to 
\begin{eqnarray}
   g\to g\cdot  \begin{pmatrix}
        1 && -s\nn\\
        0 && 1
    \end{pmatrix}
\end{eqnarray}
whereas the unstable one corresponds to reversing the direction of the tangent $(u,\vec{v})\to (u,-\vec{v})$. In $PSL(2,\mathbb{R})$ this action corresponds to $w(u)=-1/u$ or $W=\begin{pmatrix}
    0 &&1\\
    -1 && 0
\end{pmatrix}\ .$
As a result, the unstable horocycle corresponds to the flow generated by 
\begin{eqnarray}
    g\to g\cdot \begin{pmatrix}
        1 && 0\\
        s && 1
    \end{pmatrix}\ .
\end{eqnarray}
Consider the action $g\to g\cdot T$. From the eigenvalue system
\begin{eqnarray}
    g\cdot (T-\lambda 1)=0
\end{eqnarray}
it follows that 
\begin{eqnarray}
&&\lambda^2-\text{tr}(T)\lambda +1=0\nn\\
&&\lambda_\pm=\frac{1}{2}(\text{tr}T\pm \sqrt{\text{Tr}(T)^2-4})\ .
\end{eqnarray}
The elements of $PSL(2,\mathbb{R})$ split into three groups
\begin{enumerate}
\item {\it Elliptic:} When $|\text{tr}(T)|\leq 2$. In this case, the two eigenvalues are complex conjugates of each other $\lambda_+=\lambda_-^*$ and have unit norm, i.e. $|\lambda_\pm|=1$. 

\item {\it Parabolic:} When $\text{tr}(T)=2$ then the eigenvalues are $\lambda=\pm 1$. For example, the elements $\gamma^\pm$ are elliptic $\text{tr}(\gamma^\pm)=2$.
    \item {\it Hyperbolic:} When $\text{tr}(T)\geq 2$. As an example, the matrix $\gamma=\begin{pmatrix}
    e^{-t/2}& 0\\
    0 && e^{t/2}
\end{pmatrix}$
is elliptical because $\text{Tr}(\gamma)=2\cosh(t/2)\geq 2$. The action of these elements on the torus forms Anosov flows. 
\end{enumerate}

The geodesic flow in the upper half-plane corresponds to $u\to e^t u$, and the following flow on the unit tangent bundle $g\to T(g)=g\cdot \gamma(t)$ on $PSL(2,\mathbb{R})$:
\begin{eqnarray}
\gamma(t)=\begin{pmatrix} 
e^{-t/2} && 0\\
0 && e^{t/2}
\end{pmatrix}\ .
\end{eqnarray}
The translation $u\to u+s$ corresponds to the stable horocycle flow $g\to S^+(g)= g\cdot \gamma^+(s)$
\begin{eqnarray}
\gamma^+(s)=\begin{pmatrix} 
1 && s\\
0 && 1
\end{pmatrix}
\end{eqnarray}
whereas the unstable horocycle flow $g\to S^-(g)=g\cdot \gamma^-(s)$ is generated by
\begin{eqnarray}
    \gamma^-(s)=\begin{pmatrix} 
1 && 0\\
s && 1
\end{pmatrix}\ .
\end{eqnarray}
They satisfy the algebraic relations 
\begin{eqnarray}
    T(t) S^\pm (s)T(-t)=S^{\pm}(e^{\mp t} s)\ .
\end{eqnarray}
Next, we lift these actions from the phase space $T^1\mathbb{H}$ to the classical observable algebra $L^\infty(T^1\mathbb{H})$:
\begin{eqnarray}
    &&\tau(f(u))=f(T(-u))\nn\\
    &&\sigma^{\pm}(f(u))=f(S^\pm(-u))
\end{eqnarray}
so that
\begin{eqnarray}
    \tau(t)\circ \sigma^\pm(s)\circ\tau(-t)=\sigma^\pm(e^{\mp t}s)\ .
\end{eqnarray}
 This is a classical Anosov system.

\section{Quantum Dynamical Systems}\label{app:quantum_dynamics}

In this section, we briefly discuss the extension of ergodic theory to quantum systems. Every Abelian von Neumann algebra is isomorphic to the algebra $L^\infty(X,\mu)$ for some space $X$ and measure $\mu$. Therefore, the observable algebra of any classical system is always the algebra of functions on some geometry (``phase space"). This naturally connects geometry with classical dynamical systems. This connection does not generalize to non-Abelian algebras, and we cannot always associate geometry with a quantum system. Non-commutative geometry is a mathematical program that attempts at generalizing geometry and associate a non-commutative geometry to a large class of von Neumann algebras \cite{connes1994noncommutative}.

The space $R$ in our classical example of glass is a geometry with a Haar measure (volume element), and the observable algebra is the algebra of complex bounded measurable functions on $R$. Since the space $R$ is measurable, it comes with a $\sigma$-algebra of measurable subsets on $R$. In classical ergodic theory, we can think of measurable functions on $X$ (observables) in terms of the characteristic functions on subsets $X\in R$ inside the $\sigma$-algebra. In the quantum case, there is no geometry analogous to $R$, and the Abelian algebra of functions is replaced by the observable algebra of a quantum system $\mR$ represented in the Hilbert space $\mH$. Quantum dynamics is described by an automorphism $\tau_g:\mR\to \mR$ with $g$ in some group $G$ (or monoid if the dynamics is not reversible). Similar to the classical case, we often have a state $\omega:\mR\to \mathbb{C}$ (non-commutative measure) that is preserved under the action of $\tau_g$ i.e., $\omega\circ \tau_g=\omega$.

If the dynamical system preserves some state $\omega$ it is natural to represent the algebra in the GNS Hilbert space $\mH_\omega$ with the identity operator corresponding to the vacuum vector $\ket{\Omega}\in \mH_\omega$. 
In most of the examples we consider here, the dynamical transformation group is $\mathbb{R}$. Time evolution corresponds to a strongly continuous unitary flow realized in the Hilbert space by $U(t)=e^{iHt }$ with $t\in \mathbb{R}$; i.e. for all observables $a\in \mR$ we have $\tau_t(a)=U^\dagger(t) aU(t)$.\footnote{In these cases, there is a Haar measure on the group (the Lebesgue measure $dt$).} 
Note that every non-Abelian von Neumann algebra (quantum system) is a quantum dynamical system with evolution given by modular flow. In fact, modular flow $\Delta_\Omega^{it}$ is the unique automorphism of the algebra that preserves a state $\omega$ and satisfies the KMS relation \cite{stratila2020modular}. Therefore, any quantum dynamical system that satisfies the KMS relation describes {\it modular dynamics}. This is the intuition that underlies the connections between modular dynamics and dynamical systems with a positive generator in Theorem \ref{thm:HST}.

More generally, we define a quantum dynamical system as a triple $(\mA,S,\tau)$ where $\mA$ is a C$^*$-algebra, $S$ is a locally-compact monoid that is  represented as endomorphisms $\tau_t:\mA\to \mA$ with $t\in S$ such that the three conditions
\begin{enumerate}
    \item {\bf Identity action:} If $e\in S$ is the identity element then $\tau_e(a)=a$ for all $a\in \mA$.
    \item {\bf Homomorphism:} Monoid multiplication is preserved $\alpha_{g_1}\circ \alpha_{g_2}=\alpha_{g_1g_2}$.
    \item {\bf Continuity:} The map $g\to \alpha_g$ is continuous in operator norm topology.
\end{enumerate}

If the algebra $\mA$ is a von Neumann algebra with a state $\omega$ we can represent it on $B(\mH_\omega)$. The automorphisms of $\mA$ can be represented as a strongly continuous unitary flow on $\mH_\omega$: $\alpha_t(a)=U(t) a U(t)^\dagger$. We obtain a W$^*$-dynamical system.

Every automorphism of a von Neumann algebra $\tau_s:\mR\to \mR$ can be implemented as a unitary flow on the GNS Hilbert space:
\begin{eqnarray}
    \tau_s(a)\ket{\Omega}=U(s) a \ket{\Omega}\ .
\end{eqnarray}
This unitary representation is unique if we further require that
\begin{eqnarray}
    J_\mR U(s)=U(s) J_\mR\ .
\end{eqnarray}
This means that it preserves the natural cone because
\begin{eqnarray}
    U(t) a Ja \ket{\Omega}=a_t U(t) J a\ket{\Omega}=a_t J a_t\ket{\Omega}
\end{eqnarray}
which is equivalent to 
\begin{eqnarray}
    \Delta^{1/2}U(t)=U(t) \Delta^{1/2}\ .
\end{eqnarray}
Intuitively, we can interpret the above equation as the invertibility of automorphisms with respect to an alternate inner product \cite{furuya2022real}
\begin{eqnarray}
    (a,\Phi_t(b)):=\braket{a|\Delta^{1/2}\Phi_t(b)}=\braket{ \Phi_{-t}a|\Delta^{1/2}b}\ .
\end{eqnarray}
With respect to this inner product, we have $\Phi_t^\dagger=\Phi_{-t}$.

\subsection{Examples of quantum Anosov systems}

The motivation to consider quantum Anosov systems is to have examples of systems that cluster at least exponentially fast in modular time. Note that they might cluster faster, but certainly not slower. Theorem 3.6 of \cite{emch1994anosov} ensures that, in quantum Anosov systems, clustering occurs at least exponentially fast. 

\paragraph{Discrete translations on non-commutative torus (Quantum Arnold's Cat map):}
We start by considering a linear flow on a non-commutative torus. A non-commutative torus is a subalgebra of $B(L^2(\mathbb{R}/\mathbb{Z}))$ that is generated by the following two unitaries
\begin{eqnarray}
    &&(U f)(z)=z f(z)\nn\\
    && (V f)(z)=f(z e^{-2\pi i \theta})
\end{eqnarray}
for some deformation parameter $\theta$. It follows that
\begin{eqnarray}
    UV U^\dagger =e^{-2\pi i \theta} V\ .
\end{eqnarray}
This algebra is, the Abelian algebra of continuous functions on $S^1$ enlarged by cross product with rotation by irrational angle $\theta$: $e^{2\pi i \theta \p_z}$. This resulting cross-product is a noncommutative torus. 

Consider the parameter space of the Weyl algebra of a quantum particle. It is parameterized by a complex number $z$ in Weyl unitary operators 
\begin{eqnarray}
    &&W(z)=e^{z a^\dagger-z^* a}
\end{eqnarray}
satisfying the {\it twisted} algebra 
\begin{eqnarray}
    &&W(z_1)W(z_2)=e^{2\pi i \theta\Im(z_1\bar{z}_2)} W(z_1+z_2)\ .
\end{eqnarray}
We compactify this space by identifying $z\sim z+i\sim z+1$ to obtain a torus. Now, the fundamental domain is $\mathbb{Z}^2$. The element $\gamma=\begin{pmatrix}
    e^{-t/2}&& 0\\
    0&& e^{t/2}
\end{pmatrix}$
acts on the torus as an Anosov quantum system with discrete time evolution.

\paragraph{Free quantum particle on a manifold $\mM$:}
As opposed to classical systems, quantum systems (von Neumann algebras) are dynamical systems with the flow given by the modular flow. In fact, every classical dynamical system can be viewed as a ``quantum" (noncommutative) system, as well. The idea is simply to enlarge the set of observables by including dynamical transformations in them. Taking the closure of the resulting algebra is a non-commutative algebra. In mathematics, this is called the group measure space construction. We postpone the generalities of this construction to section \ref{app:crossprod}. Here, we start with an example to build physics intuition.

Consider a free non-relativistic quantum particle on a line. The classical observables form the Abelian algebra of bounded complex functions on the real line, i.e. $L^\infty(\mathbb{R},dx)$, and choose translation $(T_t f)(x)=f(x+t)$ as the dynamics. The cross product of this Abelian algebra with the automorphism is $L^\infty(\mathbb{R},dx)\rtimes_{T_t} \mathbb{R}$ is isomorphic to the algebra of a quantum particle on a line $B(L^2(\mathbb{R},dx))$ i.e., it is type I$_\infty$. Note that there is no invariant state.\footnote{The Lebesgue measure is invariant but it is not normalizable; hence it is a weight and not a state.} The translation operator belongs to this algebra and generates an Abelian centralizer for this type I$_\infty$ algebra. The spectrum of translation is continuous and has no normalizable eigenoperators. Therefore, this algebra is ergodic, but it is not strongly mixing because functions of momentum do not decay.

\paragraph{Free quantum particle on upper half-plane}

In generalizing classical Anosov systems to the quantum world  using the crossed product construction an important subtlety presents itself. As we saw in Lemma \ref{Anosoventirespec}, we obtain the quantum Anosov relation only if the spectrum of the generator of the flow is complete. In the examples we considered, the dynamical flow was either generated by translations, boost, or dilatation. In Rindler space, the spectrum of boost and null translations as self-adjoint operators in the Hilbert space are entire, therefore the examples of free non-relativistic and a free particle on a line, and free relativistic particles on $\mathbb{R}^{1,1}$ become quantum Anosov systems. However, in hyperbolic space (the upper half-plane) the spectrum of the Laplace-Beltrami operator is discrete. Therefore, a free non-relativistic quantum particle on hyperbolic space is not a quantum Anosov system. For more details see \cite{emch1994anosov}.

For examples of quantum K-systems associated with quantum many body systems on a lattice see \cite{narnhofer2006dynamical}. For more examples of quantum Anosov systems see \cite{savvidy2020maximally}.

\subsection{First quantization and cross-products}\label{app:crossprod}

Consider the classical dynamical system $(X,d\mu,\tau)$ where $X$ is a measure space with measure $d\mu$ and $\tau:G\to \text{Aut}(X)$ is the dynamical transformation. The observables of this classical system are $\mK=L^\infty(X,d\mu)$, and they act on the Hilbert space $L^2(X,d\mu)$. Consider the set of $G$-valued vectors in $L^2(X,d\mu)$: $\mH=\ket{\Psi;g}\in L^2(G,\mK)$ so that the observable $a$ acts as 
\begin{eqnarray}
    \tau_g(a)\ket{\Psi;g}=\ket{a\Psi;g},
\end{eqnarray}
and we have added the dynamical transformation $u_g$ to the observables
\begin{eqnarray}
    u_h\ket{\Psi;g}=\ket{\Psi; hg}\ .
\end{eqnarray}
We define the crossed product von Neumann algebra
\begin{eqnarray}
    \mM=\mK\rtimes_\tau G
\end{eqnarray}
as the double-commutant of the resulting operators, and we have
\begin{eqnarray}
    u_g f u_g^\dagger= \tau_g(f)\ .
\end{eqnarray}
We can revisit the example of the classical K-system of translations on $\mathbb{R}$. The action is free but not ergodic because the constant functions on $\mathbb{R}$ are in $L^\infty(X,d\mu)$. In this case, the crossed product algebra is generated by position and momentum; therefore, the observables are the type I$_\infty$ factor $B(L^2(X,d\mu))$; i.e. a free quantum particle on a line. The dilatations $f(x)\to f(e^t x)$ are inner automorphisms, and there exists no state that is invariant under them. The trace is simply integration over $x$.\footnote{If we restrict the translation group to translations by rational amount, then we obtain a type II$_\infty$ factor.} The resulting Hilbert space is a one-particle Hilbert space.
In physics, quite often, the one-particle observables are type I. 

More generally, every group $G$ acts on itself transitively by left or right translation. If $G$ is a finite group of order $n$, then $\mM$ is a type I$_n$ factor, and if $G$ is infinite then $\mM$ is type I$_\infty$. In defining ergodicity, we focus on the dynamical group $\mathbb{R}$, and the time-averages were defined with respect to the Haar measure on $\mathbb{R}$.
However, for the group generated by translations and dilatations $x\to a x+b$ the left-Haar measure is $\frac{da\wedge db}{a^2}$, whereas the right-Haar measures is $\frac{da\wedge db}{|a|}$.\footnote{A group with matching left and right Haar measures is called unimodular. The group $x\to ax+b$ is not unimodular.} Take a left-invariant measure $\mu(h)$ and translate it on the right by $g$, the result is another left-invariant measure 
\begin{eqnarray}
  \forall h\in G:  \mu(g^{-1}h)=\Delta(g)\mu(h),
\end{eqnarray}
where the function $\Delta(g)$ is called the {\it modular function} of the group.\footnote{By definition, for unimodular groups, the modular function is always one.}
Consider the above action restricted to the subalgebra of functions on $\mathbb{R}^+$:
\begin{eqnarray}
    \sigma_t f(x)=f(x-t)\Theta(x)\Theta(x-t)
\end{eqnarray}
where $\Theta$ is the Heaviside function. We have
\begin{eqnarray}
    \sigma_t \circ \sigma_{-t}=\begin{cases}
    \text{id}\qquad &t\geq 0\nn\\
    \Theta(t) \qquad &t<0
    \end{cases}
\end{eqnarray}
where $\Theta(t)$ can be viewed as a conditional expectation.

We say the group action of $G$ on $X$ is {\it free} if every non-identity element of $G$ moves every element of $X$. We say it is {\it transitive} if the orbit of any point gives the whole $X$. We say it is {\it ergodic} if it acts irreducibly. More formally, we write
\begin{definition}
    The action of $G$ on $X$ is free if $gx=x$ for some $x\in X$ implies that $g=1$.\footnote{Said differently, the map $x\to g x$ is one-to-one.} It is transitive if for any two points $x,y\in X$ there exists a $g\in G$ such that $g x=y$. 
The action of $G$ on $X$ is ergodic if the only subsets that are preserved under it are $\emptyset$ and $X$.
\end{definition}
von Neumann and Murray proved the following result that characterizes the von Neumann algebras that can be constructed from classical dynamical systems using crossed products: 
\begin{theorem}
    If the action of $G$ on $L^\infty(X,d\mu)$ is free then the crossed product $\mM$ is a factor if and only if the action of $G$ is ergodic.
\end{theorem}
\begin{theorem}
 Consider $\mM=L^\infty(X,d\mu)\rtimes_\tau G$ where the action of $G$ on $\mA$ is free and ergodic. Then
    \begin{itemize}
        \item $\mM$ is type I$_\infty$ factor if and only if the action of $G$ is transitive.
        \item $\mM$ is type II$_1$ factor if and only if $G$ is a discrete infinite group and there exists a finite invariant measure $d\nu$ that is absolutely continuous with respect to $d\mu$.
        \item $\mM$ is type II$_\infty$ factor if and only if the action is not transitive, and there exists a left-invariant measure $d\nu$ that vanishes on the same set as $d\mu$.
        \item $\mM$ is type III factor if and only if there exists no left-invariant measure that vanishes on the same set as $d\mu$.
    \end{itemize}
\end{theorem}
\begin{proof}
    For proof of both results, see chapter XIII of \cite{takesaki2003theory}.
\end{proof}
To clarify the interpretation of the above result in terms of dynamical systems, it is helpful to discuss some examples. Consider irrational rotations on a circle: $T_\theta e^{i \phi}=e^{2\pi i \phi+\theta}$ with $\theta$ irrational. The group is $\mathbb{Z}$, it acts freely but not transitively. The Lebesgue measure $d\phi$ is invariant under these rotations. The Fourier transform of any observable is an eigenfunction of translations:
\begin{eqnarray}
    (T_\theta\circ f)_n=e^{2\pi i n \theta}f_n\ .
\end{eqnarray}
The only invariant observables are constant (proportional to identity), therefore the action is ergodic. This algebra is a type II$_1$ factor.

Let $G=\mathbb{R}$ and $X=T^2$ (a torus: $(0,1]\times (0,1]$). Consider the translation $T_s(x+i y)= (x+s+ y+i \theta s)$ where $\theta$ is irrational. This action is ergodic. The invariant measure is finite, but $G$ is not discrete, and the algebra is type II$_\infty$. 

Finally, consider the example $X=\mathbb{R}$ and $G$ the group $x\to a x+b$. This action is free and the Lebesgue measure $dx$ is ergodic with respect to the subgroup of translations. There exists no nontrivial invariant measure that vanishes on the same set as $d\mu$. Therefore, this algebra is type III.

\subsection{K-systems, entropy and strong asymptotic Abelianness}
\label{app:Ksystem}

For a classical K-system, there are several equivalent conditions to define a K system. One may define the K-system through the half-sided inclusions, mixing property\cite{Cornfeld1982ergodic}. It is interchangeable to use K-system and K-mixing in the classical case. A K-system is known to have positive Kolmogorov-Sinai entropy (KS entropy). A positive KS entropy is often taken to be an indicator of classical chaos \cite{FRIGG200626}. For completeness, we review the definition of the KS entropy here: 
\begin{definition}[Partition]
   A partition of a measure space $(X, \Sigma, \mu)$ is defined as a collection $\alpha$ of nonempty disjoint measurable subsets such that the union of all subsets is $X$, i.e.
   \begin{align}
       \forall A, B \in \alpha, A \neq B \implies A\cap B = \emptyset,
   \end{align}
   \begin{align}       
       \cup_{A\in \alpha} A = X.
   \end{align}
\end{definition}

\begin{definition}[Entropy of partition]
   The entropy $H(\alpha)$ of a finite partition $\alpha = \{A^{(1)}, \ldots,A^{(r)}\}$ is defined as
   \begin{align}
       H(\alpha) = -\sum^{r}_{i=1} \mu(A^{(i)})\log(\mu(A^{(i)})).
   \end{align}
\end{definition}

\begin{definition} [KS entropy with automorphism]
   The KS entropy $h(T)$ of a given automorphism $T$ is then defined as follow.
   \begin{align}
       h(T) = \sup_{\alpha} \lim_{n \to \infty} \frac{1}{n}H(\alpha^{n-1}_0),
   \end{align}
   where the supremum is taken over all finite partitions $\alpha$, and
   \begin{align}
       \alpha^{n-1}_0 = \alpha \vee T^1\alpha\vee\ldots\vee T^{n-1}\alpha.
   \end{align}
\end{definition}
The KS entropy for a flow $\{T_t\}$ is naturally defined as $h(\{T_t\})=h(T_1)$ because $h(T_t) = |t| h(T_1)$.  It was proven in \cite{Cornfeld1982ergodic} that a system is a K-system if and only if $\lim_{n \to \infty} \frac{1}{n}H(\alpha^{n-1}_0)$ is positive for all non-trivial finite partitions $\alpha$, thus one of the alternative ways to define a K-system. From this result, it is clear that K-systems have positive KS entropy, but the inverse may not be true.

The equivalent definitions of classical K-system do not generalize to the quantum (non-Abelian) systems. In \cite{emch1976generalized}, the authors defined the quantum K-system through algebraic relations as in Definition \ref{qksystem}. This is the definition that we use in this paper. Another definition of quantum K-system was introduced in \cite{narnhofer1989quantum} that is based on the properties of entropy. It was proven that with additional assumptions of asymptotic Abelianness and KMS property, the algebraic definition implies the entropic definition \cite{narnhofer2006dynamical}. However, it is unclear whether these assumptions are necessary.

There is an intuitive connection between K-systems and strong asymptotic Abelianness. In type II$_1$ algebras, we have strong asymptotic Abelianness if and only if we have past algebras (a quantum K-system) \cite{golodets1998non}. 
It was also mentioned by Narnhofer \cite{narnhofer2006dynamical} that there is an example, a Fermi lattice system with shift automorphism, that is a quantum K-system but not strongly asymptotically Abelian. For a general type II or III this remains an open problem.  

As we discussed in section \ref{sec:algebratypes}, a strong mixing system with the KMS property has weak asymptotic Abelianness.
Strong mixing requires the two-point function to vanish in the limit of infinite time separation, but the assumption of strong asymptotic Abelianness implies that the commutator should also vanish. This property is stable under multiplication: 
\begin{eqnarray}
    \lim_{t\to \infty}[a_1(-t)a_2(-t),b]=\lim_{t\to \infty}(a_1(-t)[a_2(-t),b]+[a_1(-t),b]a_2(-t))=0\ .
\end{eqnarray}
To construct the algebra of the past we need to take limits $\lim_n a_n(-t)$. Now, if the limits $n$ and $t$ commute strong asymptotic Abelianness implies that the whole algebra generated by operators in the past commute with $b$. We conclude that this is a proper subalgebra of $\mR$ because otherwise, $b$ has to be in the center of $\mR$ but it is a factor. 

Finally, we point out the following algebraic implication of strong asymptotic Abelianness:
\begin{lemma}
    If we have a state-preserving automorphism of a von Neumann algebra that satisfies the strong asymptotic Abelianness in a state $\omega$, the spectrum of $\Delta_\omega$ is the same as the Arveson spectrum.
\end{lemma}
\begin{proof}
   For a proof see \cite{summers2005tomita}.
\end{proof}

\bibliographystyle{JHEP}
\bibliography{main}
\end{document}